\documentclass[sigconf, nonacm]{acmart}

\usepackage[scaled=0.78]{beramono}
\usepackage[ruled,vlined,linesnumbered]{algorithm2e}
\usepackage{graphicx}
\usepackage{caption}
\usepackage{subcaption}
\usepackage{amsfonts}
\usepackage{bm}
\usepackage{pifont}
\usepackage{booktabs,makecell,multirow}
\usepackage{siunitx}
\usepackage{multirow}
\usepackage{xspace}
\usepackage[capitalise]{cleveref}
\usepackage{tikz}
\usepackage{comment}
\usepackage{breqn} 
\usepackage{tabularx}
\usepackage{mathtools}
\usepackage{float}
\usepackage[shortlabels]{enumitem}
\setlist[enumerate]{leftmargin=1cm}
\setlist[itemize]{leftmargin=0.5cm}
\usepackage{breakurl}
\usepackage{listings}
\usepackage{pgfplots}
\usetikzlibrary{patterns}
\usetikzlibrary{calc,positioning,fit}
\usepgfplotslibrary{groupplots,units}
\usepackage{lipsum}
\usepackage{symbols}
\usepackage{thmtools}
\usepackage{thm-restate}
\usepackage{xpatch}
\usepackage{svg}
\usepackage{adjustbox}

\setlist[enumerate]{nosep}
\setlist[enumerate]{nosep}
\setlist{nolistsep,leftmargin=20pt}
\setenumerate{label=\arabic*.}
\crefname{section}{§}{§§}
\Crefname{section}{§}{§§}
\crefname{figure}{Fig}{Fig}
\Crefname{figure}{Fig}{Fig}

\newcommand{\namebase}{Chipmink}
\newcommand{\name}{\namebase\xspace}
\newcommand{\system}{\textsf{\namebase}\xspace}
\newcommand{\datamodel}{\textsf{ObjectGraph}\xspace}
\newcommand{\podgraph}{\textsf{PodGraph}\xspace}

\newcommand{\podurl}{\url{https://github.com/illinoisdata/pod}}

\newcommand{\repoexpbranchurl}{\url{https://github.com/illinoisdata/pod/tree/supawit/con}}

\newcommand{\dill}{\textsf{Dill}\xspace}
\newcommand{\shev}{\textsf{Shelve}\xspace}
\newcommand{\zosp}{\textsf{ZODB}\xspace}
\newcommand{\zodb}{\textsf{ZODB-Hist}\xspace}
\newcommand{\criu}{\textsf{CRIU}\xspace}
\newcommand{\exhaust}{\textsf{Exhaustive}\xspace}

\newcommand{\snp}{\textsf{Snapshot}\xspace}
\newcommand{\snz}{\textsf{Snapshot-LZ4}\xspace}
\newcommand{\snx}{\textsf{Snapshot-Xdelta}\xspace}
\newcommand{\bundle}{\textsf{BundleAll}\xspace}
\newcommand{\pga}{\system}
\newcommand{\pgaz}{\system-LZ4\xspace}
\newcommand{\pgcachez}{\textsf{OnlyAVF}\xspace}

\newcommand{\pgnoavf}{\textsf{OnlyCD}\xspace}
\newcommand{\pgcacheznoavf}{\textsf{NoCD/AVF}\xspace}
\newcommand{\pnv}{\textsf{SplitAll}\xspace}
\newcommand{\prand}{\textsf{Random}\xspace}
\newcommand{\pfl}{\textsf{TbH}\xspace}
\newcommand{\pfa}{\textsf{TbH}\xspace}
\newcommand{\pgz}{\textsf{LGA-0}\xspace}
\newcommand{\pgi}{\textsf{LGA-1}\xspace}
\newcommand{\pglga}{\textsf{LGA}\xspace}

\newcommand{\pgnostatic}{\textsf{OnlyAVL}\xspace}
\newcommand{\pgnoavl}{\textsf{OnlyASCC}\xspace}

\newcommand{\pglnoavlstatic}{\textsf{Sync}\xspace}
\newcommand{\pnjpickle}{\textsf{Dill:Neo4j}\xspace}
\newcommand{\pnjdelta}{\textsf{SplitAll:Neo4j}\xspace}
\newcommand{\pganj}{\textsf{\namebase:Neo4j}\xspace}

\newcommand{\buildats}{\texttt{buildats}\xspace}
\newcommand{\storesfg}{\texttt{storesfg}\xspace}
\newcommand{\itsttime}{\texttt{itsttime}\xspace}

\newcommand{\skltweet}{\texttt{skltweet}\xspace}

\newcommand{\aicode}{\texttt{ai4code}\xspace}
\newcommand{\agripred}{\texttt{agripred}\xspace}
\newcommand{\msciedaw}{\texttt{msciedaw}\xspace}
\newcommand{\ecomsmph}{\texttt{ecomsmph}\xspace}
\newcommand{\netmnist}{\texttt{netmnist}\xspace}
\newcommand{\rlactcri}{\texttt{rlactcri}\xspace}
\newcommand{\vaenet}{\texttt{vaenet}\xspace}
\newcommand{\tseqpred}{\texttt{tseqpred}\xspace}
\newcommand{\wordlang}{\texttt{wordlang}\xspace}

\newcommand{\fix}[2][]{#2}

\newcommand{\ignore}[1]{{}}
\newcommand{\cut}[1]{{}}

\definecolor{Acolor}{HTML}{944003}
\definecolor{Bcolor}{HTML}{7678ED}
\definecolor{Ccolor}{HTML}{213F80}
\definecolor{Dcolor}{HTML}{8ECAE6}
\definecolor{Ecolor}{HTML}{2A9D8F}
\definecolor{Fcolor}{HTML}{E9C46A}
\definecolor{Gcolor}{HTML}{F4A261}
\definecolor{Hcolor}{HTML}{E76F51}

\definecolor{cAmain}{HTML}{1f77b4}  
\colorlet{cAlight}{cAmain!25}
\definecolor{cBmain}{HTML}{ff7f0e}  
\colorlet{cBlight}{cBmain!25}
\definecolor{cCmain}{HTML}{2ca02c}  
\colorlet{cClight}{cCmain!25}
\definecolor{cDmain}{HTML}{d62728}  
\colorlet{cDlight}{cDmain!25}
\definecolor{cEmain}{HTML}{9467bd}  
\colorlet{cElight}{cEmain!25}
\definecolor{cFmain}{HTML}{8c564b}  
\colorlet{cFlight}{cFmain!25}
\definecolor{cGmain}{HTML}{e377c2}  
\colorlet{cGlight}{cGmain!25}
\definecolor{cZmain}{HTML}{030303}  
\colorlet{cZlight}{cZmain!25}
\colorlet{cZlightlight}{cZmain!5}
\definecolor{cPositivemain}{HTML}{2ca02c}  
\colorlet{cPositivelight}{cPositivemain!25}
\definecolor{cNegativemain}{HTML}{d62728}  
\colorlet{cNegativelight}{cNegativemain!25}

\definecolor{dillcolor}{HTML}{1f77b4}
\colorlet{dillcolorlight}{dillcolor!25}
\definecolor{shevcolor}{HTML}{2ca02c}
\colorlet{shevcolorlight}{shevcolor!25}
\definecolor{zospcolor}{HTML}{ff7f0e}
\colorlet{zospcolorlight}{zospcolor!25}
\definecolor{zodbcolor}{HTML}{e377c2}
\colorlet{zodbcolorlight}{zodbcolor!25}
\definecolor{criucolor}{HTML}{9467bd}
\colorlet{criucolorlight}{criucolor!25}
\definecolor{exhaustcolor}{HTML}{030303}
\colorlet{exhaustcolorlight}{exhaustcolor!25}
\definecolor{pgacolor}{HTML}{d62728}
\colorlet{pgacolorlight}{pgacolor!25}

\def\undercaptionspace{0mm}
\def\overtopfigurespace{-3mm}
\def\beforeeqspace{-1mm}
\def\normallabelsize{\footnotesize}
\def\smalllabelsize{\scriptsize}

\tikzset{
    barzero/.style={%
        draw=cZmain, fill=white,
        pattern color=cZmain, postaction={pattern=grid},
    },
    >=LaTeX
}
\tikzset{
    barone/.style={%
        draw=cAmain, fill=white,
        pattern color=cAmain, postaction={pattern=crosshatch},
    },
    >=LaTeX
}
\tikzset{
    bartwo/.style={%
        draw=cCmain, fill=white,
        pattern color=cCmain, postaction={pattern=crosshatch dots},
    },
    >=LaTeX
}
\tikzset{
    barthree/.style={%
        draw=cBmain, fill=white,
        pattern color=cBmain, postaction={pattern=north west lines},
    },
    >=LaTeX
}
\tikzset{
    barfour/.style={%
        draw=cGmain, fill=white,
        pattern color=cGmain, postaction={pattern=north east lines},
    },
    >=LaTeX
}
\tikzset{
    barfive/.style={%
        draw=cFmain, fill=white,
        pattern color=cFmain, postaction={pattern=grid},
    },
    >=LaTeX
}
\tikzset{
    barsix/.style={%
        draw=cEmain, fill=white,
        pattern color=cEmain, postaction={pattern=horizontal lines},
    },
    >=LaTeX
}
\tikzset{
    barmain/.style={%
        draw=pgacolor, fill=pgacolorlight,
        pattern color=pgacolor, postaction={pattern=crosshatch},
    },
    >=LaTeX
}

\tikzset{
    bardill/.style={%
        barone
    },
    >=LaTeX
}
\tikzset{
    barsnz/.style={%
        bartwo
    },
    >=LaTeX
}
\tikzset{
    barsnx/.style={%
        barthree
    },
    >=LaTeX
}
\tikzset{
    barshev/.style={%
        bartwo
    },
    >=LaTeX
}
\tikzset{
    barzosp/.style={%
        barthree
    },
    >=LaTeX
}
\tikzset{
    barzodb/.style={%
        barfour
    },
    >=LaTeX
}
\tikzset{
    barcriu/.style={%
        barsix
    },
    >=LaTeX
}
\tikzset{
    barpga/.style={%
        barmain
    },
    >=LaTeX
}
\tikzset{
    barpgaz/.style={%
        barmain, fill=pgacolor
    },
    >=LaTeX
}

\newcommand\vldbdoi{XX.XX/XXX.XX}

\newcommand\vldbavailabilityurl{https://github.com/illinoisdata/pod}

\begin{document}

\title{\name: Efficient Delta Identification for Massive Object Graph}


\author[Supawit Chockchowwat, Sumay Thakurdesai, Zhaoheng Li, Matthew S. Krafczyk, Yongjoo Park]{Supawit Chockchowwat, Sumay Thakurdesai, Zhaoheng Li, Matthew S. Krafczyk, Yongjoo Park}
\affiliation{%
  \institution{University of Illinois Urbana-Champaign}
}
\email{{supawit2,sumayst2,zl20,mkrafcz2,yongjoo}@illinois.edu}







\begin{abstract}
Ranging from batch scripts to computational notebooks, modern data science tools rely on massive and evolving object graphs that represent structured data, models, plots, and more. Persisting these objects is critical, not only to enhance system robustness against unexpected failures but also to support continuous, non-linear data exploration via versioning. Existing object persistence mechanisms (e.g., Pickle, Dill) rely on complete snapshotting, often redundantly storing unchanged objects during execution and exploration, resulting in significant inefficiency in both time and storage. Unlike DBMSs, data science systems lack centralized buffer managers that track dirty objects. Worse, object states span various locations such as memory heaps, shared memory, GPUs, and remote machines, making dirty object identification fundamentally more challenging.

In this work, we propose a graph-based object store, named \system, that acts like the centralized buffer manager. Unlike static pages in DBMSs, persistence units in \system are dynamically induced by partitioning objects into appropriate subgroups (called \emph{\textbf{pods}}), minimizing expected persistence costs based on object sizes and reference structure. These pods effectively isolate dirty objects, enabling efficient partial persistence. Our experiments show that \system is general, supporting libraries that rely on shared memory, GPUs, and remote objects. Moreover, \system achieves up to 36.5$\times$ smaller storage sizes and 12.4$\times$ faster persistence than the best baselines in real-world notebooks and scripts.

\end{abstract}

\def\picwidth{0.18\linewidth}
\def\plotwidth{0.28\linewidth}

\pgfplotsset{
    compat=1.15,
    objcount/.style={
        solid,
        thick,
        draw=cAmain,
        fill=cAlight,
    },
}

\pgfplotstableread{figures/_data/numobjects.txt}{\tablenumobjects}

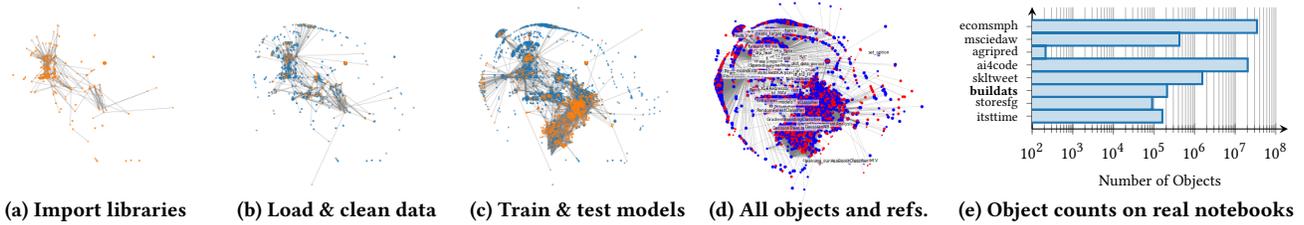
\begin{teaserfigure}
    \centering
    \centering
    \begin{subfigure}[b]{\picwidth}
        \adjustbox{trim=0mm 6mm 0mm 6mm}{
            \includegraphics[width=\linewidth]{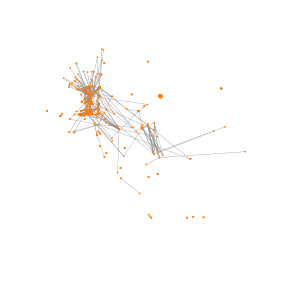}
        }
        \caption{Import libraries}
        \label{fig:intro_massive_cell1}
    \end{subfigure}%
    \begin{subfigure}[b]{\picwidth}
        \adjustbox{trim=0mm 6mm 0mm 6mm}{%
            \includegraphics[width=\linewidth]{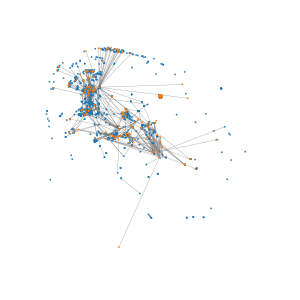}
        }
        \caption{Load \& clean data}
        \label{fig:intro_massive_cell2}
    \end{subfigure}%
    \begin{subfigure}[b]{\picwidth}
        \adjustbox{trim=0mm 6mm 0mm 6mm}{%
            \includegraphics[width=\linewidth]{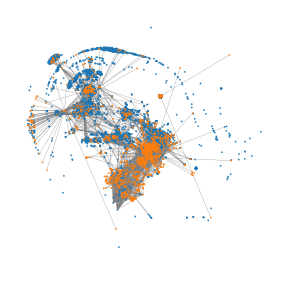}
        }
        \caption{Train \& test models}
        \label{fig:intro_massive_cell3}
    \end{subfigure}%
    \begin{subfigure}[b]{\picwidth}
        \adjustbox{trim=0mm 6mm 0mm 6mm}{%
            \includegraphics[width=\linewidth]{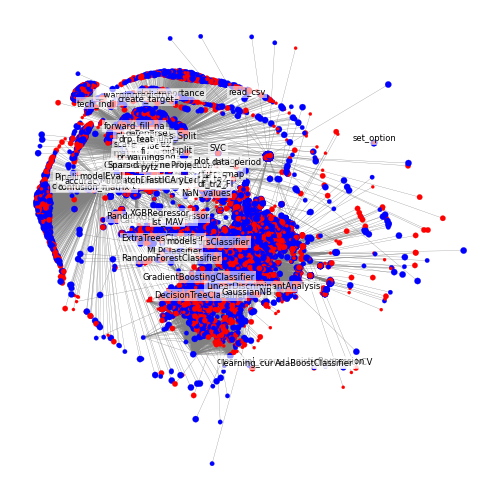}
        }
        \caption{All objects and refs.}
        \label{fig:intro_massive_all}
    \end{subfigure}%
    \begin{subfigure}[b]{\plotwidth}
        \centering
        \hspace{-3mm}
        \begin{tikzpicture}
        \begin{axis}[
            height=32mm,
            width=1.0\linewidth,
            xbar,
            bar width=1.0,
            axis y line=left,
            axis x line=bottom,
            xlabel=Number of Objects,
            xmode=log,
            xmajorgrids,
            xminorgrids,
            xmin=100,
            xmax=200000000,
            xtick={100,1000,10000,100000,1000000,10000000,100000000},
            ymin=-1,
            ymax=8.5,
            ytick=data,
            yticklabels from table={\tablenumobjects}{label},
            xlabel style={font=\scriptsize},
            every x tick label/.append style={font=\scriptsize},
            every y tick label/.append style={font=\tiny},
        ]
            \addplot[objcount] table [x=max,y=index] {\tablenumobjects};
        \end{axis}
        \end{tikzpicture}
        \vspace{-2mm}
        \caption{Object counts on real notebooks}
        \label{fig:intro_massive_numobjects}
    \end{subfigure}
    \vspace{-7mm}
    \caption{Underneath data systems lie a massive and evolving object graph. For example, \Cref{fig:intro_massive_cell1,fig:intro_massive_cell2,fig:intro_massive_cell3} illustrate three snapshots of object graphs in \buildats notebook with changed (\textcolor{cBmain}{orange}) and unchanged (\textcolor{cAmain}{blue}) objects. While \buildats consists of \fix{210k} objects supporting user-defined and library variables (\Cref{fig:intro_massive_all}), real notebooks may comprise tens of millions of objects (\Cref{fig:intro_massive_numobjects}).}
    \label{fig:intro_massive}
\end{teaserfigure}

\received{1 June 2025}
\received[revised]{1 October 2025}
\received[accepted]{15 November 2025}

\maketitle



\ifdefempty{\vldbavailabilityurl}{}{
\vspace{.3cm}
\begingroup\small\noindent\raggedright\textbf{PVLDB Artifact Availability:}\\
The source code, data, and/or other artifacts have been made available at \url{\vldbavailabilityurl}.
\endgroup
}


\pgfplotstableread{figures/_data/delta_count.txt}{\tabledeltacount}
\pgfplotstableread{figures/_data/exp_i.txt}{\tableexpi}
\pgfplotstableread{figures/_data//exp_i_script.txt}{\tableexpiscript}
\pgfplotstableread{figures/_data/exp_i_skltweet.txt}{\tableexpiskltweet}
\pgfplotstableread{figures/_data/exp_ii.txt}{\tableexpii}
\pgfplotstableread{figures/_data/exp_ii_cache.txt}{\tableexpiicache}
\pgfplotstableread{figures/_data/exp_iv.txt}{\tableexpiv}
\pgfplotstableread{figures/_data/exp_v.txt}{\tableexpv}
\pgfplotstableread{figures/_data/exp_vi.txt}{\tableexpvi}
\pgfplotstableread{figures/_data/exp_vii_mr.txt}{\tableexpviimr}
\pgfplotstableread{figures/_data/exp_c.txt}{\tableexpc}
\pgfplotstableread{figures/_data/exp_c_bd.txt}{\tableexpcbd}
\pgfplotstableread{figures/_data/exp_ci.txt}{\tableexpci}
\pgfplotstableread{figures/_data/exp_cii.txt}{\tableexpcii}
\pgfplotstableread[header=has colnames]{
index snp_storage_gb snp_load_s pga_storage_gb pga_load_s exec_s cexec_s snp_save_s snp_total_s pga_save_s pga_total_s nb cell
0 3.35e-07 0.0009870529174804688 2.065e-06 0.0021743178367614746 0.1531352996826172 0.1531352996826172 0.0006544589996337891 0.15378975868225098 0.0069959163665771484 0.16013121604919434 ai4code 0
1 1.974472474 43.23672133684158 1.972183878 51.127599596977234 47.306201219558716 47.45933651924133 80.59558463096619 127.9017858505249 114.68553972244263 161.99174094200134 ai4code 1
2 3.948944613 45.53468316793442 1.972184474 0.0015814900398254395 0.02712702751159668 47.48646354675293 80.64124870300293 80.66837573051453 0.009580850601196289 0.03670787811279297 ai4code 2
3 5.923416752 43.58277177810669 1.972184625 36.38897478580475 0.0006706714630126953 47.48713421821594 80.50858926773071 80.50925993919373 0.0008172988891601562 0.0014879703521728516 ai4code 3
4 7.993485937 55.629074454307556 2.086949873 51.15731918811798 1.0277631282806396 48.51489734649658 101.8916642665863 102.91942739486694 139.64471578598022 140.67247891426086 ai4code 4
5 10.063555122 55.91993290185928 2.086950617 46.054823875427246 0.0009210109710693359 48.51581835746765 102.37985420227051 102.38077521324158 0.0015506744384765625 0.0023708343505859375 ai4code 5
6 12.16726128 55.707202196121216 3.425605526 28.844497501850128 60.28975057601929 108.80556893348694 101.87223958969116 162.16199016571045 16.140039443969727 76.42979001998901 ai4code 6
7 14.297305625 56.4598702788353 3.451948742 0.05166393518447876 0.14554309844970703 108.95111203193665 102.1934506893158 102.3389937877655 61.88349795341492 62.029041051864624 ai4code 7
8 16.404086688 56.02906531095505 4.75861396 28.507317900657654 0.8531973361968994 109.80430936813354 102.05878448486328 102.91198182106018 62.69690656661987 63.55010390281677 ai4code 8
9 18.510867751 56.107735216617584 4.758614828 41.64581245183945 0.0009236335754394531 109.80523300170898 101.79775190353394 101.79867553710938 0.0010690689086914062 0.0018918514251708984 ai4code 9
10 20.634977385 56.19158488512039 5.388149007 21.5903303027153 31.514989376068115 141.3202223777771 101.83859729766846 133.35358667373657 8.09297251701355 39.607961893081665 ai4code 10
11 22.769918711 56.18927747011185 5.402060676 0.04163479804992676 0.08250999450683594 141.40273237228394 101.68121218681335 101.76372218132019 46.64308261871338 46.725592613220215 ai4code 11
12 24.894043067 55.8440899848938 5.475443704 21.36371660232544 1.045964241027832 142.44869661331177 102.121102809906 103.16706705093384 0.3141167163848877 1.3600809574127197 ai4code 12
13 27.03321151 58.10323113203049 7.331982096 38.996995747089386 679.3021986484528 821.7508952617645 112.65009546279907 791.9522941112518 82.1675705909729 761.4697692394257 ai4code 13
14 29.172393831 58.75176697969437 7.334584767 0.3450809717178345 0.21559858322143555 821.966493844986 112.22568106651306 112.4412796497345 101.78269743919373 101.99829602241516 ai4code 14
15 31.311576152 58.94751399755478 7.334585879 0.34907931089401245 0.0014295578002929688 821.9679234027863 111.91239261627197 111.91382217407227 0.0041332244873046875 0.005562782287597656 ai4code 15
16 33.452929728 58.39855456352234 7.335900662 0.3474043011665344 0.13564705848693848 822.1035704612732 112.38870525360107 112.52435231208801 1.001966953277588 1.1374485492706299 ai4code 16
17 35.594283304 58.36668932437897 7.335901774 0.346748411655426 0.0014183521270751953 822.1049888134003 111.83719420433044 111.83861255645752 0.0018742084503173828 0.003292560577392578 ai4code 17
18 37.73568016 58.93553984165192 7.337235604 0.3498193025588989 0.12553644180297852 822.2305252552032 111.95638298988342 112.0819194316864 0.0264892578125 0.15134119987487793 ai4code 18
19 39.877953195 58.69380933046341 7.338115327 0.3333693742752075 0.004534244537353516 822.2350594997406 113.107750415802 113.11228466033936 1.0315678119659424 1.036102056503296 ai4code 19
20 42.021350969 58.55374866724014 7.338341651 0.35207128524780273 0.06152081489562988 822.2965803146362 112.47156310081482 112.53308391571045 1.4568238258361816 1.5183446407318115 ai4code 20
}\tableexpipercellaicode
\pgfplotstableread[header=has colnames]{
index snp_storage_gb snp_load_s pga_storage_gb pga_load_s exec_s cexec_s snp_save_s snp_total_s pga_save_s pga_total_s nb cell
0 3.04e-07 0.0006308555603027344 1.793e-06 0.002447366714477539 0.1980276107788086 0.1980276107788086 0.0006225109100341797 0.19865012168884277 0.0013110637664794922 0.19729971885681152 ecomsmph 0
1 2.902244312 171.37395429611206 2.902300438 214.17941617965698 50.24350452423096 50.441532135009766 275.6112411022186 325.8547456264496 0.0007452964782714844 43.17206883430481 ecomsmph 1
2 10.406232604 449.7997958660126 7.504127144 336.3614387512207 69.93970727920532 120.38123941421509 727.0285053253174 796.9682126045227 452.64527010917664 522.584977388382 ecomsmph 2
3 21.320432331 607.026534318924 10.914482689 600.0008587837219 9.71443772315979 130.09567713737488 1063.9839000701904 1073.6983377933502 725.7520503997803 735.4664881229401 ecomsmph 3
4 32.236935388 608.7465415000916 10.928482685 1.3705220222473145 0.7371468544006348 130.8328239917755 1067.0191040039062 1067.7562508583069 1932.6081683635712 1933.3453152179718 ecomsmph 4
5 43.153438445 607.0466487407684 10.928483448 600.0008945465088 5.24735164642334 136.08017563819885 1066.8729219436646 1072.120273590088 0.08239102363586426 5.329742670059204 ecomsmph 5
6 54.069921283 607.0430178642273 10.944292999 1.3300039768218994 0.08001995086669922 136.16019558906555 1072.5848989486694 1072.6649188995361 0.0756218433380127 0.15564179420471191 ecomsmph 6
7 64.986404894 607.20134806633 10.944299498 1.373711109161377 0.011908531188964844 136.17210412025452 1070.4160823822021 1070.427990913391 3.1576507091522217 3.1695592403411865 ecomsmph 7
8 75.902888505 607.1117985248566 10.948210558 1.4220402240753174 0.06770133972167969 136.2398054599762 1072.6463935375214 1072.714094877243 3.0133004188537598 3.0810017585754395 ecomsmph 8
9 86.82327814 606.9791142940521 10.952354683 2.586744785308838 2.400681495666504 138.6404869556427 1075.2899713516235 1077.69065284729 3.208317518234253 5.608999013900757 ecomsmph 9
10 97.74483956 606.8752422332764 10.953535638 2.7525064945220947 0.212080717086792 138.8525676727295 1071.6078026294708 1071.8198833465576 13.699182033538818 13.91126275062561 ecomsmph 10
11 108.667435225 606.8369722366333 10.957796681 2.9021337032318115 0.04159259796142578 138.89416027069092 1070.8069784641266 1070.848571062088 10.910316705703735 10.951909303665161 ecomsmph 11
12 119.590252527 606.870760679245 10.961273478 2.9614834785461426 0.27152204513549805 139.16568231582642 1074.3373465538025 1074.608868598938 10.399435758590698 10.670957803726196 ecomsmph 12
13 130.513109444 606.9533965587616 10.964642116 2.833970546722412 0.17240571975708008 139.3380880355835 1072.0963790416718 1072.2687847614288 10.440646409988403 10.613052129745483 ecomsmph 13
14 141.43599271 606.8649001121521 10.964727052 0.029769420623779297 0.03543734550476074 139.37352538108826 1073.9039034843445 1073.9393408298492 10.003769636154175 10.039206981658936 ecomsmph 14
15 152.358898004 607.0997347831726 10.964822675 0.040612220764160156 0.02391505241394043 139.3974404335022 1070.285855293274 1070.3097703456879 0.11815285682678223 0.14206790924072266 ecomsmph 15
16 163.282727457 606.8845250606537 10.974006362 2.899566173553467 0.08645415306091309 139.4838945865631 1071.2843704223633 1071.3708245754242 0.3472607135772705 0.4337148666381836 ecomsmph 16
17 174.20655867 607.1266376972198 10.974017421 0.006746768951416016 0.003787517547607422 139.48768210411072 1066.4055318832397 1066.4093194007874 10.08500623703003 10.088793754577637 ecomsmph 17
18 185.130451849 606.9750366210938 10.974153191 0.05546283721923828 0.04643750190734863 139.53411960601807 1067.2053604125977 1067.251797914505 0.04308342933654785 0.08952093124389648 ecomsmph 18
19 196.054345888 606.8786737918854 10.974159086 0.005461454391479492 0.2729668617248535 139.80708646774292 1070.485859155655 1070.7588260173798 0.34137392044067383 0.6143407821655273 ecomsmph 19
20 206.978240787 607.0428247451782 10.974164879 0.003910064697265625 0.25751709938049316 140.0646035671234 1073.9708437919617 1074.2283608913422 0.1096036434173584 0.36712074279785156 ecomsmph 20
21 217.90213655 607.1954221725464 10.974170718 0.0037784576416015625 0.26295971870422363 140.32756328582764 1072.9854023456573 1073.2483620643616 0.05439615249633789 0.3173558712005615 ecomsmph 21
22 228.826033485 607.2926528453827 10.974173731 0.00308990478515625 0.0006043910980224609 140.32816767692566 1070.3738317489624 1070.3744361400604 0.0022363662719726562 0.002829313278198242 ecomsmph 22
23 239.751768129 607.1589138507843 10.97763206 2.903348445892334 0.07310032844543457 140.4012680053711 1070.4371807575226 1070.510281085968 0.004788398742675781 0.07788872718811035 ecomsmph 23
24 250.677503149 606.9367477893829 10.980861701 2.7407586574554443 0.02600383758544922 140.42727184295654 1075.2679543495178 1075.2939581871033 8.85227084159851 8.87827467918396 ecomsmph 24
25 261.603238924 607.2504937648773 10.980867011 0.004918098449707031 0.002523183822631836 140.42979502677917 1069.5775451660156 1069.5800683498383 8.53776741027832 8.540290594100952 ecomsmph 25
26 272.528975224 607.1269204616547 10.98237281 2.6849281787872314 0.03936934471130371 140.46916437149048 1067.6738913059235 1067.7132606506348 0.21318864822387695 0.25255799293518066 ecomsmph 26
27 283.454711524 607.0408802032471 10.982377413 0.004392385482788086 0.002106904983520508 140.471271276474 1069.0159838199615 1069.018090724945 8.077731847763062 8.079838752746582 ecomsmph 27
28 294.380448237 606.9185183048248 10.985607327 2.8149375915527344 0.061434268951416016 140.53270554542542 1073.154113292694 1073.2155475616455 0.2595198154449463 0.3209540843963623 ecomsmph 28
29 305.306799059 607.0242500305176 10.989454284 2.9444143772125244 0.014263153076171875 140.5469686985016 1068.122444152832 1068.1367073059082 8.518543481826782 8.532806634902954 ecomsmph 29
30 316.233150462 606.8439147472382 10.989459609 0.005728006362915039 0.0026237964630126953 140.5495924949646 1070.4735271930695 1070.4761509895325 8.025287628173828 8.02791142463684 ecomsmph 30
31 327.15987754 607.0612661838531 10.995502491 2.9710841178894043 0.01229238510131836 140.56188488006592 1077.0236337184906 1077.035926103592 0.19762587547302246 0.20991826057434082 ecomsmph 31
32 338.086605199 608.7936730384827 10.995508068 0.0045244693756103516 0.0023703575134277344 140.56425523757935 1073.9939379692078 1073.9963083267212 7.290108919143677 7.2924792766571045 ecomsmph 32
33 349.013332986 607.0444741249084 10.995512366 0.003996133804321289 0.0014526844024658203 140.5657079219818 1069.4004075527191 1069.4018602371216 0.003462076187133789 0.004537105560302734 ecomsmph 33
34 359.940095107 607.0280160903931 10.995723759 0.07353758811950684 0.03724479675292969 140.60295271873474 1072.5151998996735 1072.5524446964264 0.14583444595336914 0.18307924270629883 ecomsmph 34
35 370.880441417 606.9494361877441 11.644390793 64.298095703125 0.32659006118774414 140.92954277992249 1073.368286371231 1073.6948764324188 1.4158828258514404 1.7424728870391846 ecomsmph 35
36 381.821494707 607.1683654785156 11.649190507 1.5850005149841309 0.009598016738891602 140.93914079666138 1080.838350534439 1080.847948551178 134.5624840259552 134.5720820426941 ecomsmph 36
37 392.762687823 607.0325074195862 11.652491703 0.2960522174835205 13.880286931991577 154.81942772865295 1080.433336019516 1094.3136229515076 5.361552715301514 19.24183964729309 ecomsmph 37
38 403.703879931 607.3057806491852 11.652671747 0.008895158767700195 0.011314630508422852 154.83074235916138 1077.025134563446 1077.0364491939545 0.3573427200317383 0.36865735054016113 ecomsmph 38
39 414.64519782 606.9511206150055 11.652803757 0.008273601531982422 0.01174163818359375 154.84248399734497 1068.8498675823212 1068.8616092205048 0.006833791732788086 0.018575429916381836 ecomsmph 39
40 425.586625688 606.9210169315338 11.653134153 0.10722517967224121 0.13245081901550293 154.97493481636047 1079.5356059074402 1079.6680567264557 0.33168840408325195 0.4641392230987549 ecomsmph 40
}\tableexpipercellecomsmph

\pgfplotstableread[header=has colnames]{
index tid criu_median_load_s dill_median_load_s pga_median_load_s shev_median_load_s zodb_median_load_s zosp_median_load_s
0 0 1.881821870803833 0.0005543231964111328 0.0015287399291992188 0.03149402141571045 0.03013312816619873 0.0013794898986816406
1 1 1.921797513961792 0.0005755424499511719 0.0009714365005493164 0.03111732006072998 0.030176997184753418 0.0014289617538452148
2 2 8.864209651947021 2.997584342956543 4.543018817901611 3.808641195297241 3.6223431825637817 3.6599985361099243
3 3 8.926539421081543 3.0385777950286865 4.500296592712402 3.8185718059539795 3.639498710632324 3.691025137901306
4 4 9.150938510894775 3.0655293464660645 4.626583576202393 3.903646230697632 3.697868585586548 3.7688299417495728
5 5 0.9752490520477295 3.153893232345581 4.61650824546814 0.11667418479919434 3.710590124130249 3.7578314542770386
6 6 1.071216106414795 3.1231234073638916 0.10316526889801025 0.11716711521148682 3.692350387573242 3.765600085258484
7 7 1.1068174839019775 3.1309597492218018 0.14876842498779297 0.15738928318023682 3.721179246902466 3.783666253089905
8 8 1.0477728843688965 3.093266725540161 0.1511918306350708 0.1587291955947876 3.757538676261902 3.7724415063858032
9 9 1.115217924118042 3.1482908725738525 0.19296395778656006 0.19841885566711426 3.7824007272720337 3.796500563621521
10 10 1.1143195629119873 3.160284996032715 0.19624578952789307 0.19973528385162354 3.7783912420272827 3.8064916133880615
11 11 0.7895395755767822 3.118790626525879 0.0014477968215942383 0.03122580051422119 3.787308692932129 3.795392155647278
12 12 1.261591911315918 3.1138193607330322 0.4044684171676636 0.37738656997680664 3.9721025228500366 4.024314761161804
13 13 1.3915512561798096 3.2190945148468018 0.3073066473007202 0.30013716220855713 4.048603177070618 4.106717228889465
14 14 0.7383518218994141 3.282066583633423 0.0016498565673828125 0.03117954730987549 4.017289042472839 4.110533714294434
15 15 0.8333690166473389 3.253706693649292 0.3051217794418335 0.06766819953918457 4.016634821891785 4.069862365722656
16 16 0.7583694458007812 3.2207796573638916 0.2675137519836426 0.04067659378051758 4.009328365325928 4.085416555404663
17 17 0.7614717483520508 3.266284704208374 0.2763873338699341 0.045827507972717285 4.028452634811401 4.108308434486389
18 18 0.8849947452545166 3.266798257827759 0.25669050216674805 0.07579803466796875 4.035627722740173 4.093069314956665
19 19 0.890925407409668 3.2553484439849854 0.25595784187316895 0.07622694969177246 4.03943943977356 4.10927426815033
20 20 0.6928482055664062 3.259718179702759 0.0010116100311279297 0.030918002128601074 0.053473591804504395 0.04916870594024658
21 21 0.7055935859680176 3.4245100021362305 0.0010235309600830078 0.03093886375427246 0.05370807647705078 0.04903578758239746
22 22 0.8606975078582764 3.2734193801879883 0.2576233148574829 0.07586812973022461 4.045641541481018 4.096073389053345
23 23 0.734926700592041 3.2543888092041016 0.2766324281692505 0.04598712921142578 4.0601314306259155 4.066859602928162
24 24 0.8378310203552246 3.3090507984161377 0.25809526443481445 0.07626020908355713 4.03816556930542 4.087356686592102
25 25 11.381054639816284 3.9821794033050537 5.564544677734375 4.801591873168945 4.92271089553833 5.002801895141602
26 26 1.1289401054382324 4.06072998046875 0.15667295455932617 0.1928313970565796 4.917425155639648 5.021994352340698
27 27 0.7600576877593994 4.047946214675903 0.002584218978881836 0.03207099437713623 4.91260552406311 4.980514407157898
28 28 1.9437205791473389 4.199854135513306 0.6243723630905151 0.727629542350769 5.057949423789978 5.130547046661377
29 29 1.964015007019043 4.203962087631226 0.5777173042297363 0.5419561862945557 5.063809394836426 5.148876309394836
30 30 1.5685575008392334 4.2148332595825195 0.4179297685623169 0.38541948795318604 5.1061238050460815 5.1770607233047485
31 31 1.6041882038116455 4.224746942520142 0.4232161045074463 0.38451123237609863 5.060809016227722 5.139584302902222
32 32 2.7318685054779053 4.334736108779907 0.9288990497589111 1.0482112169265747 5.26790452003479 5.3352168798446655
33 33 1.6086244583129883 4.344067096710205 0.6254812479019165 0.3846778869628906 5.240716576576233 5.293531894683838
34 34 1.632239580154419 4.401421308517456 0.4594205617904663 0.44781315326690674 5.289361238479614 5.310911536216736
35 35 2.3836255073547363 4.566009521484375 0.8534893989562988 1.0705034732818604 5.4593212604522705 5.502958536148071
36 36 0.9395859241485596 4.576421022415161 0.03941452503204346 0.061691880226135254 5.400526762008667 5.493740439414978
37 37 1.1936419010162354 4.521757125854492 0.19738376140594482 0.1967930793762207 5.4277263879776 5.534760117530823
38 38 0.8627986907958984 4.567801237106323 0.028018951416015625 0.05274856090545654 5.422949194908142 5.544627785682678
39 39 0.7819046974182129 4.582802057266235 0.02744579315185547 0.05208420753479004 5.427946329116821 5.521815061569214
40 40 0.9296519756317139 4.614873170852661 0.08679592609405518 0.11512601375579834 5.467456817626953 5.592257618904114
41 41 1.395324945449829 4.676462173461914 0.3422485589981079 0.36627495288848877 5.559225916862488 5.60733163356781
42 42 0.9028027057647705 4.7012317180633545 0.13450932502746582 0.12828481197357178 5.532176375389099 5.529140830039978
43 43 0.9234514236450195 4.58858323097229 0.16284418106079102 0.08117222785949707 5.447508096694946 5.449789524078369
}\tableexpiploadskltweet
\pgfplotstableread[header=has colnames]{
index tid criu_median_load_s dill_median_load_s pga_median_load_s shev_median_load_s zodb_median_load_s zosp_median_load_s
0 0 1.939518690109253 0.0006239414215087891 0.0022264719009399414 0.012523293495178223 0.8406229019165039 0.0020301342010498047
1 1 102.41108703613281 35.03896713256836 51.04378032684326 43.58430087566376 44.97039294242859 44.69644367694855
2 2 2.00087571144104 35.644914388656616 0.0015922784805297852 0.011814236640930176 45.01551306247711 44.62732541561127
3 3 72.30765891075134 35.2493679523468 36.2918883562088 30.961753726005554 45.08497619628906 45.0302791595459
4 4 129.18931651115417 44.578389167785645 50.975181579589844 73.86263751983643 54.80252122879028 54.55973660945892
5 5 57.55409383773804 44.98522067070007 46.032209038734436 24.631693720817566 54.96726059913635 54.94218575954437
6 6 61.83716297149658 44.97174096107483 28.765637278556824 24.01186430454254 54.960904240608215 54.605020403862
7 7 2.2056143283843994 45.84068441390991 0.052304744720458984 0.08893454074859619 55.23283386230469 54.67962312698364
8 8 57.85855460166931 45.23249101638794 28.44775903224945 24.53201675415039 54.967108845710754 55.26812410354614
9 9 42.77252173423767 44.84162735939026 41.600051164627075 17.868030309677124 54.82775616645813 54.74940323829651
10 10 44.0695059299469 44.98510551452637 21.457839369773865 18.22492504119873 55.412927746772766 54.81646108627319
11 11 2.192182779312134 44.94423460960388 0.04121983051300049 0.0739145278930664 54.87626838684082 54.91429591178894
12 12 43.46233344078064 44.9602427482605 21.308376789093018 18.11756432056427 54.81588268280029 54.73226618766785
13 13 83.53056621551514 46.77193546295166 39.063990116119385 32.982483863830566 57.562811613082886 57.05765724182129
14 14 2.6405751705169678 47.50453281402588 0.3446786403656006 0.2802814245223999 57.519875288009644 57.36996579170227
15 15 2.6378138065338135 47.57096195220947 0.34895241260528564 0.26265621185302734 57.76117289066315 57.11663794517517
16 16 2.5848357677459717 47.17713785171509 0.3445608615875244 0.2812999486923218 57.49424850940704 57.016817927360535
17 17 2.6228177547454834 47.29309916496277 0.34598493576049805 0.26512837409973145 57.454046845436096 57.17354226112366
18 18 2.6238038539886475 47.232905626297 0.35022854804992676 0.2831718921661377 57.217323899269104 57.234763979911804
19 19 2.686689853668213 47.72394895553589 0.3336256742477417 0.30863702297210693 57.114792585372925 57.532785415649414
20 20 2.6062304973602295 47.09588694572449 0.3529047966003418 0.2838306427001953 57.30004954338074 57.02253031730652
}\tableexpiploadaicode
\pgfplotstableread[header=has colnames]{
index tid criu_median_load_s dill_median_load_s pga_median_load_s shev_median_load_s zodb_median_load_s zosp_median_load_s
0 0 3.1974802017211914 0.0007865428924560547 0.001712799072265625 0.006349802017211914 1.5633008480072021 0.0025217533111572266
1 1 6.064088582992554 0.788456916809082 1.6014938354492188 1.552990198135376 3.957690954208374 3.4267213344573975
2 2 6.203167200088501 0.7588796615600586 1.5544133186340332 1.5669288635253906 4.041912078857422 3.4072301387786865
3 3 6.492919683456421 0.7599835395812988 1.5546176433563232 1.5167725086212158 4.16685152053833 3.5667083263397217
4 4 11.182701349258423 1.5045068264007568 3.0225303173065186 3.1152408123016357 6.350940942764282 6.70999002456665
5 5 6.924327850341797 1.5510964393615723 3.0630030632019043 1.5663182735443115 6.433945417404175 6.767306327819824
6 6 6.481212854385376 1.5408058166503906 2.933522939682007 1.6256661415100098 6.332332611083984 6.773514270782471
7 7 6.148644685745239 1.5449271202087402 2.8493571281433105 1.2563960552215576 6.3438873291015625 6.723032236099243
8 8 2.490387439727783 1.6692538261413574 0.09247493743896484 0.10107612609863281 6.511790752410889 6.845372676849365
9 9 5.379019498825073 1.694075584411621 0.11136293411254883 0.12659502029418945 6.471327304840088 6.929198741912842
10 10 8.480592966079712 1.6357040405273438 2.8571226596832275 1.4273362159729004 6.455469608306885 6.898061275482178
11 11 5.730705261230469 1.6809327602386475 1.9017913341522217 0.46135663986206055 6.331913948059082 6.80410099029541
12 12 5.458991289138794 1.6845242977142334 1.8203487396240234 0.34636545181274414 6.5402209758758545 6.793839454650879
13 13 4.621214151382446 1.6770813465118408 1.5237703323364258 0.006776094436645508 6.466709852218628 6.817770719528198
14 14 4.698560953140259 1.7012312412261963 1.519118070602417 0.006524085998535156 6.312997579574585 6.810298919677734
}\tableexpiploadagripred
\pgfplotstableread[header=has colnames]{
index tid criu_median_load_s dill_median_load_s pga_median_load_s shev_median_load_s zodb_median_load_s zosp_median_load_s
0 0 1.9312965869903564 0.0009579658508300781 0.0027799606323242188 0.024735450744628906 1.996748447418213 0.0029850006103515625
1 1 2.0285959243774414 0.0007863044738769531 0.0009670257568359375 0.024090290069580078 2.004481315612793 0.002335071563720703
2 2 2.773174524307251 0.7729465961456299 0.9754238128662109 0.7935576438903809 2.7728662490844727 0.7931575775146484
3 3 3.0436179637908936 0.8152117729187012 1.0320243835449219 0.8849267959594727 2.8294167518615723 0.8457906246185303
4 4 7.114777326583862 1.4270992279052734 1.818016767501831 1.6092236042022705 3.478489637374878 1.5233361721038818
5 5 2.966756820678711 1.461587905883789 1.0978114604949951 0.8533406257629395 3.473618268966675 1.4896066188812256
6 6 16.6446430683136 3.4603374004364014 3.7551536560058594 3.9789085388183594 10.313087940216064 11.775361061096191
7 7 13.52112627029419 3.482276439666748 3.8026621341705322 3.943317174911499 10.06029486656189 10.523507118225098
8 8 16.43744921684265 3.451205015182495 3.733254909515381 3.9293789863586426 10.010751724243164 10.382175207138062
9 9 17.231325387954712 4.2214860916137695 4.786235809326172 6.0128867626190186 11.06238317489624 11.018104553222656
10 10 7.529056787490845 2.754483222961426 2.180349826812744 2.298475742340088 6.914752006530762 6.07240104675293
11 11 6.499216318130493 3.4332239627838135 1.6336984634399414 1.6999523639678955 8.796539068222046 8.64698076248169
12 12 9.704592943191528 3.341247320175171 2.622861862182617 2.8481996059417725 8.841068267822266 8.687938690185547
13 13 9.45313286781311 3.318420886993408 2.545938014984131 2.8134548664093018 9.019781827926636 8.692152738571167
14 14 3.6757779121398926 3.4203898906707764 2.2113699913024902 0.7571985721588135 8.794234275817871 8.618474245071411
15 15 2.0562446117401123 3.394347906112671 0.0011866092681884766 0.02356266975402832 3.6677660942077637 3.6167619228363037
16 16 5.640801906585693 4.169373035430908 2.573240041732788 2.156607151031494 9.649665832519531 9.483923435211182
17 17 3.6263539791107178 3.5245988368988037 2.1543102264404297 0.7340342998504639 8.899169921875 8.790399551391602
18 18 2.094895601272583 3.5290112495422363 0.0012378692626953125 0.023690462112426758 3.673957586288452 3.656696081161499
19 19 2.123666763305664 3.5108799934387207 1.8618831634521484 0.05689716339111328 8.969085216522217 8.794148206710815
20 20 21.140208959579468 3.509535551071167 0.7764825820922852 5.760538578033447 9.13151741027832 9.019148349761963
21 21 2.039689540863037 3.5211801528930664 0.003657102584838867 0.0248258113861084 9.100397825241089 8.919359683990479
22 22 17.664263486862183 3.6136438846588135 0.8795962333679199 5.867793321609497 9.196199417114258 9.046559810638428
23 23 17.097323894500732 3.5917892456054688 0.29132080078125 5.541682481765747 9.256198644638062 9.010636329650879
24 24 2.225092887878418 3.5659027099609375 0.05701279640197754 0.10874319076538086 9.132915258407593 9.632666826248169
25 25 4.8483710289001465 3.7035045623779297 1.118654727935791 1.0295131206512451 9.818684816360474 10.242636919021606
26 26 3.486332654953003 3.7023513317108154 0.380629301071167 0.4094994068145752 9.833276748657227 10.090211391448975
27 27 17.566730976104736 3.717702865600586 0.19414734840393066 nan 9.789427042007446 10.097873449325562
}\tableexpiploadmsciedaw
\pgfplotstableread[header=has colnames]{
index tid criu_median_load_s dill_median_load_s pga_median_load_s shev_median_load_s zodb_median_load_s zosp_median_load_s
0 0 1.8024194240570068 0.0005738735198974609 0.002447366714477539 0.025461196899414062 3.6113922595977783 0.021886110305786133
1 1 457.59813046455383 167.7987837791443 214.17941617965698 168.73454904556274 174.16518425941467 176.12939381599426
2 2 733.9792227745056 432.0712733268738 336.3614387512207 270.0965576171875 447.75597286224365 451.8991267681122
3 3 1784.2995917797089 643.0205161571503 600.0008587837219 666.5374846458435 657.3495051860809 679.1474597454071
4 4 12.656126976013184 648.6677832603455 1.3705220222473145 1.1150531768798828 660.4835059642792 678.0662972927094
5 5 316.27541494369507 645.1602110862732 600.0008945465088 114.07495355606079 657.6458518505096 677.1051669120789
6 6 12.4599769115448 646.7546262741089 1.3300039768218994 1.1107535362243652 663.1259565353394 675.0004034042358
7 7 12.675613403320312 646.2441987991333 1.373711109161377 1.1342852115631104 656.0397703647614 679.5897591114044
8 8 12.644606828689575 646.3274257183075 1.4220402240753174 1.1061439514160156 659.7559185028076 678.0159342288971
9 9 17.101531982421875 652.2139956951141 2.586744785308838 2.0584981441497803 661.0627334117889 682.3561434745789
10 10 17.203794240951538 647.5009350776672 2.7525064945220947 2.212735891342163 653.7619652748108 683.0042479038239
11 11 17.136858701705933 651.0555346012115 2.9021337032318115 2.219294309616089 656.8323862552643 679.7217106819153
12 12 18.372743368148804 651.8146712779999 2.9614834785461426 2.433727264404297 660.5494132041931 679.092592716217
13 13 17.503509044647217 647.011869430542 2.833970546722412 2.311325788497925 653.8381941318512 684.2821774482727
14 14 9.822913646697998 648.3660571575165 0.029769420623779297 0.08202791213989258 653.4190986156464 682.7051846981049
15 15 9.928304195404053 653.1773471832275 0.040612220764160156 0.09412598609924316 659.8387584686279 670.6445624828339
16 16 17.806105375289917 650.7347681522369 2.899566173553467 2.212862968444824 661.5713987350464 679.1689369678497
17 17 10.131132125854492 648.6633682250977 0.006746768951416016 0.03215456008911133 666.0267281532288 677.6141130924225
18 18 10.12227749824524 649.0713775157928 0.05546283721923828 0.12702083587646484 655.9127883911133 685.9472763538361
19 19 9.832016468048096 650.3158760070801 0.005461454391479492 0.026881694793701172 664.6085305213928 684.1606540679932
20 20 9.944578409194946 651.7680871486664 0.003910064697265625 0.026867151260375977 655.8393032550812 679.4189693927765
21 21 9.852567195892334 648.5499503612518 0.0037784576416015625 0.026894807815551758 658.4071309566498 682.2157533168793
22 22 1775.5320298671722 646.1856908798218 0.00308990478515625 657.2001173496246 653.0133366584778 680.3789200782776
23 23 16.906981945037842 644.4432644844055 2.903348445892334 2.2111799716949463 660.2404556274414 683.8646957874298
24 24 16.87443447113037 646.8070342540741 2.7407586574554443 2.2155091762542725 662.5565013885498 684.8757953643799
25 25 1785.3916375637054 643.2834656238556 0.004918098449707031 655.998456954956 659.3345165252686 685.0211727619171
26 26 17.00328755378723 644.1333668231964 2.6849281787872314 2.2116076946258545 658.566264629364 682.6850926876068
27 27 1813.2527203559875 644.9829428195953 0.004392385482788086 653.1196413040161 662.7895112037659 677.8578717708588
28 28 16.990028381347656 645.8369963169098 2.8149375915527344 2.2160356044769287 657.6808683872223 685.8413062095642
29 29 17.098278760910034 645.2635982036591 2.9444143772125244 2.324425458908081 664.8260145187378 682.481050491333
30 30 14.554795980453491 646.6442453861237 0.005728006362915039 0.027047157287597656 665.0498163700104 688.2077808380127
31 31 16.999372720718384 644.3291521072388 2.9710841178894043 2.2969789505004883 656.9879813194275 688.8478846549988
32 32 9.34920072555542 646.4774248600006 0.0045244693756103516 nan 665.3262791633606 683.4253554344177
33 33 10.576134204864502 646.1475331783295 0.003996133804321289 nan 654.4379386901855 686.0197982788086
34 34 9.653322696685791 649.8278207778931 0.07353758811950684 nan 655.6085796356201 680.2454881668091
35 35 nan 647.4206962585449 64.298095703125 nan 656.6571898460388 688.0126557350159
36 36 nan 649.969183921814 1.5850005149841309 nan 659.1612203121185 677.3618054389954
37 37 nan 647.8547818660736 0.2960522174835205 nan 661.4168071746826 671.7721478939056
38 38 nan 647.1872477531433 0.008895158767700195 nan 664.9969205856323 nan
39 39 nan 644.9720637798309 0.008273601531982422 nan 666.5449182987213 nan
40 40 nan 645.9301016330719 0.10722517967224121 nan 661.2438507080078 nan
}\tableexpiploadecomsmph

\input{figures/_data/exp_iv_waits}
\input{figures/_data/exp_iv_ecdf}

\section{Introduction}


Modern data science tools, ranging from batch scripts~\cite{pythonscript,rscript,taylor2024data} to computational notebooks~\cite{jupyter,jupyterlab,DBLP:journals/cse/PerezG07ipython,Baumer_2014rmarkdown,baumer2015rmarkdown,xie2018rmarkdowndefinitive,kishu,kishudemo,elasticnotebookdemo}, rely upon a massive and evolving object graph---a new and increasingly critical modality of data.
These underlying objects are used to represent many different forms of data,
    such as (semi-)structured data (e.g., \texttt{pandas}, \texttt{ElementTree}), 
        graphs (e.g., \texttt{networkx}),
        machine learning models (e.g., \texttt{pytorch}, \texttt{tensorflow}), 
        time series data, and so on.
With the ability to save these objects efficiently
    and restore them correctly,
    we can improve systems' resilience against unexpected failures~\cite{DBLP:conf/chi/ChattopadhyayPH20nbpainpoints,de2024bug,colabdisconnect}
    and even allow data scientists to keep track of their work
        and explore alternative hypotheses~\cite{DBLP:conf/chi/ChattopadhyayPH20nbpainpoints,kishuchi,kishuchilbw}.

However, existing persistence solutions are vastly infeasible.
For each saving, existing persistence libraries like Pickle~\cite{pickle}, Dill~\cite{dill1,dill2}, and ZODB~\cite{zodb} traverse these objects through their references and create a complete byte-representation snapshot of the object graph.
This means saving variables at multiple points during a script's execution may capture the same objects, even if they remain unchanged (e.g., blue nodes in \cref{fig:intro_massive_cell1,fig:intro_massive_cell2,fig:intro_massive_cell3}).
As a result, these methods \fix[R2W1,R2D1]{must handle \emph{2.96M} objects to save 42 snapshots of a real notebook \buildats (\Cref{fig:deltacount}), incuring an infeasible storage and time cost prohibiting safer data exploration.}

At the core, unlike DBMS'es~\cite{psql-buffer,mohan1992aries,mysql-buffer,upp,mojoframe,sieve}, the existing object persistence mechanisms lack \emph{delta identification} (i.e., identifying differences between object graphs) to avoid unnecessary saving.
Nonetheless, existing delta identifications only offer inaccurate, incomprehensive, and/or inefficient solutions.
Graph databases like Neo4j~\cite{besta2023demystifying} may inaccurately record nodes and edges after implicit updates where an execution appears static or access only a few variables but mutates nested shared structures, e.g., importing libraries (\Cref{fig:intro_massive_cell1}), data cleaning (\Cref{fig:intro_massive_cell2}), and model training (\Cref{fig:intro_massive_cell3}).
Manual approaches such as annotating objects to track~\cite{objecttracker,shelve,redisshelve,chest}, translating into database formats~\cite{actiannosql,db4o,DBLP:journals/cacm/LambLOW91objectstore,DBLP:conf/sigmod/OrensteinHMS92objectstoreqp,qstore,airindex}, or selectively saving some object types~\cite{npio,dfio,tfsave,torchsave,pysos,pickledb} do not comprehensively capture all objects, especially given the diversity of user-defined types and external library internals (e.g., \Cref{fig:intro_massive_all} shows 4 built-in, 63 external, and 37 user-defined variables).
Finally, classical algorithms like graph alignment~\cite{zeng2021comprehensive,yan2016short}, edit-distance algorithms~\cite{gao2010survey,neuhaus2006fast,piao2023computing}, and graph partitioning~\cite{DBLP:journals/pvldb/AbbasKCV18graphpartsurvey}
are inefficient for the millions of objects and references present in real sessions (\Cref{fig:intro_massive_numobjects}).

\def\subfigwidth{0.33\linewidth}
\def\subfigheight{32mm}

\pgfplotsset{
    ploadplot/.style={
        height=\subfigheight,
        width=\linewidth,
        axis lines=left,
        xlabel=State number,
        xlabel shift={-1mm},
        ymajorgrids,
        minor y tick num=1,
        minor x tick num=1,
        label style={font=\normallabelsize},
        every tick label/.append style={font=\normallabelsize},
    },
}
\tikzset{
    ploadsnp/.style={%
        draw=cNegativemain,
        mark=x,
        mark options={
            draw=cNegativemain,
            fill=cNegativelight,
            fill opacity=0.0,
            scale=0.75
        }
    },
    >=LaTeX
}
\tikzset{
    ploadpos/.style={%
        draw=cPositivemain,
        mark=*,
        mark options={
            draw=cPositivemain,
            fill=cPositivelight,
            fill opacity=0.0,
            scale=0.75
        }
    },
    >=LaTeX
}

\begin{figure}[t]
    \centering
    \begin{tikzpicture}
    \begin{axis}[
        ploadplot,
        ylabel=\# saved objects,
        ymode=log,
        ymajorgrids,
        yminorgrids,
        xmin=0,
        xmax=43,
        ymin=100000,
        ymax=25000000,
        ytick={100000, 1000000, 10000000, 100000000},
        legend columns=-1,
        legend style={at={(0.5,1.0)},anchor=north,align=center,/tikz/every even column/.append style={column sep=3mm},nodes={scale=0.65, transform shape}},
    ]
        \addplot[ploadsnp] table [x=nth,y=all_objects] {\tabledeltacount};
        \addlegendentry{Snapshot-based storage}
        \addplot[ploadpos] table [x=nth,y=delta_objects] {\tabledeltacount};
        \addlegendentry{With delta identification}
%
        \node[cPositivemain, anchor=south east, align=right, font=\scriptsize] at (axis cs:42, 200000) {Only 6.1\% of objects change (16.5$\times$ opportunity)};
    \end{axis}
    \end{tikzpicture}
    \vspace{-4mm}
    \caption{\fix[R2W1,R2D1]{Cumulative numbers of \buildats's objects to be saved by snapshot-based storage vs. delta storage.}}
    \label{fig:deltacount}
    \vspace{\undercaptionspace}
\end{figure}
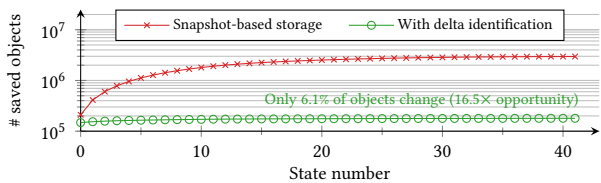

\paragraph{Goldilocks granularity of deltas: a proposal and challenges}
By contrasting these solutions, we propose a novel concept for efficiently identifying deltas between massive and evolving object graphs:
instead of object-level deltas and graph-level snapshots, one should find the Goldilocks (i.e., ``just right'') granularity of deltas by partitioning the object graph into subgraphs.
To identify dirty objects between two points in time, we compare the subgraphs; only non-matching subgraphs need to be persisted.

However, finding the right granularity comes with two new challenges: within-subgraph and cross-subgraph.
\textbf{\emph{(Within-subgraph)}} We must optimize object-subgraph assignments in terms of both node composition and subgraph size. The composition should be \emph{stable}: an object graph should consist of nearly the same set of subgraphs compared to the previous version, to minimize new subgraphs that need to be saved. The sizes should strike a \emph{balance}: small enough to isolate only the dirty objects, yet large enough to reduce per-subgraph management overheads.
\textbf{\emph{(Cross-subgraph)}} All references across subgraphs must be properly preserved when saving dirty subgraphs and restoring prior states.

Subgraph deltas enable an accurate, comprehensive, and efficient delta identification for massive object graphs. Compared to graph databases, it can track cross-subgraph dependencies to capture both object and reference changes, including implicit updates. Unlike manual annotation systems, it requires no user intervention and supports all serializable objects via its compatibility with generic protocols like Pickle~\cite{pickle}. In contrast to graph alignment and partitioning algorithms, we can trade off false 
positives (i.e., unchanged objects in changed subgraphs) for scalability.


\paragraph{Our approach}

With novel techniques to tackle the within-subgraph and cross-subgraph challenges, we present an algorithm called \emph{podding}: a structure-aware partitioning of object graphs into subgraphs called \emph{pods}.
To determine the optimal composition and size of pods, we formulate an optimization problem that aims to minimize expected persistence costs
    considering the properties of object graphs including both node-specific properties (e.g., object size) and structural attributes (e.g., fanout).
This principled approach 
    naturally avoids extreme cases
        such as splitting
         or bundling all the objects.
    
Second, pods are designed to correctly preserve cross-pod edges
    by distinguishing local references and cross-pod references.
Cross-pod references can be considered virtual addresses
    that are later resolved to appropriate local references
        when related objects need to be put together during restoration.
\cut{
    Pod-level dependencies, uniquely arising in our setting,
        are also automatically tracked by pods
            to ensure those cross-pod references are valid during restoration.
}


Combining these techniques,
    we introduce \system, 
        a system designed to serve as an off-the-shelf persistence library. 
Aligning with the latest trends in data science and machine learning,
    \system focuses on Python-based ecosystem,
        supporting 100\% libraries empirically.
This includes the libraries for 
        data analytics (e.g., \texttt{pandas}, \texttt{arrow}),
        machine learning (e.g., \texttt{pytorch}, \texttt{tensorflow}, \texttt{scikit-learn}),
        distributed computing (e.g., \texttt{ray}),
        and data visualization (e.g., \texttt{matplotlib}, \texttt{seaborn}).
As a result, \system can fully replace existing generic persistence tools
        such as Pickle~\cite{pickle},
            Dill~\cite{dill1,dill2}, and Cloudpickle~\cite{cloudpickle},
delivering significant improvements in both \fix[R2W1,R2D1]{time (up to 12.4$\times$) and storage (up to 36.5$\times$)} efficiency.

\paragraph{Contribution} 
This work makes the following contributions:
\begin{enumerate}
    \item \emph{Partial saving and loading} through podding, change detector, synonym resolver, and active variable filter (\cref{sec:podding}),
    \item \emph{Podding optimization} via a learned greedy algorithm (\cref{sec:podding-optimization}), 
    \item And, \emph{Asynchronous saving} safeguarded by active variable locking and static allowlist-based static code checker (\cref{sec:async}).
\end{enumerate}

\noindent
We empirically compare \system to existing systems/libraries in \cref{sec:exp}.
Lastly, we discuss related work and future opportunities (\cref{sec:related}). 

\section{Motivation}
\label{sec:motivation}



\cut{
    This section motivates graph-based deltas (\cref{sec:motivation-partial}) in storing objects.
    Next, it discusses related technical challenges (\cref{sec:motivation-persist}).
    
    \subsection{A Case for Structure-aware Deltas}
    \label{sec:motivation-partial}
}

Practical object persistence needs structure-aware deltas to reduce storage and time costs,
otherwise infeasible through snapshotting and/or classical byte-string delta compression.

\paragraph{Limitation of snapshotting} Snapshotting (e.g., \dill and \zosp) stores and recovers a given state as a whole on disk.
Suppose \texttt{dataset} is 10.1~GB large while \texttt{imputer} and \texttt{model} are 0.1~GB each, \emph{snapshotting all three states} would take up $10.1 + (10.1 + 0.1 + 0.1) + (10.1 + 0.1 + 0.1)$ = \emph{30.7~GB} space and proportional time to process them. Furthermore, if the user would like to ``study the distribution of rows with missing values in \texttt{dataset} 9'' (Query1) would require loading the entire 10.1~GB snapshot.
\cut{
    while ``comparing the performance between \texttt{model} trained with an imputed partition at 10~AM and \texttt{model} trained with latest data at 11~AM'' (Query2) would require loading two snapshots at 10~AM and 11~AM totaling 20.6~GB of data to load
}
That is, snapshotting can be inefficient for both saving and loading objects.


\paragraph{Limitation of byte-level deltas} Delta compression algorithms like LZ77~\cite{ziv1977universallz77}, xdelta~\cite{macdonald2000filexdelta}, and others~\cite{hunt1998delta,suel2002zdelta,xia2014ddelta,xia2015edelta} encode the differences of \emph{target byte string} over \emph{reference byte string(s)}~\cite{suel2019delta}.
For example, block move-based approaches~\cite{tichy1984stringblockmove} replace fixed-size blocks of byte-string deltas between the target and reference.
While delta compression may reduce storage usage, it is slower than snapshotting \fix{because doing so} requires both serializing the entire namespace and compressing the byte string each time. Meanwhile, loading involves reading both reference and target-delta byte strings, applying the delta, and deserializing the entire namespace. Determining the \emph{reference namespace byte string} also remains an open question.

\cut{
    \def\subfigwidth{0.49\linewidth}
\def\subfigheight{30mm}

\pgfplotsset{
    barplot/.style={
        height=\subfigheight,
        width=1.0\linewidth,
        ybar,
        bar width=0.24,
        axis lines=left,
        ymajorgrids,
        log origin=infty,
        enlarge x limits=0.15,
        ytick={0.01, 0.1, 1, 10, 100, 1000, 10000},
        xtick={data},
        xticklabels={},
        nodes near coords,
        point meta=explicit symbolic,
        nodes near coords style={font=\scriptsize, align=left, rotate=90, anchor=west},
        legend columns=2,
        legend style={at={(0.03,1.0)},anchor=north west},
        label style={font=\normallabelsize},
        every tick label/.append style={font=\normallabelsize},
    },
}

\tikzset{
    hbrace/.style={%
        draw=pgacolor,
        anchor=south,
        decorate,
        decoration={brace,amplitude=3pt,raise=1.5ex},
    },
    >=LaTeX
}

\tikzset{
    hbracetext/.style={%
        midway,
        yshift=3mm,
        anchor=south,
        inner xsep=0,
        align=center,
        text=pgacolor,
        font=\scriptsize,
    },
    >=LaTeX
}

\begin{figure}[t]
    \centering
    \begin{subfigure}[b]{\linewidth} \centering
        \begin{tikzpicture}
        \begin{axis}[
            ybar,
            ticks=none,
            height=20mm,
            width=\linewidth,
            hide axis,
            xmin=10,  
            xmax=50,
            ymin=0,
            ymax=0.4,
            area legend,
            legend columns=3,
            legend style={at={(0.0,0.0)},anchor=south,align=center,/tikz/every even column/.append style={column sep=3mm},nodes={scale=0.75, transform shape}},
        ]
            \node[align=center, opacity=1] {
                \addlegendimage{bardill}
                \addlegendentry{\snp}
                \addlegendimage{barsnx}
                \addlegendentry{\snx}
                \addlegendimage{barpga}
                \addlegendentry{\pga (Ours)}
            };
        \end{axis}
        \end{tikzpicture}
    \end{subfigure}
    \begin{subfigure}[b]{\subfigwidth}
        \begin{tikzpicture}
        \begin{axis}[
            barplot,
            ylabel=Storage (GB),
            ymode=log,
            ymin=0.01,
            ymax=150,
        ]
            \addplot[bardill] table [x=index,y=snp_storage_gb] {\tableexpiskltweet};
            \addplot[barsnx] table [x=index,y=snx_storage_gb] {\tableexpiskltweet};
            \addplot[barpga] table [x=index,y=pga_storage_gb] {\tableexpiskltweet};
            
        \end{axis}
        \end{tikzpicture}
        \vspace{-5mm}
        \caption{Total storage size}
        \label{fig:intro_delta_storage}
    \end{subfigure}
    \begin{subfigure}[b]{\subfigwidth}
        \begin{tikzpicture}
        \begin{axis}[
            barplot,
            ylabel=Save Time (s),
            ymode=log,
            ymin=0.1,
            ymax=1500,
        ]
            \addplot[bardill] table [x=index,y=snp_avg_save_s] {\tableexpiskltweet};
            \addplot[barsnx] table [x=index,y=snx_avg_save_s] {\tableexpiskltweet};
            \addplot[barpga] table [x=index,y=pga_avg_save_s] {\tableexpiskltweet};
            
        \end{axis}
        \end{tikzpicture}
        \vspace{-5mm}
        \caption{Average saving time}
        \label{fig:intro_delta_save}
    \end{subfigure}
    \vspace{-3mm}
    \caption{Object-aware deltas reduce storage and time costs.}
    \label{fig:intro_delta}
    \vspace{\undercaptionspace}
\end{figure}
}

\paragraph{Our approach: structure-aware deltas} Instead, structure awareness can \fix{both} reduce the storage and time overheads significantly. Suppose an object store detects that \texttt{imputer} only imputes 5\% of the dataset (0.5~GB), the storage would only need to store the original dataset (10.1~GB), the imputed dataset (0.5~GB), \texttt{imputer} (0.1~GB), and two versions of \texttt{model} (0.1~GB each), totaling $10.1 + (0.5 + 0.1 + 0.1) + 0.1$ = \emph{10.9~GB}---almost a $3\times$ reduction over snapshotting. Being aware of unrelated objects, the object store could filter out and bypass serializing those objects. Moreover, structure-aware deltas would allow precisely loading the relevant parts, e.g., when \texttt{model} points to 25\% of the dataset, answering Query1 and Query2 would load only 10.1~GB and $(2.5 + 0.1) + (2.5 + 0.1)$ = 5.2~GB respectively.



\cut{
    \subsection{Overcoming Persistence Challenges}
    \label{sec:motivation-persist}
    
    \system's novelty contributes to addressing the following persistence challenges for partial object stores.
    While other challenges such as simplifying programming, satisfying orthogonal persistence, and supporting schema evolution are important, they are orthogonal and addressed by existing works like Pickle~\cite{pickle}, Dill~\cite{dill1,dill2}, and Cloudpickle~\cite{cloudpickle} which \system is built upon.
    
    \paragraph{Identifying dirty objects} Oftentimes, only a few objects mutate between code executions. Simply saving the entire namespace or saving each object individually would result in huge storage and processing overheads. Per-object tracking~\cite{eliot1995expressing} would incur execution penalties~\cite{atkinson1995orthogonallyprinciple}. Whereas many mitigate this issue at a lower level in other settings such as memory paging~\cite{soukup1994taming,mmap,biliris1993makingpersistpointer,aritsugi1995severalpersistpointer,soukup2014fundamentals} or relational databases' buffer pools, such a mechanism is missing in the data science tools, since objects there might reside in different locations beyond the program's memory (i.e., GPUs and remote machines).
    To address this, we introduce a dynamic paging mechanism we refer to as \emph{podding} that gathers and groups objects optimally into page-like persistence units, called \emph{pods} (\cref{sec:podding-optimization}).
    
    \paragraph{Maintaining object references} Storing serialized objects across pods naively through existing serialization libraries could break the relationships between them. To overcome this limitation, we need to define a swizzling and unswizzling method that translates in-memory references to/from references across current and past pods. Inspired by operating systems' virtual address space that decouples local and physical (global) addresses~\cite{tanenbaum2009modern}, \system seamlessly translates object references through \emph{virtual memo space} to build upon Pickle's memo-ID-based swizzling protocol for within-pod and across-pod references. To manage references to past pods, \system implements \emph{change detector} and \emph{synonym resolver} that detect, link, and reuse past pods (\cref{sec:podding}).
    
    \paragraph{Ensuring atomicity} Our solution enables \emph{asynchronous saving}, a novel opportunity to unblock users' workflow by ensuring \emph{atomicity} where saved states must be results of complete (and not partial) code executions (similar to thread-atomicity~\cite{baldassin2021persistentsurvey}). Without a safeguard, ongoing code executions may mutate objects and corrupt the saved state. To guarantee atomicity, \system employs an active variable locking (\cref{sec:async-lock}) to safely minimize blocking on code executions (\cref{sec:async}).
}

\section{System Overview}
\label{sec:overview}

\def\sldist{10mm }

\tikzset{
    cellnode/.style={%
        draw=cZmain,
        fill=cZlightlight,
        anchor=west,
        align=left,
        font=\tiny,
        text width=14mm,
        inner xsep=1mm,
    },
    >=LaTeX
}
\tikzset{
    explorenode/.style={%
        draw=cZmain,
        fill=cZlightlight,
        align=left,
        anchor=west,
        text width=27mm,
        font=\tiny,
    },
    >=LaTeX
}
\tikzset{
    statenode/.style={%
        draw=cAmain,
        fill=cAlight,
        font=\scriptsize,
    },
    >=LaTeX
}
\tikzset{
    savenode/.style={%
        above,
        sloped,
        text=cAmain,
        font=\small,
        anchor=south west,
        align=left,
        pos=0.1,
        inner xsep=0mm,
    },
    >=LaTeX
}
\tikzset{
    loadnode/.style={%
        above,
        sloped,
        text=cAmain,
        font=\small,
        anchor=south west,
        align=left,
        pos=0.25,
        inner xsep=0mm,
    },
    >=LaTeX
}
\tikzset{
    execnum/.style={%
        text width=2.5mm,
        align=left,
        font=\scriptsize,
    },
    >=LaTeX
}

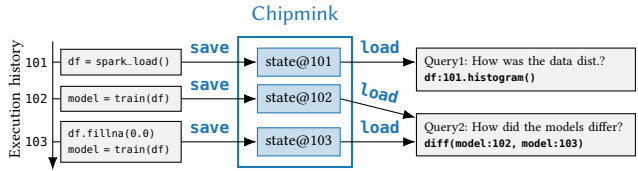
\begin{figure}[t]
    \centering
    \begin{tikzpicture}
        \node[
            matrix,
            column sep={\sldist},
            row sep={1mm},
            nodes={},
        ] (mainnode) {
            \node[cellnode] (cell1) {\texttt{df = spark_load()}}; &
            \node[statenode] (state1) {state@101};
            \\
    
            \node[cellnode] (cell2) {\texttt{model = train(df)}}; &
            \node[statenode] (state2) {state@102}; \\
    
            \node[cellnode] (cell3) {\texttt{df.fillna(0.0)}\\\texttt{model = train(df)}}; &
            \node[statenode] (state3) {state@103};
            \\
        };

        \node[explorenode, right=\sldist of state1.north east, anchor=north west] (explore1) {Query1: How was the data dist.?\\\textbf{\texttt{df:101.histogram()}}};
        \coordinate(mideaststate23) at ($(state2.east)!0.5!(state3.east)$);  
        \node[explorenode, right=\sldist of state3.south east, anchor=south west] (explore2) {Query2: How did the models differ?\\\textbf{\texttt{diff(model:102, model:103)}}};

        \draw[->] (cell1.east) -- node[savenode] {\textbf{\texttt{save}}} (state1.west);
        \draw[->] (cell2.east) -- node[savenode] {\textbf{\texttt{save}}} (state2.west);
        \draw[->] (cell3.east) -- node[savenode] {\textbf{\texttt{save}}} (state3.west);

        \draw[->] (state1.east) -- node[loadnode] {\textbf{\texttt{load}}} ($(explore1.north west)!(state1.east)!(explore1.south west)$);
        \coordinate (explore2-outanchor) at ($(explore2.west) - (1.5mm, 0)$);
        \draw[->] (state2.east)
            -- node[loadnode, pos=0.2] {\textbf{\texttt{load}}} ($(explore2.north west)!0.1!(explore2.south west)$);
        \draw[->] (state3.east)
            -- node[loadnode] {\textbf{\texttt{load}}} ($(explore2.north west)!(state3.east)!(explore2.south west)$);

        \node[execnum, left=1mm of cell1.west] (t1) {\texttt{101}};
        \node[execnum, left=1mm of cell2.west] (t2) {\texttt{102}};
        \node[execnum, left=1mm of cell3.west] (t3) {\texttt{103}};
        \node[left=-1.5mm of t2.west, anchor=south, align=center, rotate=90, text depth=0mm, font=\scriptsize, text width=20mm] (hist) {Execution history};
        \coordinate (t1-outanchor) at ($(t1.north east) + (0, 1mm)$);
        \coordinate (t3-outanchor) at ($(t3.south east) - (0, 2mm)$);
        \draw[->] (t1-outanchor) -- (t3-outanchor);
        \draw[-] (cell1.west) -- (t1.east);
        \draw[-] (cell2.west) -- (t2.east);
        \draw[-] (cell3.west) -- (t3.east);

        \begin{scope}[on background layer]
            \node[
                fit=(state1)(state2)(state3),
                thick,
                draw=cAmain,
                fit margins={left=1.25mm, right=0.75mm},
            ] (statewrap) {};
            \node[
                above=2mm of statewrap,
                anchor=south,
                align=center,
                inner sep=0,
                text depth=0,
                text=cAmain,
                font=\small,
            ] (systext) {\system};
        \end{scope}
    \end{tikzpicture}
    \vspace{-7mm}  
    \caption{\system saves and loads state for exploration.}
    \label{fig:user_apis}
\end{figure}

\input{figures/2_method/architecture}

This section presents interface (\cref{sec:interface}), system architecture (\cref{sec:architecture}), data model (\cref{sec:background-datamodel}), and the equivalence guarantee (\cref{sec:background-equal}) of \system to set the context before delving into technical details.

\subsection{User APIs}
\label{sec:interface}

With \system, users can store and restore states through \texttt{save} and \texttt{load} APIs. When the user invokes \textbf{\texttt{save(namespace: NS) -> TimeID}} on a target namespace, they receive a \texttt{TimeID}.
Later, the user can refer to this \texttt{TimeID} through \textbf{\texttt{load(names: set[str], time_id: TimeID) -> NS}} to restore the namespace.
\cut{
    Optionally, they can also select specific variables by requesting a subset of variable names.
    After the user receives the restored namespace(s), they can perform any query on the namespace as needed.
}
For example, knowing \system, \fix{a user could save} the states of the original \texttt{df:101} (in state@101) preprocessed through a Spark pipeline, an initial model \texttt{model:102} (in state@102), and an alternative model \texttt{model:103} (in state@103). \fix{Later on, the user} is able to quickly load \texttt{df:101} to answer Query1 without re-executing the Spark pipeline. 
\fix{For Query2, the user can} swiftly restore just \texttt{model:102} and \texttt{model:103} without loading the entire namespace.

\cut{
    Beyond correctly saving and loading variables, \system's users can expect reduced storage requirements (Objective 1), faster saving (Objective 2), and faster loading proportional to the requested variables (Objective 3).
    Our experiments (\cref{sec:exp}) measure respective metrics: storage size on file systems in bytes (\cref{sec:exp-storage}), saving time in seconds (\cref{sec:exp-save}), and loading time in seconds (\cref{sec:exp-load}).
}

\subsection{System Architecture}
\label{sec:architecture}

\Cref{fig:architecture} illustrates \system's architecture and interactions. \system reads and writes to the runtime \emph{kernels} (e.g., IPython's kernels for Jupyter). It also instruments namespaces for additional information and access control for safety. \system stores its data in an \emph{underlying storage} (e.g., key-value store).
\system comprises many components, enabling partial saving and loading, optimizing podding decisions, and safely unblocking exploration.

\paragraph{Partial saving and loading} Effective partial saving in \system involves three steps: (1) identify the variables to save (Tracker and Active Variable Filter, \cref{sec:podding-filter}), (2) decompose the variables' dependent objects into multiple pods (Podding, \cref{sec:podding-protocol}), and (3) detect the pods that have changed to be written to storage and link those unchanged with their synonymous pod IDs (Change Detector, \cref{sec:podding-change}). 
Inversely, \system reverts the process to partially load a set of variable names: (1) resolve synonymous pods (if any) and read-only relevant pods from underlying storage (Synonym Resolver, \cref{sec:podding-change}), (2) assemble deserialized data from pods into original objects (Unpodinng, \cref{sec:podding-protocol}), and (3) send back the requested variables.

\paragraph{Intelligent podding optimization} Podding optimization instructs the podding process on how to decompose objects being saved. Given an object and the current podding state, the podding optimizer assesses different courses of action by calculating their anticipated costs and takes the best course of podding action. \cref{sec:podding-optimization} presents the core of the optimizer: \fix{Learned Greedy Algorithm} (\pglga).

\paragraph{Concurrency for safe unblocking} \system unblocks code executions as soon as possible and allows access to variables unrelated to the current saving. After identifying relevant variables being saved, it offloads the rest of the podding process into a \emph{podding thread} (\cref{sec:async-thread}) and protects the relevant variables via \emph{namespace locks} (\cref{sec:async-lock}) while allowing code executions to continue accessing other variables. It also automatically \emph{checks for read-only code executions} (\cref{sec:async-static}) to grant users free access to all variables.
\cut{
    without being blocked.
}




\subsection{Data Model}
\label{sec:background-datamodel}

\cut{
    This section covers \system's conceptual model for capturing all the objects. The model is minimal to capture interdependent objects and useful for discussions on partial saving and loading, object splitting, as well as asynchronous saving.
}
This section covers \system's conceptual model for capturing interdependent objects and useful for discussions on partial saving and loading, object splitting, as well as asynchronous saving.

\paragraph{Namespace, objects, and variables} A \emph{namespace} is a collection of \emph{objects} and their \emph{names}. These named objects are interchangeably called \emph{variables}.
\Cref{fig:datamodel-example} shows an instance of a namespace, 5 variables, and 9 objects. An object $u$ may \emph{(directly) depend} on another object $v$ if and only if object $u$ has a reference to object $v$.
\cut{
    For example, \texttt{fig} depends on \texttt{ax} in \Cref{fig:datamodel-example}.
}

\paragraph{Object graph} \datamodel is a labeled and rooted directed graph $\calG = (\calU, \calE, \calV, \ell)$ where each node $u \in \calU$ represents an object and a directed edge $e = (u, v) \in \calE$ implies that object $u$ directly depends on $v$.
$\calV \subseteq \calU$ is the set of variables (named objects) whereas the variable naming function $\ell$ is a mapping from a string variable name to its variable $u \in \calV$.
\datamodel considers the namespace dictionary object~\footnote{The object returned from IPython kernel's \texttt{globals()}.} the \emph{root} of the graph.


\paragraph{Code execution} A code execution $\text{Exec}$ runs a code block $C$ (e.g., Python statement(s)) on the current state $\calG$ to create a new state: $\text{Exec}(\calG, C) = \calG^{+} = (\calU^{+}, \calE^{+}, \calV^{+}, \ell^{+})$. Depending on the context, code execution is also interchangeably called ``cell execution,'' ``program execution,'' or simply ``execution.''

We presume the \emph{code execution locality}: a code execution can only change the objects and dependencies connected to the variables accessed in the code block.\footnote{Dismissing unsafe pointer arithmetics and dereferencing which are rare in data science programs and libraries.} For example, code execution can change the ``title'' object of a figure by accessing and using any connected variable like \texttt{ax} or \texttt{plt} and \emph{not} by accessing unconnected variables like \texttt{__name__}.
Let $\calV_C$ be the set of accessed variables in $C$, an object $u \in \calU$ does not change (i.e., $u \in \calU^{+}$) if the object $u$ is not connected to any accessed variable $v \in \calV_C$. Similarly, a dependency $e = (u, u') \in \calE$ does not change $e = (u, u') \in \calE^{+}$, if object $u$ or $u'$ is not connected to accessed variables $v \in \calV_C$.

\tikzset{
    vname/.style={%
        draw=cAmain,
        fill=cAlight,
        minimum width=#1,
        minimum height=4mm,
        text height=1.7mm,
        text depth=0mm,
        node distance=0 and 0,
    },
    vname/.default=0.8cm,
    >=LaTeX
}
\tikzset{
    obj/.style={%
        draw,
        shape=circle,
        minimum size=1.3mm,
        node distance=1mm and 2.7mm,
    },
    >=LaTeX
}
\tikzset{
    vobj/.style={%
        obj,
        draw=cAmain,
        fill=cAlight,
    },
    >=LaTeX
}
\tikzset{
    vnameobj/.style={%
        dashed,
        draw=cAmain,
    },
    >=LaTeX
}

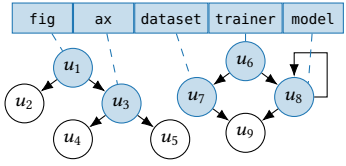
\begin{figure}[t]
    \vspace{\overtopfigurespace}
    \centering
    \footnotesize
    \begin{tikzpicture}    
        \begin{scope}
            \node[vname] (v1) at (0,0) {\texttt{fig}};
            \node[vname, right = of v1] (v2) {\texttt{ax}};
            \node[vname, right = of v2] (v3) {\texttt{dataset}};
            \node[vname, right = of v3] (v4) {\texttt{trainer}};
            \node[vname, right = of v4] (v5) {\texttt{model}};
        \end{scope}
        \node[fit=(v1)(v2)](v12){};

        \begin{scope}
            \node[vobj, below = of v12] (o1) {$u_1$};
            \node[obj, below left = of o1] (o1l) {$u_2$};
            \node[vobj, below right = of o1] (o2) {$u_3$};
            \node[obj, below left = of o2] (o2l) {$u_4$};
            \node[obj, below right = of o2] (o2r) {$u_5$};

            \node[vobj, below = of v4] (o4) {$u_6$};
            \node[vobj, below left = of o4] (o3) {$u_7$};
            \node[vobj, below right = of o4] (o5) {$u_8$};
            \node[obj, below right = of o3] (o35share) {$u_9$};

            \draw[->] (o1) -- (o1l);
            \draw[->] (o1) -- (o2);
            \draw[->] (o2) -- (o2l);
            \draw[->] (o2) -- (o2r);
            
            \draw[->] (o4) -- (o3);
            \draw[->] (o4) -- (o5);
            \draw[->] (o3) -- (o35share);
            \draw[->] (o5) -- (o35share);
            \draw[->] (o5.east) -| ($(o5.north east) + (2.5mm,4mm)$) -| (o5.north);
        \end{scope}

        \draw[vnameobj] (v1) -- (o1);
        \draw[vnameobj] (v2) -- (o2);
        \draw[vnameobj] (v3) -- (o3);
        \draw[vnameobj] (v4) -- (o4);
        \draw[vnameobj] (v5) -- (o5.north east);
    \end{tikzpicture}
    \vspace{-3mm}
    \caption{An instance of \datamodel. Namespace consists of objects $\calU$ (circles) and dependencies $\calE$ (arrows) between them. Each variable naming $\ell$ (blue dash line) connects a variable name (rectangle) to its object $v \in \calV$ (blue circle).}
    \label{fig:datamodel-example}
    \vspace{\undercaptionspace}
\end{figure}

\fix[R2W1,R2D1]{
    \paragraph{Pod dependency graph} As a result of podding, \system produces a \emph{pod dependency graph} or \podgraph $\calG_p = (\calU_p, \calE_p)$ based on \datamodel. Here $\calU_p$ is the set of pods each containing a disjoint partitioning set of objects from \datamodel $\calG = (\calU, \calE, \calV, \ell)$, i.e., $\forall u_p, v_p \in \calU_p, u_p \neq v_p, u_p \cap v_p = \varnothing$ and $\bigcup_{u_p}^{\calU_p} u_p = \calU$. An edge $e_p = (u_p, v_p)$ exists in $\calE_p$ if and only if an object in the source pod $u_p$ has an edge to an object in the destination pod $v_p$, i.e., $\exists u \in u_p, \exists v \in v_p, (u, v) \in \calE$.
    \podgraph is useful for prefetching dependent pod bytes during unpodding, identifying active variables (\cref{sec:podding-filter}), and safely locking variables being saved (\cref{sec:async}).
}

\subsection{Equivalence Guarantee}
\label{sec:background-equal}

\system guarantees to store and load the same objects when their \emph{serialized} representations are \emph{equal}. Serialization $\texttt{Ser}$ transforms an \datamodel $\calG$ into a byte stream $\texttt{Ser}(\calG) = \calB$ that can written to storage. Conversely, \emph{deserialization} $\texttt{Deser}$ reverts the byte stream back to the \datamodel $\texttt{Deser}(\texttt{Ser}(\calG)) = \calG$. \system's guarantee based on serialized representations is consistent if deserialization is \emph{consistent}: deserializing the same byte stream produces the same object $\texttt{Deser}(\calB) = \texttt{Deser}(\calB')$ if $\calB = \calB'$. In Python, Pickle~\cite{pickle} as well as its extensions like Dill~\cite{dill1,dill2} and Cloudpickle~\cite{cloudpickle} are examples of serialization protocols that correctly operate with \datamodel. Our notion of equality is similar to Li et al.~\cite{elasticnotebook2023li} where both value and structural information are considered.

\paragraph{Serialization equality} Two \datamodel $\calG_1$ and $\calG_2$ are equivalent under the \emph{serialization equality} if and only if $\texttt{Ser}(\calG_1) = \texttt{Ser}(\calG_2)$. On the contrary, we define \emph{dirty/changed/modified/mutated} \datamodel $\calG$ as a \datamodel with a different serialized byte stream with respect to its previous state $\calG'$: $\texttt{Ser}(\calG) \neq \texttt{Ser}(\calG')$.





\section{Enabling Partial Saving and Loading}
\label{sec:podding}

\system achieves a partial saving of modified objects and loading of variables through a novel concept called \emph{podding} that groups objects into \emph{pod(s)} (\cref{sec:podding-protocol}). \system detects changes in the serialized bytes of pods over time and only saves changed pods (\cref{sec:podding-change}). To reduce overheads, \system utilizes a novel filtering technique to
identify variables accessed directly and indirectly (\cref{sec:podding-filter}).

\subsection{Podding and Unpodding}
\label{sec:podding-protocol}

\def\pb{1.5mm}
\def\subfigwidth{0.3\linewidth}
    
\tikzset{
    vname/.style={%
        draw,
        minimum width=#1,
        minimum height=5mm,
        text height=1.7mm,
        text depth=0mm,
        node distance=0 and 0,
    },
    vname/.default=1.2cm,
    >=LaTeX
}
\tikzset{
    obj/.style={%
        draw,
        shape=circle,
        minimum size=3mm,
        node distance=2.5mm and 2mm,
    },
    >=LaTeX
}
\tikzset{
    curobj/.style={%
        obj,
        draw=cBmain,
        fill=cBlight,
    },
    >=LaTeX
}
\tikzset{
    invobj/.style={%
        obj,
        draw=none,
        minimum size=2mm,
    },
    >=LaTeX
}
\tikzset{
    pod/.style={%
        thick,
        rounded corners,
        draw=cZmain,
        fill=cZmain,
        fill opacity=0.2,
    },
    >=LaTeX
}
\tikzset{
    podactive/.style={%
        thick,
        rounded corners,
        draw=cAmain,
        fill=cAmain,
        fill opacity=0.2,
    },
    >=LaTeX
}

\begin{figure}[t]
    \vspace{\overtopfigurespace}
    \vspace{-1mm}
    \centering
    \scriptsize
    \begin{subfigure}[b]{\subfigwidth}
    \begin{tikzpicture}
        \begin{scope}
            \node[obj, below = of v12] (o1) {};
            \node[obj, below left = of o1] (o1l) {};
            \node[curobj, below right = of o1] (o2) {};
            \node[invobj, below left = of o2] (o2l) {...};
            \node[invobj, below right = of o2] (o2r) {...};
            \node[invobj, below left = of o1l] (o1ll) {};  

            \draw[->] (o1) -- (o1l);
            \draw[->] (o1) -- (o2);
            \draw[->] (o2) -- (o2l);
            \draw[->] (o2) -- (o2r);
        \end{scope}

        \draw[pod] ($(o1.north)+(0,\pb)$) -- ($(o1l.west)-(\pb,0)$) -- ($(o1l.south)-(0,\pb)$) -- ($(o1.south)-(0,\pb)$) -- ($(o2.south)-(0,\pb)$) -- ($(o2.east)+(\pb,0)$) -- ($(o1.east)+(\pb,0)$) -- cycle;
        \node[curobj] (o2f) at (o2) {};
    \end{tikzpicture}
    \caption{Bundle}
    \end{subfigure}%
    \begin{subfigure}[b]{\subfigwidth}
    \begin{tikzpicture}
        \begin{scope}
            \node[obj, below = of v12] (o1) {};
            \node[obj, below left = of o1] (o1l) {};
            \node[curobj, below right = of o1] (o2) {};
            \node[invobj, below left = of o2] (o2l) {...};
            \node[invobj, below right = of o2] (o2r) {...};
            \node[invobj, below left = of o1l] (o1ll) {};  

            \draw[->] (o1) -- (o1l);
            \draw[->] (o1) -- (o2);
            \draw[->] (o2) -- (o2l);
            \draw[->] (o2) -- (o2r);
        \end{scope}

        \draw[pod] ($(o1.north)+(0,\pb)$) -- ($(o1l.west)-(\pb,0)$) -- ($(o1l.south)-(0,\pb)$) -- ($(o1.east)+(\pb,0)$) -- cycle;
        \draw[podactive] ($(o2.north)+(0,\pb)$) -- ($(o2.west)-(\pb,0)$) -- ($(o2.south)-(0,\pb)$) -- ($(o2.east)+(\pb,0)$) -- cycle;
        \node[curobj] (o2f) at (o2) {};
    \end{tikzpicture}
    \caption{Split-continue}
    \end{subfigure}%
    \begin{subfigure}[b]{\subfigwidth}
    \begin{tikzpicture}
        \begin{scope}
            \node[obj, below = of v12] (o1) {};
            \node[obj, below left = of o1] (o1l) {};
            \node[curobj, below right = of o1] (o2) {};
            \node[invobj, below left = of o2] (o2l) {...};
            \node[invobj, below right = of o2] (o2r) {...};
            \node[invobj, below left = of o1l] (o1ll) {};  

            \draw[->] (o1) -- (o1l);
            \draw[->] (o1) -- (o2);
            \draw[->] (o2) -- (o2l);
            \draw[->] (o2) -- (o2r);
        \end{scope}

        \draw[pod] ($(o1.north)+(0,\pb)$) -- ($(o1l.west)-(\pb,0)$) -- ($(o1l.south)-(0,\pb)$) -- ($(o1.east)+(\pb,0)$) -- cycle;
        \draw[podactive] ($(o2.north)+(0,\pb)$) -- ($(o2l.west)-(\pb,0)$) -- ($(o2l.south)-(0,\pb)$) -- ($(o2.south)-(0,\pb)$) -- ($(o2r.south)-(0,\pb)$) -- ($(o2r.east)+(\pb,0)$) -- node[sloped,near end,anchor=south,text=cAmain,text opacity=1.0] {\scriptsize Automatically bundle} ($(o2.east)+(\pb,0)$) -- cycle;
        \node[curobj] (o2f) at (o2) {};
    \end{tikzpicture}
    \caption{Split-final}
    \end{subfigure}%
    \vspace{-3mm}
    \caption{Different podding decision actions on the same current object (orange circle), parent pod (grey round box), and new pod (blue round box) if any.}
    \label{fig:podding-decision}
    \vspace{\undercaptionspace}
\end{figure}
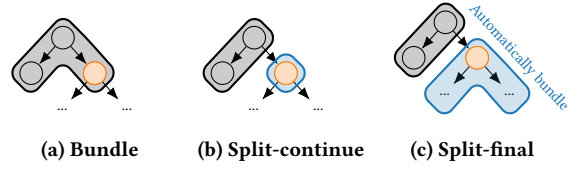

\paragraph{Podding} \system leverages a novel concept named \emph{podding} which partitions \datamodel's objects into multiple pods, each producing its byte stream (a.k.a. \emph{pod bytes}) via a serialization protocol, tagged with a \emph{pod ID}. Each pod ID and pod bytes form a unit of data to be written to a storage. With \datamodel divided into pods, \system can incrementally save the namespace by isolating and writing only changed pods. Furthermore, \system can partially load by reading and unpickling only the necessary pods.

As the serialization protocol traverses through objects in \datamodel in depth-first order,\footnote{By following reduction path depending on the object type, e.g., Pickle's builtin save functions, \texttt{__getinitargs__}, \texttt{__getstate__}, \texttt{__reduce__}, etc.} podding also observes the objects and consults \emph{podding decision} on how to construct pods.
On each object in \datamodel, the podding decision selects either one of three \emph{podding actions} to (1) \textbf{bundle} the object with the current pod, or (2) \textbf{split} the object into a separate pod.
In case of a split, the podding decision also decides whether to (a) \textbf{split-continue} which splits the object and recursively continues applying the podding decision on descendent objects, or (b) \textbf{split-final} which splits the object and skips further podding on descendent objects (\Cref{fig:podding-decision}).
\cut{
    Split-final is equivalent to split-continue the object and bundling the rest of the descendent objects, but with reduced overheads from eliminating further recursions. 
}

\system's efficiency in incremental saving and partial loading hinges on a good podding decision. Na\"ively splitting each object into separate pods introduces tremendous pod management overheads while using one pod for the whole graph reverts podding and disregards partial saving and loading opportunities (see \cref{sec:exp-podoptimize}). We will shortly derive an intelligent optimization to discover the optimal podding decision (\cref{sec:podding-optimization}).


\paragraph{Unpodding} As opposed to podding that splits objects, \emph{unpodding} assembles objects back together during loading.
\system watches for relevant pod IDs in the byte stream. When one is found, \system reads the associated pod bytes from storage and recursively deserializes the pod bytes.
\cut{
    Note that \system memoizes these deserialized pods; \system only deserializes each pod at most once per loading.
}


\paragraph{Base serialization} \system applies to any serialization protocol that supports (explicitly or implicitly) a mechanism to serialize dependent objects to multiple byte streams. Currently, we implement \system to be fully compatible with Pickle~\cite{pickle} as well as its extensions like Dill~\cite{dill1,dill2} and Cloudpickle~\cite{cloudpickle}.

\paragraph{Encoding references within and across pods} \system needs a special protocol to encode shared references across pods because, for Pickle-like serializations, \emph{memo IDs}---numbers denoting object references---must be natural numbers for references within pod and the total number of objects is unknown during serialization.

\cut{
    Analogously to an operating system's virtual memory space,
}
\system implements a \emph{virtual memo space} that denote references for serialized data using \emph{virtual memo IDs} while uniquely associating object references to \emph{global memo IDs}.
\cut{
    (akin to physical memory addresses)
}
In this protocol, virtual memo IDs for within-pod references are the original natural-number memo IDs, while those for across-pod references are their global memo IDs plus $2^{31}$.
To uniquely associate references to global memo IDs, each pod allocates page(s) of $B$ global memo IDs in range $[\delta_i, \delta_i + B)$ dynamically as needed.
As the pod encounters a new object reference, it associates the reference to a global memo ID in the allocated page.
\system persists these page offsets $\{\delta_i\}$ as metadata.
Therefore, given a virtual memo ID, \system can retrieve the object reference stored at the global memo ID calculated by \Cref{eq:global-virtual-memo-id}.
\vspace{\beforeeqspace}
\begin{equation} \label{eq:global-virtual-memo-id}
    \begin{split}
        m_{\text{global}}(m_{\text{virtual}}) = \begin{cases}
            \delta_{i} + r, & \text{if } m_{\text{virtual}} < 2^{31} \\
            m_{\text{virtual}} - 2^{31}, & \text{if } m_{\text{virtual}} \geq 2^{31} \\
        \end{cases} \\
        \text{where } i = \left\lfloor m_{\text{virtual}} \; / \; B \right\rfloor 
        \text{and } r = m_{\text{virtual}} \bmod B \\
    \end{split}
\end{equation}

\subsection{Change Detector and Synonym Resolver}
\label{sec:podding-change}

\system aims to detect changes in pod bytes based on the serialization equality (\cref{sec:background-equal}).
\cut{
    ; that is, we say that the pod has changed if its sequence of bytes differs from before.
}
Changed pods signify new data to write while unchanged pods present an opportunity to save storage.
\cut{
    When \system detects changed pod bytes, it \emph{must} write the pod bytes to underlying storage. On the other hand, when \system detects unchanged (i.e., redundant) pod bytes, it \emph{can} skip writing the pod bytes to save I/O costs. As a result, \system needs a mechanism to quickly detect new pod bytes and manage unchanged pods based on existing ones.
}

\paragraph{Detecting new pods}
To detect new pods, \system maintains a cache-like mapping structure in memory, named \emph{pod thesaurus},
that associates each written pod bytes to its pod ID. Pod IDs that share exactly the same pod bytes are called synonymous pod IDs. Before passing pod bytes to the underlying storage, \system checks whether the pod thesaurus is missing the pod bytes. If so, \system writes the pod bytes to storage, then inserts new pod bytes and pod ID to the pod thesaurus.
On insertion, the pod thesaurus may evict unused pod bytes and pod IDs until it fits within its capacity. We select the last in first out (LIFO) eviction policy for its simplicity.

\paragraph{Skipping and reading redundant pods} Otherwise, if the pod thesaurus already contains the pod bytes (i.e., it was unchanged from a prior write), \system skips writing pod bytes and instead notes down the synonymous pod IDs in the storage. On reading a pod ID, if the pod ID is associated with a synonymous pod ID, \system reads pod bytes of the synonymous pod ID instead.

\paragraph{Thesaurus of hashes for lower memory usage} In place of exact pod bytes, pod thesaurus can store the hashes of pod bytes with their pod IDs. To avoid collision with high probability, the pod thesaurus should hold only a limited number of pods while relying on a large enough number of hash bits. For example, in our experiment (\cref{sec:exp}), we use 128-bit \texttt{xxhash} and set the thesaurus capacity to 1~GB, meaning it can hold 62.5~M (1~GB / 16~B) pods.
Even after 1 billion pods, such a capacity has a collision probability of only $1.8\times10^{-22}$.
\fix[R3W1,R3E2]{In the rare case of a collision, \system would raise an error on load due to missing references and unpickling errors. Given the execution history, one can recover by loading the previous state and re-executing cells. Further safeguards (e.g., combining hashes with checksums or length checks) can reduce collision risk further.}

\subsection{Active Variable Filter}
\label{sec:podding-filter}

Variables not referred to during code execution \fix{presents} an opportunity to reduce change detection overhead further. We call these unreferenced variables \emph{inactive variables} and referenced variables \emph{active variables}. A variable is active when it is connected in \datamodel to an accessed variable in the recent execution $\text{Exec}(\calG, C)$. For example, in \Cref{fig:activevars}, executing \texttt{model.predict(...)} accesses the variable \texttt{model} so \texttt{dataset}, \texttt{trainer}, and \texttt{model} are active variables, while \texttt{fig} and \texttt{ax} are inactive variables based on the \datamodel.

Based on \cref{sec:background-datamodel}, \system can identify inactive and active variables by tracking namespace accesses and expanding from accessed variables to active variables using \podgraph.
\cref{theorem:activeaccessed} shows that this mechanism is sufficient for determining active variables.

\tikzset{
    vname/.style={%
        draw,
        minimum width=#1,
        minimum height=3mm,
        text height=1.7mm,
        text depth=0mm,
        node distance=0 and 0,
    },
    vname/.default=1.2cm,
    >=LaTeX
}
\tikzset{
    obj/.style={%
        draw,
        shape=circle,
        minimum size=2mm,
        node distance=2mm and 4mm,
    },
    >=LaTeX
}
\tikzset{
    vobj/.style={%
        obj,
    },
    >=LaTeX
}
\tikzset{
    vnameobj/.style={%
        dashed,
    },
    >=LaTeX
}
\tikzset{
    pod/.style={%
        rounded corners,
        draw=cZmain,
    },
    >=LaTeX
}

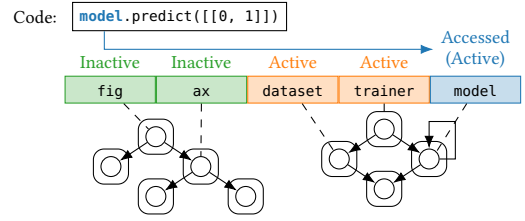
\begin{figure}[t]
    \centering
    \small
    \begin{tikzpicture}[font=\footnotesize]
        \begin{scope}[local bounding box=cellcode]
            \node (cell) {Code:};
            \node[draw, right = 1mm of cell] (code) {\texttt{{\color{cAmain}\textbf{model}}.predict([[0, 1]])}};
        \end{scope}
    
        \begin{scope}[shift={($(cell.south east)+(0,-6mm)$)}, anchor=north west]
            \node[vname, draw=cCmain, fill=cClight] (v1) at (0,0) {\texttt{fig}};
            \node[vname, draw=cCmain, fill=cClight, right = of v1] (v2) {\texttt{ax}};
            \node[vname, draw=cBmain, fill=cBlight, right = of v2] (v3) {\texttt{dataset}};
            \node[vname, draw=cBmain, fill=cBlight, right = of v3] (v4) {\texttt{trainer}};
            \node[vname, draw=cAmain, fill=cAlight, right = of v4] (v5) {\texttt{model}};
            \node[fit=(v1)(v2)](v12){};
    
            \node[cCmain, above=0mm of v1] (v1lock) {Inactive};
            \node[cCmain, above=0mm of v2] (v2lock) {Inactive};
            \node[cBmain, above=0mm of v3] (v3lock) {Active};
            \node[cBmain, above=0mm of v4] (v4lock) {Active};
            \node[cAmain, align=center, above=0mm of v5] (v5lock) {Accessed\\(Active)};
    
            \draw[->, cAmain] ($(code.south west)!0.15!(code.south east)$) |- (v5lock.west) {};
    
            \node[vobj, below = of v12] (o1) {};
            \node[obj, below left = of o1] (o1l) {};
            \node[vobj, below right = of o1] (o2) {};
            \node[obj, below left = of o2] (o2l) {};
            \node[obj, below right = of o2] (o2r) {};
    
            \node[vobj, below = of v4] (o4) {};
            \node[vobj, below left = of o4] (o3) {};
            \node[vobj, below right = of o4] (o5) {};
            \node[obj, below right = of o3] (o35share) {};
    
            \draw[->] (o1) -- (o1l);
            \draw[->] (o1) -- (o2);
            \draw[->] (o2) -- (o2l);
            \draw[->] (o2) -- (o2r);
            
            \draw[->] (o4) -- (o3);
            \draw[->] (o4) -- (o5);
            \draw[->] (o3) -- (o35share);
            \draw[->] (o5) -- (o35share);
            \draw[->] (o5.east) -| ($(o5.north east) + (2.5mm,4mm)$) -| (o5.north);
    
            \draw[vnameobj] (v1) -- (o1);
            \draw[vnameobj] (v2) -- (o2);
            \draw[vnameobj] (v3) -- (o3);
            \draw[vnameobj] (v4) -- (o4);
            \draw[vnameobj] (v5) -- (o5.north east);
            
            \node[pod,fit=(o1)] (o1p) {};
            \node[pod,fit=(o1l)] (o1lp) {};
            \node[pod,fit=(o2)] (o2p) {};
            \node[pod,fit=(o2l)] (o2lp) {};
            \node[pod,fit=(o2r)] (o2rp) {};
            \node[pod,fit=(o3)] (o3p) {};
            \node[pod,fit=(o4)] (o4p) {};
            \node[pod,fit=(o5)] (o5p) {};
            \node[pod,fit=(o35share)] (o35sharep) {};
        \end{scope}
    \end{tikzpicture}
    \vspace{-3mm}
    \caption{Inactive, active, and accessed variables.}
    \label{fig:activevars}
    \vspace{\undercaptionspace}
\end{figure}

\begin{theorem} \label{theorem:activeaccessed}
    All active variables belong to pod(s) that are connected to accessed variables' pod(s) on the prior \podgraph.
\end{theorem}
\begin{proof}

    Let $u^{\text{act}}$ be any active variable, it is connected to an accessed variable $u^{\text{acc}}$ in $\calG$ at a point in time by definition. We use the following lemmas to prove the statement.

    \begin{lemma} \label{lemma:priorconnection}
        By code execution locality (\cref{sec:background-datamodel}), if $u, u' \in \calU$ are connected in $\calG$, $u$ is connected to $\exists \; v \in \calV_C$ in the prior $\calG^{0}$.
    \end{lemma}

    \begin{lemma} \label{lemma:connectednessinpod}
        By pod dependency graph properties (\cref{sec:podding-protocol}), if $u, v \in \calU$ are connected in $\calG$, pods $u_p \ni u$ and $v_p \ni v$ are also connected in $\calG_p$
    \end{lemma}

    While the proofs of lemmas are deferred, key assumptions and properties give hints to derive them. By \Cref{lemma:priorconnection}, $u^{\text{act}}$ is connected to a $v^{\text{acc}} \in \calV_C$ in the prior $\calG^{0}$. By \Cref{lemma:connectednessinpod}, it implies that $u_p \ni u^{\text{act}}$ and $v_p \ni v^{\text{acc}}$are also connected in the prior \podgraph $\calG^{0}_p$.
\end{proof}
\section{Podding Optimization}
\label{sec:podding-optimization}

For each object encountered during serialization, a podding optimizer chooses the best podding action (whether to bundle, split-continue, or split-final the object) to minimize storage cost.
The podding optimizer must be designed as an \emph{online streaming algorithm} to decide podding actions in one serialization pass per saving to \emph{efficiently} and \emph{generalizably} handle large numbers of objects of all types.
This section proposes \system's podding optimizer---Learned Greedy Algorithm (\pglga)---by introducing an abstract objective function (\cref{sec:podding-optimization-problem-abstract}), modeling related uncertainties (\cref{sec:podding-optimization-model}), formulating a podding problem (\cref{sec:podding-optimization-problem}), and presenting the algorithm (\cref{sec:podding-optimization-greedy}).



\subsection{High-level Objective}
\label{sec:podding-optimization-problem-abstract}

\pglga aims to minimize the storage cost by partitioning objects $\calU$ into pods $\calU_p$.
Based on performance profiles, the total podding cost is dominated by (1) overheads per pod $c_{\text{pod}}$, including storage overheads and additional podding time, and (2) rewritten pods due to changed objects:
\vspace{\beforeeqspace}
\begin{equation} \label{eq:lga-cost-prob}
\begin{split}
    \bfL(\calU_p; \calG) &= \sum_{u_p}^{\calU_p} \left[ c_{\text{pod}} + s(u_p) \; \Phi(u_p) \right]
\end{split}
\end{equation}

\noindent
where $c_{\text{pod}}$ is a constant overhead per pod, $s(u_p)$ is the size of the pod $u_p$, and $\Phi(u_p)$ is the random-variable number of times the pod $u_p$ changes (modeled in \cref{sec:podding-optimization-model}). Here we make simplifying assumptions that $s(u_p) = \sum_u^{u_p} s(u)$ and the size of an object $s(u)$
can be measured through a function (e.g., \texttt{sys.getsizeof}~\footnote{This function does not exactly measure but provides a signal sufficiently correlated with the size of the object after serialization.
}). 

\subsection{Composable Volatility Model}
\label{sec:podding-optimization-model}

To support \pglga in weighing between podding actions, \emph{volatility model} predicts object mutations and estimates the mutations of different poddings of objects.
For these purposes, \pglga represents object mutations as a Poisson distribution widely used to model event rates~\cite{Clarke_1946poisson} $\text{Pois}(\lambda(u))$ where the mutations occur $\lambda(u) \leq 1$ times per code execution (i.e., rate of changes). 
This volatility model guides \pglga in separating volatile and non-volatile objects apart, reducing unnecessary pod changes.
Moreover, \pglga can hypothesize the volatility of a pod via Poisson's \emph{composability}
$\lambda(u_p) = \sum_{u}^{u_p} \lambda(u)$.

\fix[R3W1,R3E2]{
    While object changes may exhibit correlations with themselves or across time (e.g., training vs. testing phases), our Poisson assumption is a simplifying choice to minimize model and optimization overheads. To address this limitation, future work may extend the model and the subsequent algorithm to account for correlated or higher-order statistics of object changes and study the resulting effect of improved volatility model on the optimality gap.
}


\fix[R3W1,R3E2]{
    \paragraph{Parameter estimation} To process millions of objects, \pglga estimates the per-object volatility $\lambda(u)$ based on lightweight and broadly applicable feature extraction to arbitrary object types. We initially profiled 14 such features (e.g., ID, scope, type-based functions, functional tests, different measures of size) and measured their importance against tracked object changes on a few notebook samples.\footnote{These notebooks are held out and not benchmarked against in the experiments.}.
    The important features with high average gain~\cite{lightgbmfeatureimportance} used to train the final model are \emph{the object's immediate size, object's length (if any), and the length of the object's \texttt{\_\_dict\_\_}}.
    Note that more sophisticated domain-specific features (e.g., object role or code complexity) are possible extensions, but the current choice is pragmatic for generality.
    After training and validating many models, we find that gradient boosting methods such as XGBoost~\cite{10.1145/2939672.2939785xgboost} and LightGBM~\cite{10.5555/3294996.3295074lightgbm} perform well with compact models and settle with LightGBM for its inference speed in our implementation.
}


\subsection{Objective Function}
\label{sec:podding-optimization-problem}

With the composable volatility model, \pglga substitutes $\Phi(u_p) = \text{Pois}(\lambda(u))$ and optimizes the expectation of the probabilistic cost function, $\calL(\calU_p; \calG) = \E[\bfL(\calU_p; \calG)]$:
\begin{equation} \label{eq:lga-cost}
\begin{split}
    \calL(\calU_p; \calG)
    &= \sum_{u_p}^{\calU_p} \left[ c_{\text{pod}} + \sum_u^{u_p} s(u) \sum_v^{u_p} \lambda(v) \right]
\end{split}
\end{equation}

\paragraph{Hardness} \pglga's podding problem is NP-hard, via equivalences to a \emph{tree partitioning} and subsequently a \emph{supermodular minimization}, a dual to the NP-hard \textit{submodular maximization} problem~\cite{doi:10.1137/090779346maxsubmodnphard}. We defer the proof to our technical report.



\subsection{Learned Greedy Algorithm}
\label{sec:podding-optimization-greedy}

While it is possible to derive \emph{two-pass approximate solution} from an approximation algorithm (e.g., RandomizedUSM~\cite{buchbinder2015tight}) for submodular maximization, the additional serialization pass would double the algorithm computation overhead and the algorithm would increase the implementation complexity.
\system settles on a one-pass greedy solution that discovers good podding solutions more quickly.

\begin{algorithm}[t]
    \caption{\small LGA decision, $\texttt{lga\_action}(u, u_p, \calG_p)$}
    \label{algo:lga}
    \footnotesize
    
    \DontPrintSemicolon
    
    \KwInput{Object $u$, current pod $u_p$, current pod dependency graph $\calG_p$}
    \KwOutput{Podding action $a \in \{ \text{bundle}, \text{split-continue}, \text{split-final} \}$}
    \BlankLine
    $\Delta \calL_{\text{bundle}}
        \leftarrow s(u_p) \times \lambda(u) + s(u) \times (\lambda(u_p) + \lambda(u))$ \\
    $\Delta \calL_{\text{split}}
        \leftarrow c_{\text{pod}} + s(u) \times \lambda(u)$ \\
    $\texttt{pod_depth} \leftarrow \texttt{PodDepth}(u_p, \calG_p)$ \\
    \lIf{$\Delta \calL_{\text{bundle}} < \Delta \calL_{\text{split}}$} {
        \Return $\text{bundle}$
    }
    \lElseIf {$\texttt{pod_depth} < \texttt{MAX\_POD\_DEPTH}$} {
        \Return $\text{split-continue}$
    }
    \lElse {
        \Return $\text{split-final}$
    }
\end{algorithm}

\paragraph{One-pass greedy solution} Our one-pass greedy solution, \pglga, selects the podding action with the locally optimal cost. Given a target object and the current podding state, it measures and compares the additional cost to bundle $\Delta \calL_{\text{bundle}}$ and the additional cost to split $\Delta \calL_{\text{split}}$. If the former is lower, it decides to bundle the object. Otherwise, it decides to split-continue or split-final.

Bundling the object to the current pod makes the pod larger and more volatile. The differences $\Delta \calL_{\text{bundle}}$
can be calculated as:
\begin{equation} \label{eq:delta-bundle}
\begin{split}
    \Delta \calL_{\text{bundle}}
        &= \calL_{\text{bundle}} - \calL(\calU_p; \calG) \\
    \Delta \calL_{\text{bundle}}
        &= [c_{\text{pod}} + (s(u_p) + s(u)) (\lambda(u_p) + \lambda(u))] \\
        &\quad- [c_{\text{pod}} + s(u_p) \lambda(u_p)] \\
    \Delta \calL_{\text{bundle}}
        &= s(u_p) \lambda(u) + s(u) (\lambda(u_p) + \lambda(u))
\end{split}
\end{equation}

On the other hand, split actions maintain the current pod's cost at the expense of an additional pod overhead instead (\cref{eq:delta-split}).
\begin{equation} \label{eq:delta-split}
\begin{split}
    \Delta \calL_{\text{split}}
        &= c_{\text{pod}} + s(u) \lambda(u)
\end{split}
\end{equation}

\pglga selects split-continue if the pod's depth (i.e., distance between the pod and the root pod) is lower than a constant $\texttt{MAX\_POD\_DEPTH}$; otherwise, it selects split-final.
\fix[R1O1,R1D4]{In addition, \pglga memoizes previous podding decisions for each object to regulate the composition across pods and to reduce redundant computation.}


\section{Minimal Blocking for Exploration}
\label{sec:async}

While \system must ensure the consistency of active variables and all their dependent objects during saving, it facilitates the user to continue exploring data in parallel by capitalizing on two key opportunities: when the exploration (1) accesses inactive variables, or (2) only statically reads active variables.
\system concurrently saves in a separate thread (\cref{sec:async-thread}) without locking inactive variables (\cref{sec:async-lock}) and safely permits executions of static code (\cref{sec:async-static}).

\subsection{Podding Thread for Asynchronous Saving}
\label{sec:async-thread}

\system employs thread-based parallelism to concurrently save active variables.\footnote{While standard Python currently executes one thread at a time due to the Global Interpreter Lock (GIL), PEP 703~\cite{pep703}---a formally accepted enhancement---is making GIL optional, thereby enabling thread-based parallelism in the near future.} After identifying active variables, \system initiates a new \emph{podding thread} to proceed with the remaining saving steps, which include optimized podding, detecting changed pods, and writing them to the underlying storage (see \Cref{fig:architecture}).
Currently, \system only allows a single concurrent save at a time; if the next saving is requested before the podding thread finishes, it will have to wait by joining the existing podding thread first.


\cut{
    \paragraph{Considerations for multiple asynchronous saves} \system could potentially allow multiple asynchronous saves simultaneously; however, based on our current empirical observations, doing so would not yield significant benefits but amplify the complexity.\footnote{The implementation is outlined in our experimental branch: \repoexpbranchurl.} For one, \system would need to multiply the locks by the number of asynchronous saves. Additionally, synchronization of ID generations and storage operations would be necessary. Future work is welcome in this direction.
}

\subsection{Active Variable Locking}
\label{sec:async-lock}

Synchronization is necessary to protect active variables from being altered while \system is saving them
to avoid data corruption.
\system could opt to naively lock down the entire namespace and concurrently save active variables. Nonetheless, this approach mostly equivalently blocks subsequent execution
except in rare cases when executions do not access the namespace.

\paragraph{Locking active variables} After identifying accessed variables (\cref{sec:podding-filter}), \system minimally locks \emph{active variables} and leaves inactive variables free for access. 
Instead of applying as many locks as the number of active variables, two simple non-reentrant locks suffice: one locking namespace $l_{\text{ns}}$ to make the namespace thread-safe and another locking active variables $l_{\text{active}}$ held during saving.
\cut{
    When an execution attempts to access a variable, it first checks whether the variable is active through $l_{\text{ns}}$. If so, the execution acquires the active variable lock $l_{\text{active}}$.
    Future works may reduce locking overhead by holding the lock during the execution.
}

\subsection{Safe Static Execution}
\label{sec:async-static}




Apart from accessing inactive variables, users should be able to explore active variables if the exploration only reads those variables statically without modification. Such exploration is said to be a \emph{static execution} over a \emph{static code}.

Rather than laying the burden on users to declare static code accurately, we should automatically and safely check for static codes.
In other languages such as Rust~\cite{rustscoperules,rustownership} and C++~\cite{cppconst}, one can determine the static code through reference annotations; however, this is neither supported nor common practice in dynamically typed languages like Python.
Alternatively, we can instrument every object to watch for object mutations and acquire locks appropriately; nevertheless, such fine-grained instrumentation is excessively expensive to implement. 

\paragraph{Allowlist-based static code checker} \system thus checks for static code using allowlist-based static code checker (ASCC), which utilizes an allowlist of patterns in the abstract syntax tree (AST) that represents static code. Beyond syntactic information, ASCC also utilizes types presented in the namespace at runtime, allowing \system to recognize more static codes.
With an allowlist of rules, ASCC parses out the AST from a given code and traverses through AST nodes. At each node, ASCC searches for a rule in the allowlist that matches the node's subtree, e.g., subtrees corresponding to \texttt{print(x)}, \texttt{df.head()} (where \texttt{df} is a \texttt{DataFrame}), or \texttt{np.mean(arr)}.
If a node does not match with any rule, the code is declared non-static.

While the allowlist is extensible by domain experts like the users (e.g., to include static user-defined functions), \system pre-populates the allowlist with definitely static rules such as printing a string, calculating a summation \texttt{sum(...)}, and \texttt{DataFrame.head()}.
\section{Discussion}
\label{sec:other}

\fix[R1O1,R1D4]{
    This section analyzes the \system's
    correctness (\cref{sec:other-correctness}), optimality (\cref{sec:other-optimality}), and stability (\cref{sec:other-stability}), as well as discussing generalizability to other environments (\cref{sec:other-generalize}) and implementation details (\cref{sec:other-implement}).
}

\cut{
    \subsection{Algorithm Costs}
    \label{sec:other-costs}
    
    As a storage system, \system consumes computation time and storage space. During each saving, \system's podding incurs the most dominant worst-case time complexity $O(|\calU| + |\calE|)$ in compatibility with Pickle protocols. Podding invokes \pglga and volatility model $\lambda$ at most $O(|\calU|)$ times. The active variable filter incurs $O(|\calU| + |\calV| \alpha(|\calU|))$ where $\alpha$ is the slowly growing inverse Ackermann function, following the complexity of union-find data structures~\cite{DBLP:journals/jacm/Tarjan75DisjointSet}.
    Remaining components incur minor computation time costs.
    \system's storage space and I/O usage can be as high as Pickle's in the worst case (when every execution mutates every object) or as low as a serialized size of the initial namespace (when no executions mutate any objects).
    In terms of memory complexity, on top of Pickle's memory usage, \system holds $O(k_{\text{cap}} + |\calV|)$ of data in memory for active variable filter, change detector, and active variable locking where $k_{\text{cap}}$ is the thesaurus capacity.
}

\fix[R1O1,R1D4]{
    \subsection{Correctness Guarantee}
    \label{sec:other-correctness}

    \begin{theorem}
        $\text{Ser}(\text{Unpod}(\text{Pod}(\calG))) = \text{Ser}(\calG)$, establishing that \system’s podding–unpodding preserves the original object graph.
    \end{theorem}
    \begin{proof}
        The serialized representation $\text{Ser}(\calG)$ is consisted of (1) information \emph{within} each pod and (2) information \emph{across} pods. Let $\calG' = \text{Unpod}(\text{Pod}(\calG))$, for the within part, podding produces subgraphs of $\calG$, and during unpodding, each subgraph in $\calG'$ is obtained by deserializing the corresponding stored subgraph from $\calG$; thus each subgraph in $\calG'$ is serialized-equivalent to its counterpart in $\calG$. For the across part, \system records all inter-pod references exactly during podding and replays them during unpodding, ensuring that these references in $\calG'$ match those in $\calG$. Since both the internal content of each subgraph and the inter-subgraph references are identical in $\calG'$ and $\calG$, we have $\text{Ser}(\calG') = \text{Ser}(\calG)$, proving that $\text{Ser}(\text{Unpod}(\text{Pod}(\calG))) = \text{Ser}(\calG).$
    \end{proof}
}

\fix[R1O1,R1D4]{
    \subsection{Approximate Optimality}
    \label{sec:other-optimality}
    
    \begin{theorem}
        Let $\calL^{\pglga}$ be \pglga's cost and $L^{*}$ be the optimal cost, $\calL^{\pglga} \leq \left( 1 + \alpha \right) \calL^{*}$ where $\alpha = \min \left\{ c_{\text{pod}} / (2 \gamma), \sqrt{c_{\text{pod}} / (16 \mu_s \mu_\lambda)} \right\}$ is defined by average sized volatility $\gamma = \frac{1}{|\calU|} \sum_u s(u) \lambda(u)$, average size $\mu_s = \frac{1}{n} \sum_u s(u)$, and average volatility $\mu_\lambda = \frac{1}{|\calU|} \sum_u \lambda(u)$.
    \end{theorem}
    \begin{proof}
        The relative error is derived from the ratio of a lower bound of the optimal cost $\calL^{*}$ and a upper of bound \pglga cost $\calL^{\pglga}$. For a non-empty graph, any podding solution must incur at least $c_{\text{pod}} + \sum_u s(u) \lambda(u) \leq \calL^{*} \leq \calL^{\pglga}$. A upper bound is derived from \cref{lemma:lgatwoapprox}, similarly to DeterministicUSM's approximation~\cite{buchbinder2015tight}. 

        \begin{lemma} \label{lemma:lgatwoapprox}
            \pglga has a worst case cost: $\calL^{\pglga} \leq \frac{1}{2} (\calL^{*} + \calL^{\text{split}})$.
        \end{lemma}
        
        The deferred proof considers the minimization of the submodular dual $f = \calL^{\text{split}} - \calL \geq 0$ where $\calL^{\text{split}} = |\calU| (c_{\text{pod}} + \gamma)$.
    \end{proof}
}

\fix[R1O1,R1D4]{
    \subsection{Podding Stability}
    \label{sec:other-stability}

    We define \emph{podding stability} as the similarity in pod assignments between two consecutive executions of the podding algorithm on overlapping object sets. Let $\calU_1$ and $\calU_2$ be two sets of objects. Given two mappings of podding decisions $A_1: \calU_1 \rightarrow \left\{ \text{split}, \text{bundle}, \text{split-final} \right\}$ and similarly $A_2$, define podding similarity between $A_1$ and $A_2$ as
    \begin{equation}
        \text{Sim}(A_1,A_2) \;=\; \frac{\left|\left\{u \in \calU_1 \cap \calU_2 \;|\; A_1(u) = A_2(u) \right\}\right|}{\left|\calU_1 \cap \calU_2\right|}
    \end{equation}
    
    Under \pglga, the podding is \emph{stable}, $\text{Sim}(A_{i}, A_{i+1}) = 1$.
    Because \pglga memoizes past decisions, it acts as if all objects in $\calU_{i}$ arrive first, followed by all objects in $\calU_{i+1} \setminus \calU_{i}$.
    In this order, \pglga---a one-pass deterministic algorithm---reproduces exactly the same decisions for $\calU_{i} \cap \calU_{i+1} \subseteq \calU_{i}$ in both runs, yielding $\text{Sim}(A^{\pglga}_{i}, A^{\pglga}_{i+1}) = 1$.
}

\fix[R1O3,R1D5]{
    \subsection{Generalizability and Portability}
    \label{sec:other-generalize}

    \system’s abstractions and architecture are designed to generalize beyond the Python runtime and port to other language ecosystems. As discussed in \cref{sec:background-datamodel} and \cref{sec:background-equal}, \system requires two key properties from the host environment: (1) a serialization protocol that can express object structure and shared references, and (2) code execution locality. These requirements are satisfied by many modern runtimes. For example, Java and Rust are natural targets due to their built-in serialization protocols~\cite{javaioSerializable,serdeRust}, well-defined object identity, and controlled memory models. C++ is also compatible in principle, but requires programmer cooperation---specifically, disciplined use of pointers and the adoption of sufficiently generic serialization. 
    Python is an ideal starting point for \system due to its popularity, persistence needs, extensibility, and generic built-in serialization (e.g., Pickle), which made it a natural environment to prototype and evaluate our ideas.
}

\subsection{Implementation}
\label{sec:other-implement}

We implement \system in Python. Our podding inherits from Dill~\cite{dill1,dill2} as the serialization base for its extensive coverage of all objects in our dataset; however, we have also tested \system on Pickle~\cite{pickle} and Cloudpickle~\cite{cloudpickle} with similar results. For this evaluation, we set the thesaurus size to 1~GB (much larger than needed in experiments), $c_{\text{pod}}$ to 1200, and $\texttt{MAX_POD_DEPTH}$ to 3.

We train a LightGBM~\cite{lightgbmlib} as the composable volatility model $\lambda$ on \buildats, \storesfg, and \itsttime by bootstrapping the podding process to collect the mutation object samples. We obtain over 470k object samples where 7.3\% of those are mutating objects.

\begin{table}[t]
  \caption{Real notebooks: five benchmark notebooks for evaluation and three training notebooks for volatility model (\cref{sec:podding-optimization-model}).}
  \label{tab:notebooks}
  \vspace{-2mm}
  \footnotesize
  \begin{tabular}{llrrr}
      \toprule
      \textbf{Notebook} & \textbf{Topic} & \textbf{Dataset} & \textbf{\# Ckpts} & \textbf{\# Objects} \\
      
      \midrule

      \skltweet~\cite{skltweet} & Sentiment analysis & 185~MB & 44 & 155 K \\
      \aicode~\cite{aicode} & EDA codes and comments & 4.2~GB & 21 & 20 M \\
      \agripred~\cite{agripred} & Drought Image classification & 6.8~GB & 15 & 214 \\
      \msciedaw~\cite{msciedaw} & EDA single-cell integration & 27~GB & 28 & 424 K \\
      \ecomsmph~\cite{ecomsmph} & E-commerce data mining & 14~GB & 41 & 35 M \\

      \midrule
      
      \buildats~\cite{buildats} & Asset trading strategy & 303~MB & 42 & \fix{210 K} \\
      \storesfg~\cite{storesfg} & Time series forecasting & 119~MB & 38 & 90 K \\
      \itsttime~\cite{itsttime} & Sport forecasting & 163~MB & 54 & 161 K \\
      \bottomrule

      


      
  \end{tabular}
  \vspace{\undercaptionspace}
\end{table}

\section{Experiments}
\label{sec:exp}

This section empirically evaluates \system and its components on real notebooks \fix{with our reproducibility package published~\footnote{\fix{\podurl}}}. Our experiment pipeline shows that:

\vspace{1mm}
\noindent
\fix{\textbf{(i) End-to-end Performance.} \system offers drastically faster and more compact saving (\cref{sec:exp-storage,sec:exp-save,sec:exp-compression}) and loading target variables (\cref{sec:exp-load}). We further analyzes \system performance on parameterized benchmarks \cref{sec:exp-mutation,sec:exp-scalability} to reveal its versatility.}

\vspace{1mm}
\noindent
\fix{\textbf{(ii) Generalizability and Internals.} \system can handle real-world workloads well due to the combination of its internal components. To study their contributions, we perform ablation studies including those for podding (\cref{exp2:podding}), podding optimization (\cref{sec:exp-podoptimize}), asynchronous saving (\cref{exp4:async} and \cref{exp4:ascc-accuracy}), and storage layer (\cref{sec:exp-graph-storage}).}

\paragraph{Setup} Our testbed is a NUMA machine with 2$\times$ AMD EPYC 7552, 1~TB RAM, and about 800~GB of free disk space (SSD) \fix{to accommodate \texttt{ecomsmph} experiments on inefficient baselines. Nonetheless, all experiments on \system are possible with at least 12~GB of free disk space (not including dataset) and 32~GB of RAM. We} run all experiments under a Ubuntu 22.04 Docker container running Python 3.12 without GIL\footnote{\url{https://github.com/colesbury/nogil-3.12-final}} to enable thread-based parallelism.

\begin{table}[t]
  \caption{Real Python scripts in our dataset with their topics, dataset sizes, and numbers of checkpoints and objects.}
  \label{tab:scripts}
  \vspace{-2mm}
  \footnotesize
  \begin{tabular}{llrrr}
      \toprule
      \textbf{Script} & \textbf{Topic} & \textbf{Dataset} & \textbf{\# Ckpts.} & \textbf{\# Objects} \\
      
      \midrule

      \netmnist~\cite{netmnist} & Digit classification & 116~MB & 43 & 2.7 K \\
      \rlactcri~\cite{rlactcri} & Reinforcement learning & N/A & 256 & 94 K \\
      \vaenet~\cite{vaenet} & Image auto-encoding & 116~MB & 15 & 3.1 K \\
      \tseqpred~\cite{tseqpred} & Time sequence prediction & 24~MB & 46 & 5.2 K \\
      \wordlang~\cite{wordlang} & Language modeling & 12.3~MB & 83 & 35 K \\
      \bottomrule
  \end{tabular}
  \vspace{\undercaptionspace}
\end{table}

\paragraph{Dataset} Our dataset consists of eight real computational notebooks (\Cref{tab:notebooks}) and five Python scripts (\Cref{tab:scripts}). The testing notebook dataset includes \fix{\skltweet, \aicode, \agripred, \msciedaw, and \ecomsmph, ranging from hundreds to tens of millions of objects, from exploratory data analysis (EDA) to image classification.}.
\cut{
    (1) \emph{\skltweet}: a notebook classifying sentiments given natural language tweets using scikit-learn.
    (2) \emph{\aicode}: an exploratory data analysis notebook to analyze cell codes, comments, and notebooks.
    (3) \emph{\agripred}: a drought prediction based on satellite images using a convolution neural network.
    (4) \emph{\msciedaw}: an exploratory notebook over a multimodal single-cell integration dataset.
    And, (5) \emph{\ecomsmph}: a data mining over e-commerce smartphone products.
}
\fix{Three notebooks with relatively fewer objects and smaller state sizes are held out as training notebooks for the volatility model: \buildats, \storesfg, and \itsttime.}
On the other hand, we retrieve real Python scripts \fix{\netmnist, \rlactcri, \vaenet, \tseqpred, and \wordlang} from the PyTorch showcasing repository\footnote{\url{https://github.com/pytorch/examples}} 
for their completeness, conciseness, and diversity in data science topics and insert checkpoints to save all variables in global and local namespaces.
\cut{
    They include
    (1) \emph{\netmnist}: a script training convolutional neural network to classify digit images.
    (2) \emph{\rlactcri}: an actor-critic reinforcement learning agent solving the Cart Pole problem.
    (3) \emph{\vaenet}: a variational auto-encoding~\cite{kingma2013auto} learning to discover an embedding space that describes an image dataset.
    (4) \emph{\tseqpred}: a time sequence prediction using long short-term memory on a synthetic dataset.
    (5) \emph{\wordlang}: a transformer language modeling on Wikitext-2~\cite{merity2016pointer}.
}

\def\allnbs{\skltweet, \aicode, \agripred, \msciedaw, and \ecomsmph}


\paragraph{Baselines} We include several object store baselines to verify the effectiveness of \system and its techniques. \emph{\dill} saves snapshots of the namespace after cell executions serialized by Dill~\cite{dill1,dill2}. On partial loading, it loads the entire snapshot and selects only target variables. As a part of the Python standard library, \emph{\shev}~\cite{shelve} naturally offers persistent namespace through its dictionary interface. To handle loading past states, we save versioned variables as entries \texttt{<tid>:<variable\_name>} in a single \shev database.
\cut{
    Note \shev incorrectly handles shared references across variables, as it duplicates the shared objects upon loading. For example, given \texttt{x = []; y = x}, \shev stores the empty list twice; after restoring through \shev, \texttt{y.append(1)} would incorrectly not append to \texttt{x} as it should. On the other hand, ZODB~\cite{zodb} correctly links shared references across entries~\cite{zodbpersistref}.
}
\fix{For ZODB~\cite{zodb}, two} approaches are valid to enable time traveling: \emph{\zosp} and \emph{\zodb}. In \zosp, we save variables as entries where different versions of variables are stored under separate database paths. Instead of the separation, in \zodb, we store different versions all under one database path and load past variables through the historical connection feature~\cite{zodbhist}.
Finally, \emph{\criu}~\cite{criu} is a middleware checkpoint/restore system. To checkpoint Python namespace, we use \criu to save a forked Python process having the target namespace and then load the persisted process to extract target variables.
\cut{
    We refer to each baseline including our \system as a system under test (SUT).
}
\fix{We limit the maximum storage usage to 768~GB and terminate early when the system exceeds the limit.}

\def\allsuts{\dill, \shev, \zosp, \zodb, \criu, and \system}



\def\subfigwidth{0.49\linewidth}
\def\subfigheight{32mm}

\begin{figure*}[t]
    \centering
    \begin{subfigure}[b]{\linewidth} \centering
        \begin{tikzpicture}
        \begin{axis}[
            ybar,
            ticks=none,
            height=20mm,
            width=\linewidth,
            hide axis,
            xmin=10,  
            xmax=50,
            ymin=0,
            ymax=0.4,
            area legend,
            legend columns=-1,
            legend style={at={(0.0,0.0)},anchor=south,align=center,/tikz/every even column/.append style={column sep=3mm},nodes={scale=0.65, transform shape}},
        ]
            \node[align=center, opacity=1] {
                \addlegendimage{bardill}
                \addlegendentry{\dill}
                \addlegendimage{barshev}
                \addlegendentry{\shev}
                \addlegendimage{barzosp}
                \addlegendentry{\zosp}
                \addlegendimage{barzodb}
                \addlegendentry{\zodb}
                \addlegendimage{barcriu}
                \addlegendentry{\criu}
                \addlegendimage{barpga}
                \addlegendentry{\pga (Ours)}
            };
        \end{axis}
        \end{tikzpicture}
    \end{subfigure}
    \begin{subfigure}[b]{\subfigwidth}
        \begin{tikzpicture}
        \begin{axis}[
            height=\subfigheight,
            width=1.0\linewidth,
            ybar,
            bar width=0.08,
            axis lines=left,
            ylabel=Storage (GB),
            ymajorgrids,
            ymode=log,
            ymin=0.1,
            ymax=1500,
            log origin=infty,
            enlarge x limits=0.15,
            ytick={0.01, 0.1, 1, 10, 100, 1000, 10000},
            xtick={data},
            xticklabels from table={\tableexpi}{nb},
            nodes near coords,
            point meta=explicit symbolic,
            nodes near coords style={font=\scriptsize, align=left, rotate=90, anchor=west},
            legend columns=2,
            legend style={at={(0.03,1.0)},anchor=north west},
            label style={font=\normallabelsize},
            every tick label/.append style={font=\normallabelsize},
        ]
            \addplot[bardill] table [x=index,y=dill_storage_gb] {\tableexpi};
            \addplot[barshev] table [x=index,y=shev_storage_gb] {\tableexpi};
            \addplot[barzosp] table [x=index,y=zosp_storage_gb] {\tableexpi};
            \addplot[barzodb] table [x=index,y=zodb_storage_gb] {\tableexpi};
            \addplot[barcriu] table [x=index,y=criu_storage_gb] {\tableexpi};
            \addplot[barpga] table [x=index,y=pga_storage_gb,meta=pga_rel_storage_gb] {\tableexpi};
        \end{axis}
        \end{tikzpicture}
        \vspace{-4mm}
        \caption{Total storage usage (notebooks)}
        \label{fig:exp_i_storage}
    \end{subfigure}
    \begin{subfigure}[b]{\subfigwidth}
        \begin{tikzpicture}
        \begin{axis}[
            height=\subfigheight,
            width=1.0\linewidth,
            ybar,
            bar width=0.08,
            axis lines=left,
            ylabel=Storage (GB),
            ymajorgrids,
            ymode=log,
            ymin=0.001,
            ymax=150,
            log origin=infty,
            enlarge x limits=0.15,
            ytick={0.001, 0.01, 0.1, 1, 10, 100, 1000, 10000, 100000},
            xtick={data},
            xticklabels from table={\tableexpiscript}{nb},
            nodes near coords,
            point meta=explicit symbolic,
            nodes near coords style={font=\scriptsize, align=left, rotate=90, anchor=west},
            legend columns=2,
            legend style={at={(0.03,1.0)},anchor=north west},
            label style={font=\normallabelsize},
            every tick label/.append style={font=\normallabelsize},
        ]
            \addplot[bardill] table [x=index,y=dill_storage_gb] {\tableexpiscript};
            \addplot[barshev] table [x=index,y=shev_storage_gb] {\tableexpiscript};
            \addplot[barzosp] table [x=index,y=zosp_storage_gb] {\tableexpiscript};
            \addplot[barzodb] table [x=index,y=zodb_storage_gb] {\tableexpiscript};
            \addplot[barcriu] table [x=index,y=criu_storage_gb] {\tableexpiscript};
            \addplot[barpga] table [x=index,y=pga_storage_gb,meta=pga_rel_storage_gb] {\tableexpiscript};
        \end{axis}
        \end{tikzpicture}
        \vspace{-4mm}
        \caption{Total storage usage (scripts)}
        \label{fig:exp_v_storage}
    \end{subfigure}
    \vspace{-3mm}
    \caption{\pga stores all variables with 5.7--36.5$\times$ smaller storage on notebooks and 1.4--29.3$\times$ smaller on scripts than the best baselines. The plots show the total storage required when saving all variables in the namespace at different points in time.}
    \label{fig:exp_storage}
    \vspace{-3mm}
    \vspace{\undercaptionspace}
\end{figure*}
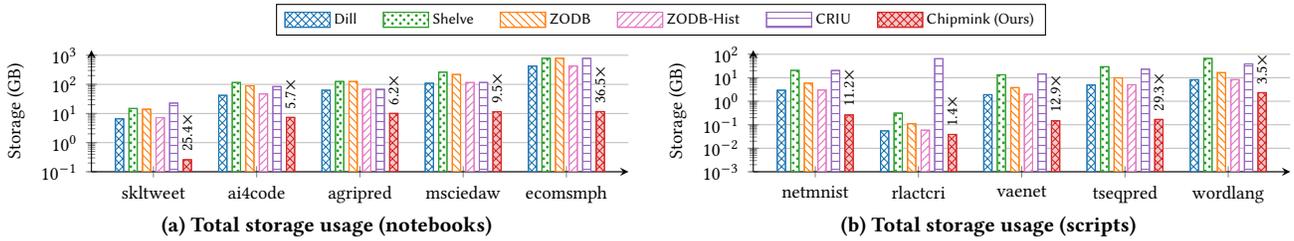

\input{figures/3_experiment/iiii_ecdf}

\def\subfigwidth{1.0\linewidth}
\def\subfigheight{28mm}

\begin{figure}
    \centering
    \begin{tikzpicture}
    \begin{axis}[
        height=\subfigheight,
        width=0.7\linewidth,
        xbar stacked,
        bar width=0.35,
        axis lines=left,
        xlabel=Fraction of Saving Time (\%),
        xlabel shift={-1mm},
        xmajorgrids,
        xminorgrids,
        minor x tick num=1,
        xmin=0.0,
        xmax=105,
        enlarge y limits=0.15,
        xtick={0, 25, 50, 75, 100},
        xticklabels={0\%, 25\%, 50\%, 75\%, 100\%},
        ytick={data},
        yticklabels from table={\tableexpcbd}{nb},
        nodes near coords,
        point meta=explicit symbolic,
        nodes near coords style={align=left, rotate=90, anchor=west},
        area legend,
        legend columns=1,
        legend style={
            at={(1.05, 0.5)},anchor=west,
            /tikz/every even column/.append style={column sep=2mm},
            legend style={nodes={scale=0.65, transform shape}},
        },
        legend cell align={left},
        label style={font=\normallabelsize},
        every tick label/.append style={font=\normallabelsize},
    ]
        \addplot[barone] table [y=index,x=pickle] {\tableexpcbd};
        \addlegendentry{Pickle}

        \addplot[bartwo] table [y=index,x=podding] {\tableexpcbd};
        \addlegendentry{Podding}
        
        \addplot[barthree] table [y=index,x=podding_decision] {\tableexpcbd};
        \addlegendentry{\pglga}

        \addplot[barfour] table [y=index,x=io] {\tableexpcbd};
        \addlegendentry{I/O}
    \end{axis}
    \end{tikzpicture}
    \vspace{-4mm}
    \caption{\fix[R3W2,R3E3]{Stepwise breakdown of \system's saving time.}}
    \label{fig:exp_cbd}
    \vspace{\undercaptionspace}
\end{figure}

\subsection{\system Lowers Storage Requirements}
\label{sec:exp-storage}

This experiment demonstrates whether \system lowers the storage space required to store namespace states.
\cut{
    For notebooks, we alternatively execute a cell and save the namespace on the SUT. For scripts, we execute the script's code with predefined checkpoint statements. Afterward, we measure the \emph{total storage used by the SUT (in bytes)} as Objective 1's metrics.
    We limit the maximum storage usage to 768~GB and terminate early when an SUT exceeds the limit.
}

\paragraph{More compact storage} \system reduces the storage usages to 5.7--36.5$\times$ smaller than the best baseline on notebooks and 1.4--29.3$\times$ smaller on scripts (\Cref{fig:exp_i_storage}).
Most impressively, users can checkpoint \ecomsmph through \system using only 12~GB of disk space while they need to set aside 426~GB of free disk space when using the best baseline (\dill).
A higher storage saving compared to baselines (e.g., in \ecomsmph and \tseqpred) indicate that the notebook/script mutates a small part of the object graph, and vice versa (e.g., in \aicode and \rlactcri).
On the other hand, \shev uses the most space across datasets, because of its poor handling of shared references resulting in both incorrect and duplicate data. For example, in \msciedaw, \shev breaks the reference from \texttt{analyze\_multiome\_x} to \texttt{cell\_summary}, resulting in an inconsistent state after loading where \texttt{analyze\_multiome\_x} does not update \texttt{cell\_summary} as it should.


\fix[R3W2,R3E3]{
    \paragraph{Effect of mutation rate.} 
    We further analyze which types of tasks benefit most from \system by grouping our benchmarks by their estimated mutation rates, i.e., the fraction of objects modified per checkpoint. 
    At very low mutation rates ($<2\%$), such as in \ecomsmph (0.3\%), \tseqpred (1.2\%), and \skltweet (1.7\%), only small batches or short-lived tensors change per step. In these cases, \system achieves the largest gains---up to $36.5\times$ over the best baseline---because pods precisely capture fine-grained dispersed updates that existing storages cannot. 
    At low but broader mutation rates (2--10\%), including \vaenet (4.6\%), \netmnist (6.7\%, and \msciedaw (7.3\%), model weights and partial dataset scans are updated between checkpoints; here, \system still consistently outperforms baselines by one order of magnitude. 
    For medium mutation rates (10--15\%), as in \agripred (10\%), and \aicode (13\%), more of the graph is touched (e.g., convolutional feature maps or recomputed EDA statistics), reducing relative savings but still showing clear efficiency improvements. 
    Finally, at higher mutation rates ($>15\%$), such as in \wordlang (27\%) and \rlactcri (70\%), most objects are updated between checkpoints (language model activations, reinforcement learning buffers), where \system’s advantage shrinks but remains non-trivial (1.4--3.5$\times$). 
    These results demonstrate that \system is effective for ML training and data analysis tasks with dispersed, fine-grained updates, while retaining meaningful efficiency even when large portions of the object graph are modified.
}

\def\subfigheight{30mm}

\begin{figure}[t]
    \centering
    \begin{subfigure}[b]{0.35\linewidth}
        \begin{tikzpicture}
        \begin{axis}[
            height=\subfigheight,
            width=1.0\linewidth,
            ybar,
            bar width=0.1,
            axis lines=left,
            ylabel=Storage (GB),
            ymajorgrids,
            xmin=0.5,
            xmax=1.5,
            ymode=log,
            ymin=1.0,
            ymax=300,
            log origin=infty,
            ytick={0.01, 0.1, 1, 10, 100, 1000, 10000},
            xtick={1},
            xticklabels={},
            nodes near coords,
            point meta=explicit symbolic,
            nodes near coords style={font=\scriptsize, align=left, rotate=90, anchor=west},
            legend columns=4,
            legend style={
                at={(0.03,1.0)},anchor=north west,
                legend style={nodes={scale=0.65, transform shape}},
            },
            label style={font=\normallabelsize},
            every tick label/.append style={font=\normallabelsize},
        ]
            \addplot[barone] table [x=index,y=snp_storage_gb] {\tableexpi};
            \addplot[bartwo] table [x=index,y=snz_storage_gb] {\tableexpi};
            \addplot[barthree] table [x=index,y=snx_storage_gb] {\tableexpi};
            \addplot[barpga] table [x=index,y=pga_storage_gb] {\tableexpi};
            \addplot[barpgaz] table [x=index,y=pgaz_storage_gb] {\tableexpi};
        \end{axis}
        \end{tikzpicture}
        \vspace{-5mm}
        \caption{Total storage size}
        \label{fig:exp_i_compress_storage}
    \end{subfigure}
    \begin{subfigure}[b]{0.63\linewidth}
        \begin{tikzpicture}
        \begin{axis}[
            height=\subfigheight,
            width=0.55\linewidth,
            ybar,
            bar width=0.1,
            axis lines=left,
            ylabel=Save Time (s),
            ymajorgrids,
            xmin=0.5,
            xmax=1.5,
            ymode=log,
            ymin=10.0,
            ymax=3000,
            log origin=infty,
            ytick={0.1, 1, 10, 100, 1000, 10000, 100000},
            xtick={1},
            xticklabels={},
            nodes near coords,
            point meta=explicit symbolic,
            nodes near coords style={font=\scriptsize, align=left, rotate=90, anchor=west},
            area legend,
            legend columns=1,
            legend style={
                at={(1.05,0.5)},anchor=west,
                legend style={nodes={scale=0.65, transform shape}},
            },
            legend cell align={left},
            label style={font=\normallabelsize},
            every tick label/.append style={font=\normallabelsize},
        ]
            \addplot[barone] table [x=index,y=snp_avg_save_s] {\tableexpi};
            \addlegendentry{\snp}
            \addplot[bartwo] table [x=index,y=snz_avg_save_s] {\tableexpi};
            \addlegendentry{\snz}
            \addplot[barthree] table [x=index,y=snx_avg_save_s] {\tableexpi};
            \addlegendentry{\snx}
            \addplot[barpga] table [x=index,y=pga_avg_save_s] {\tableexpi};
            \addlegendentry{\pga}
            \addplot[barpgaz] table [x=index,y=pgaz_avg_save_s] {\tableexpi};
            \addlegendentry{\pgaz}
        \end{axis}
        \end{tikzpicture}
        \vspace{-5mm}
        \caption{Average saving time}
        \label{fig:exp_i_compress_save}
    \end{subfigure}
    \vspace{-3mm}
    \caption{\fix[R2W1,R2D1]{Compression helps but \system remains necessary to efficiently store massive object graphs (e.g., \skltweet).}}
    \label{fig:exp_i_compress}
    \vspace{\undercaptionspace}
\end{figure}

\subsection{\system Saves Objects Quickly}
\label{sec:exp-save}

Next, we test whether \system reduces the perceived saving latency. Here we measure the perceived saving latency by timing all execution delays including the times to synchronously save and acquire locks (if any).
\cut{
    To evaluate Objective 2, we show the \emph{empirical cumulative distribution function (eCDF) of perceived saving times (in seconds)} which signifies the distribution of delays the user would experience on top of normal executions.
}



\input{figures/3_experiment/i_partial_load}

\def\subfigwidth{0.49\linewidth}
\def\subfigheight{30mm}

\pgfplotsset{
    ploadplot/.style={
        height=\subfigheight,
        width=\linewidth,
        axis lines=left,
        xlabel=Mutation Fraction (\%),
        ymajorgrids,
        minor y tick num=1,
        ymin=0.0,
        minor x tick num=1,
        label style={font=\smalllabelsize},
        every tick label/.append style={font=\smalllabelsize},
    },
}
\tikzset{
    ploaddill/.style={%
        draw=dillcolor,
        mark=x,
        mark options={
            draw=dillcolor,
            fill=dillcolorlight,
            fill opacity=0.0,
            scale=0.75
        }
    },
    >=LaTeX
}
\tikzset{
    ploadshev/.style={%
        draw=shevcolor,
        mark=triangle*,
        mark options={
            draw=shevcolor,
            fill=shevcolorlight,
            fill opacity=0.0,
            scale=0.75
        }
    },
    >=LaTeX
}
\tikzset{
    ploadzosp/.style={%
        draw=zospcolor,
        mark=diamond*,
        mark options={
            draw=zospcolor,
            fill=zospcolorlight,
            fill opacity=0.0,
            scale=0.75
        }
    },
    >=LaTeX
}
\tikzset{
    ploadzodb/.style={%
        draw=zodbcolor,
        mark=square*,
        mark options={
            draw=zodbcolor,
            fill=zodbcolorlight,
            fill opacity=0.0,
            scale=0.75
        }
    },
    >=LaTeX
}
\tikzset{
    ploadcriu/.style={%
        draw=criucolor,
        mark=pentagon*,
        mark options={
            draw=criucolor,
            fill=criucolorlight,
            fill opacity=0.0,
            scale=0.75
        }
    },
    >=LaTeX
}
\tikzset{
    ploadpga/.style={%
        thick,
        draw=pgacolor,
        mark=*,
        mark options={
            draw=pgacolor,
            fill=pgacolor,
            scale=0.75
        }
    },
    >=LaTeX
}

\begin{figure}[t]
    \centering
    \begin{subfigure}[b]{\linewidth} \centering
        \begin{tikzpicture}
        \begin{axis}[
            ticks=none,
            height=20mm,
            width=\linewidth,
            hide axis,
            xmin=10,  
            xmax=50,
            ymin=0,
            ymax=0.4,
            legend columns=3,
            legend style={at={(0.0,0.0)},anchor=south,align=center,/tikz/every even column/.append style={column sep=3mm},nodes={scale=0.65, transform shape}},
        ]
            \node[align=center, opacity=1] {
                \addlegendimage{ploaddill}
                \addlegendentry{\dill}
                \addlegendimage{ploadshev}
                \addlegendentry{\shev}
                \addlegendimage{ploadzosp}
                \addlegendentry{\zosp}
                \addlegendimage{ploadzodb}
                \addlegendentry{\zodb}
                \addlegendimage{ploadcriu}
                \addlegendentry{\criu}
                \addlegendimage{ploadpga}
                \addlegendentry{\pga (Ours)}
            };
        \end{axis}
        \end{tikzpicture}
    \end{subfigure}
    \begin{subfigure}[b]{\subfigwidth}
        \begin{tikzpicture}
        \begin{axis}[
            ploadplot,
            ylabel=Storage (GB),
            ymax=32,
        ]
            \addplot[ploaddill] table [x=tml,y=dill_storage_gb] {\tableexpviimr};
            \addplot[ploadshev] table [x=tml,y=shev_storage_gb] {\tableexpviimr};
            \addplot[ploadzosp] table [x=tml,y=zosp_storage_gb] {\tableexpviimr};
            \addplot[ploadzodb] table [x=tml,y=zodb_storage_gb] {\tableexpviimr};
            \addplot[ploadcriu] table [x=tml,y=criu_storage_gb] {\tableexpviimr};
            \addplot[ploadpga] table [x=tml,y=pga_storage_gb] {\tableexpviimr};
            
            \draw[thick, color=black, dotted] 
            (axis cs:6.5199, 0.0) -- (axis cs:6.5199, 25);
            \node[align=left, anchor=south west] (nblabel) at (axis cs:0, 23) {\scriptsize Average mutation rate};
        \end{axis}
        \end{tikzpicture}
        \vspace{-2mm}
        \caption{Total storage size}
        \label{fig:exp_vii_mr_storage}
    \end{subfigure}
    \begin{subfigure}[b]{\subfigwidth}
        \begin{tikzpicture}
        \begin{axis}[
            ploadplot,
            ylabel=Saving Time (s),
            ymax=70,
        ]
            \addplot[ploaddill] table [x=tml,y=dill_avg_save_s] {\tableexpviimr};
            \addplot[ploadshev] table [x=tml,y=shev_avg_save_s] {\tableexpviimr};
            \addplot[ploadzosp] table [x=tml,y=zosp_avg_save_s] {\tableexpviimr};
            \addplot[ploadzodb] table [x=tml,y=zodb_avg_save_s] {\tableexpviimr};
            \addplot[ploadcriu] table [x=tml,y=criu_avg_save_s] {\tableexpviimr};
            \addplot[ploadpga] table [x=tml,y=pga_avg_save_s] {\tableexpviimr};
            
            \draw[thick, color=black, dotted] 
            (axis cs:6.5199, 0.0) -- (axis cs:6.5199, 53);
            \node[align=left, anchor=south west] (nblabel) at (axis cs:0, 50) {\scriptsize Average mutation rate};
        \end{axis}
        \end{tikzpicture}
        \vspace{-2mm}
        \caption{Average saving time}
        \label{fig:exp_vii_mr_save}
    \end{subfigure}
    \vspace{-3mm}
    \caption{Storage and save time when the notebook mutates 1~GB of data over 10 cells at varied rates. The dotted lines display the mutation fraction averaged over 5 real notebooks.}
    \label{fig:exp_vii_mr}
    \vspace{\undercaptionspace}
\end{figure}
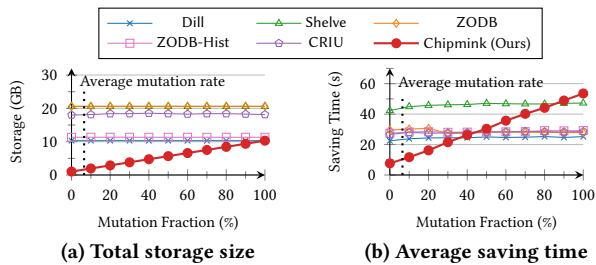

\paragraph{Lower perceived saving latency} eCDF plots (\Cref{fig:exp_iv_ecdf}) describe (1) proportions of saving latencies above/below a length of time, (2) latency percentiles, and (3) total latency speedups. For example, \system users would not perceive any saving latency beyond 23 seconds (75th percentile think time~\cite{DBLP:journals/debu/XinPTWGHJP21thinktime}) on \skltweet, \agripred, and \msciedaw, but would on 7/21 and 4/41 cell executions on \aicode and \ecomsmph respectively.
In all notebooks, half of the cell executions are not blocked by \system more than 1 second, as opposed to the median saving latencies between 3 seconds (\agripred) and 16 minutes (\ecomsmph) of the best baseline.
Measurements on Python scripts reveal similar patterns; however, \system's podding overhead can be observed when the namespace is small (e.g., \netmnist and \vaenet) or when most variables are connected (e.g., \wordlang through its token dictionary).

\fix[R3W2,R3E3]{
    \paragraph{Latency breakdown}
    \Cref{fig:exp_cbd} shows a stepwise breakdown of \system's time cost. 
    The dominant costs arise from Python Pickle’s serialization (e.g., 83.7\% in \ecomsmph, 99.9\% in \agripred) and from \pglga (e.g., 65.3\% in \aicode), where the latter is dominated by evaluating the volatility model $\lambda$ over millions of objects. 
    Other steps such as I/O and index maintenance are negligible ($\leq 5\%$).
    These results suggest clear optimization opportunities: a more lightweight volatility model, parallelized object serialization, or more efficient serializers could further reduce latency.
}

\fix[R2W1,R2D1]{
    \subsection{\system and Byte-Level Compression}
    \label{sec:exp-compression}

    While byte-level delta techniques such as LZ4~\cite{ziv1977universallz77} and Xdelta~\cite{macdonald2000filexdelta} can help reduce storage size, they do not address the inefficiency of storing massive evolving object graphs. As shown in \Cref{fig:exp_i_compress}, \system remains necessary to quickly and compactly capture fine-grained object deltas. Nonetheless, \system is fully compatible with compression: applying LZ4 to pod data further reduces storage footprint. In any case, compression's effect on saving time varies depending on its overhead and I/O reduction.
}

\subsection{\system Loads Target Variables Quickly}
\label{sec:exp-load}

This experiment tests whether \system can load target variables quickly. After saving all states, we randomly select a target cell ID one by one and request loading the variables accessed in the cell.
\cut{
    We measure the \emph{loading time per target cell ID (in seconds)} as Objective 3's metric which signifies the time the user waits to load the state.
}

\paragraph{Faster loading} By avoiding loading irrelevant variables, \system loads target variables with less delay than \dill, \zodb, \zosp, and \criu (\Cref{fig:exp_i_partial_load}) whose loading times depend on namespace sizes rather than requested variable sizes.
\shev performs just as fast; however, users should beware of \shev's broken shared references.


\def\subfigwidth{0.49\linewidth}
\def\subfigheight{29mm}

\pgfplotsset{
    ploadplot/.style={
        height=\subfigheight,
        width=\linewidth,
        axis lines=left,
        ymajorgrids,
        minor y tick num=1,
        minor x tick num=1,
        xlabel shift=-1mm,
        label style={font=\smalllabelsize},
        every tick label/.append style={font=\smalllabelsize},
    },
}
\tikzset{
    ploaddill/.style={%
        draw=dillcolor,
        mark=x,
        mark options={
            draw=dillcolor,
            fill=dillcolorlight,
            fill opacity=0.0,
            scale=0.75
        }
    },
    >=LaTeX
}
\tikzset{
    ploadshev/.style={%
        draw=shevcolor,
        mark=triangle*,
        mark options={
            draw=shevcolor,
            fill=shevcolorlight,
            fill opacity=0.0,
            scale=0.75
        }
    },
    >=LaTeX
}
\tikzset{
    ploadzosp/.style={%
        draw=zospcolor,
        mark=diamond*,
        mark options={
            draw=zospcolor,
            fill=zospcolorlight,
            fill opacity=0.0,
            scale=0.75
        }
    },
    >=LaTeX
}
\tikzset{
    ploadzodb/.style={%
        draw=zodbcolor,
        mark=square*,
        mark options={
            draw=zodbcolor,
            fill=zodbcolorlight,
            fill opacity=0.0,
            scale=0.75
        }
    },
    >=LaTeX
}
\tikzset{
    ploadcriu/.style={%
        draw=criucolor,
        mark=pentagon*,
        mark options={
            draw=criucolor,
            fill=criucolorlight,
            fill opacity=0.0,
            scale=0.75
        }
    },
    >=LaTeX
}
\tikzset{
    ploadexhaust/.style={%
        draw=exhaustcolor,
        mark=+,
        mark options={
            draw=exhaustcolor,
            fill=exhaustcolorlight,
            fill opacity=0.0,
            scale=0.75
        }
    },
    >=LaTeX
}
\tikzset{
    ploadpga/.style={%
        thick,
        draw=pgacolor,
        mark=*,
        mark options={
            draw=pgacolor,
            fill=pgacolor,
            scale=0.75
        }
    },
    >=LaTeX
}

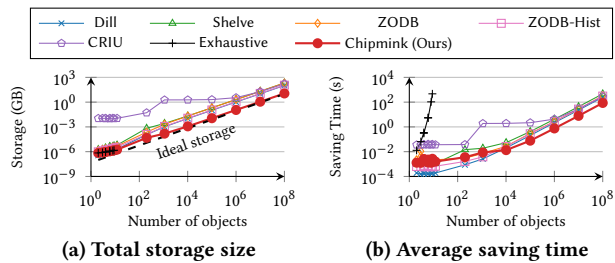
\begin{figure}[t]
    \centering
    \begin{subfigure}[b]{\linewidth} \centering
        \begin{tikzpicture}
        \begin{axis}[
            ticks=none,
            height=20mm,
            width=\linewidth,
            hide axis,
            xmin=10,  
            xmax=50,
            ymin=0,
            ymax=0.4,
            legend columns=4,
            legend style={at={(0.0,0.0)},anchor=south,align=center,/tikz/every even column/.append style={column sep=3mm},nodes={scale=0.65, transform shape}},
        ]
            \node[align=center, opacity=1] {
                \addlegendimage{ploaddill}
                \addlegendentry{\dill}
                \addlegendimage{ploadshev}
                \addlegendentry{\shev}
                \addlegendimage{ploadzosp}
                \addlegendentry{\zosp}
                \addlegendimage{ploadzodb}
                \addlegendentry{\zodb}
                \addlegendimage{ploadcriu}
                \addlegendentry{\criu}
                \addlegendimage{ploadexhaust}
                \addlegendentry{\exhaust}
                \addlegendimage{ploadpga}
                \addlegendentry{\pga (Ours)}
            };
        \end{axis}
        \end{tikzpicture}
    \end{subfigure}
    \begin{subfigure}[b]{\subfigwidth}
        \begin{tikzpicture}
        \begin{axis}[
            ploadplot,
            xlabel=Number of objects,
            ylabel=Storage (GB),
            xmode=log,
            ymode=log,
            xmax=100000000,
            xmin=1,
            ymax=1000,
            ymin=0.000000001,
            xtick={1, 100, 10000, 1000000, 100000000},
            ytick={0.000000001, 0.000001, 0.001, 1.0, 1000.0},
        ]

            
            \addplot[ploaddill] table [x=nsobj,y=dill_storage_gb] {\tableexpvi};
            \addplot[ploadshev] table [x=nsobj,y=shev_storage_gb] {\tableexpvi};
            \addplot[ploadzosp] table [x=nsobj,y=zosp_storage_gb] {\tableexpvi};
            \addplot[ploadzodb] table [x=nsobj,y=zodb_storage_gb] {\tableexpvi};
            \addplot[ploadcriu] table [x=nsobj,y=criu_storage_gb] {\tableexpvi};
            \addplot[ploadpga] table [x=nsobj,y=pga_storage_gb] {\tableexpvi};
            \addplot[ploadexhaust] table [x=nsobj,y=exhaust_storage_gb] {\tableexpvi};
            \draw[thick, dashed, color=cZmain] 
            (axis cs:2, 0.0000001) -- node[below, sloped, text=cZmain, font=\scriptsize] {Ideal storage} (axis cs:100000100, 10);
        \end{axis}
        \end{tikzpicture}
        \vspace{-2mm}
        \caption{Total storage size}
        \label{fig:exp_vi_storage}
    \end{subfigure}
    \begin{subfigure}[b]{\subfigwidth}
        \begin{tikzpicture}
        \begin{axis}[
            ploadplot,
            xlabel=Number of objects,
            ylabel=Saving Time (s),
            xmode=log,
            ymode=log,
            xmax=100000000,
            xmin=1,
            ymax=10000,
            ymin=0.0001,
            xtick={1, 100, 10000, 1000000, 100000000},
            ytick={0.0001, 0.01, 1.0, 100.0, 10000.0},
        ]
            \addplot[ploaddill] table [x=nsobj,y=dill_avg_save_s] {\tableexpvi};
            \addplot[ploadshev] table [x=nsobj,y=shev_avg_save_s] {\tableexpvi};
            \addplot[ploadzosp] table [x=nsobj,y=zosp_avg_save_s] {\tableexpvi};
            \addplot[ploadzodb] table [x=nsobj,y=zodb_avg_save_s] {\tableexpvi};
            \addplot[ploadcriu] table [x=nsobj,y=criu_avg_save_s] {\tableexpvi};
            \addplot[ploadpga] table [x=nsobj,y=pga_avg_save_s] {\tableexpvi};
            \addplot[ploadexhaust] table [x=nsobj,y=exhaust_avg_save_s] {\tableexpvi};
        \end{axis}
        \end{tikzpicture}
        \vspace{-2mm}
        \caption{Average saving time}
        \label{fig:exp_vi_save}
    \end{subfigure}
    \vspace{-3mm}
    \caption{Storage and save time as namespace size scales.}
    \label{fig:exp_vi}
    \vspace{\undercaptionspace}
\end{figure}

\subsection{\system Captures Partial Changes}
\label{sec:exp-mutation}

This experiment examines whether \system's storage usage and saving times grow proportionally to object mutations. It synthetically creates a notebook with 100 lists containing 100k byte strings, each of size 100 bytes (total namespace size is roughly 1~GB). Over the next 9 cells, the notebook mutates a fraction of the lists varied from 0\% to 100\%.
\cut{
    Similarly to \cref{sec:exp-storage} and \cref{sec:exp-save}, we measure the total storage usage (in bytes) and the average saving time (in seconds) as the mutation fraction changes.
}

\paragraph{Closely captured object mutations} As expected, \system's storage usage and saving time match with mutation fractions (\cref{fig:exp_vii_mr}).
As objects mutate more from 0\% to 100\%, \system's storage usage scales from 1~GB (the original namespace size) to 10~GB (all objects always change), converging to snapshotting methods like \dill and \zodb (\Cref{fig:exp_vii_mr_storage})
For this namespace structure, \system's saving time is faster than baselines when about 35\% of objects mutate and is slower than \dill when more objects mutate (\Cref{fig:exp_vii_mr_save}).

\cut{
    \paragraph{Real notebook characteristics}
    Mutation fractions in real notebooks are located in the left ends of \Cref{fig:exp_vii_mr}:
    \skltweet, \aicode, \agripred, \msciedaw, and \ecomsmph mutates 1.7\%, 13\%, 10\%, 7.3\%, 0.31\% of the objects respectively (6.5\% on average).
}

\def\subfigwidth{0.49\linewidth}
\def\subfigheight{32mm}
\def\xprandtwittnet{2.89}
\def\xpnvaicode{0.785}
\def\xpnvecomsmph{3.785}
\def\xprandecomsmph{3.89}
\def\yfailtop{40}
\def\yfailtopsave{300}

\def\idealskltweet{0.21216163699999946}
\def\idealaicode{2.143397773999993}
\def\idealagripred{5.028596495000002}
\def\idealmsciedaw{6.283188761999998}
\def\idealecomsmph{10.941427868000005}

\tikzset{
    idealline/.style={%
        thick,
        color=cZmain,
    },
    >=LaTeX
}

\begin{figure*}
    \centering
    \begin{subfigure}[b]{\linewidth} \centering
        \begin{tikzpicture}
        \begin{axis}[
            ybar,
            ticks=none,
            height=20mm,
            width=\linewidth,
            hide axis,
            xmin=10,  
            xmax=50,
            ymin=0,
            ymax=0.4,
            area legend,
            legend columns=-1,
            legend style={at={(0.0,0.0)},anchor=south,align=center,/tikz/every even column/.append style={column sep=3mm},nodes={scale=0.65, transform shape}},
        ]
            \node[align=center, opacity=1] {
                \addlegendimage{barone}
                \addlegendentry{\bundle}
                \addlegendimage{bartwo}
                \addlegendentry{\pnv}
                \addlegendimage{barthree}
                \addlegendentry{\prand}
                \addlegendimage{barfour}
                \addlegendentry{\pfl}
                \addlegendimage{barfive}
                \addlegendentry{\pgz}
                \addlegendimage{barsix}
                \addlegendentry{\pgi}
                \addlegendimage{barmain}
                \addlegendentry{\pglga (Ours)}
            };
        \end{axis}
        \end{tikzpicture}
    \end{subfigure}
    \begin{subfigure}[b]{\subfigwidth}
        \begin{tikzpicture}
        \begin{axis}[
            height=\subfigheight,
            width=1.0\linewidth,
            ybar,
            bar width=0.057,
            axis lines=left,
            ylabel=Storage (GB),
            ymajorgrids,
            ymode=log,
            ymin=0.1,
            ymax=400,
            log origin=infty,
            enlarge x limits=0.15,
            ytick={0.01, 0.1, 1, 10, 100, 1000},
            xtick={data},
            xticklabels from table={\tableexpi}{nb},
            nodes near coords,
            point meta=explicit symbolic,
            nodes near coords style={align=left, rotate=90, anchor=west},
            legend columns=7,
            legend style={
                at={(0.02,1.0)},anchor=north west,
                legend style={nodes={scale=0.65, transform shape}},
            },
            legend cell align={left},
            label style={font=\normallabelsize},
            every tick label/.append style={font=\normallabelsize},
        ]
            \addplot[barone] table [x=index,y=pnb_storage_gb] {\tableexpi};
    
            \addplot[bartwo] table [x=index,y=pnv_storage_gb] {\tableexpi};
            \addplot[barthree] table [x=index,y=prand_storage_gb] {\tableexpi};
    
            \addplot[barfour] table [x=index,y=pfl_storage_gb] {\tableexpi};
    
            \addplot[barfive] table [x=index,y=pg0_storage_gb] {\tableexpi};
            \addplot[barsix] table [x=index,y=pg1_storage_gb] {\tableexpi};
    
            \addplot[barmain] table [x=index,y=pga_storage_gb] {\tableexpi};


            \draw[thick, dashed, color=cCmain] 
            (axis cs:\xpnvaicode, 0.01) -- (axis cs:\xpnvaicode, \yfailtop);
            \node[cCmain, align=center, rotate=0, anchor=south, font=\footnotesize] at (axis cs:\xpnvaicode, \yfailtop) {Failed};

            \draw[thick, dashed, color=cCmain] 
            (axis cs:\xpnvecomsmph, 0.01) -- (axis cs:\xpnvecomsmph, \yfailtop);
            \node[cCmain, align=right, rotate=0, anchor=south east, inner xsep=0pt, font=\footnotesize] at (axis cs:\xpnvecomsmph, \yfailtop) {Failed};
            \draw[thick, dashed, color=cBmain] 
            (axis cs:\xprandecomsmph, 0.01) -- (axis cs:\xprandecomsmph, \yfailtop);
            \node[cBmain, align=left, rotate=0, anchor=south west, inner xsep=0pt, font=\footnotesize] at (axis cs:\xprandecomsmph, \yfailtop) {Failed};

            \draw[idealline] (axis cs:-0.45, \idealskltweet) -- (axis cs:0.45, \idealskltweet);
            \draw[idealline] (axis cs:0.55, \idealaicode) -- (axis cs:1.45, \idealaicode);
            \draw[idealline] (axis cs:1.55, \idealagripred) -- (axis cs:2.45, \idealagripred);
            \draw[idealline] (axis cs:2.55, \idealmsciedaw) -- (axis cs:3.45, \idealmsciedaw);
            \draw[idealline] (axis cs:3.55, \idealecomsmph) -- (axis cs:4.45, \idealecomsmph);

        \end{axis}
        \end{tikzpicture}
        \vspace{-4mm}
        \caption{Total storage usage}
        \label{fig:exp_iii_storage}
    \end{subfigure}
    \begin{subfigure}[b]{\subfigwidth}
        \centering
        \begin{tikzpicture}
        \begin{axis}[
            height=\subfigheight,
            width=1.0\linewidth,
            ybar,
            bar width=0.057,
            axis lines=left,
            ylabel=Save Time (s),
            ymajorgrids,
            ymode=log,
            ymin=0.1,
            ymax=4000,
            log origin=infty,
            ytick={0.1, 1, 10, 100, 1000, 10000, 100000},
            enlarge x limits=0.15,
            xtick={data},
            xticklabels from table={\tableexpi}{nb},
            legend columns=7,
            legend style={
                at={(0.02,1.0)},anchor=north west,
                legend style={nodes={scale=0.65, transform shape}},
            },
            legend cell align={left},
            unbounded coords=jump,
            label style={font=\normallabelsize},
            every tick label/.append style={font=\normallabelsize},
        ]
            \addplot[barone] table [x=index,y=pnb_avg_save_s] {\tableexpi};
    
            \addplot[bartwo] table [x=index,y=pnv_avg_save_s] {\tableexpi};
            \addplot[barthree] table [x=index,y=prand_avg_save_s] {\tableexpi};
    
            \addplot[barfour] table [x=index,y=pfl_avg_save_s] {\tableexpi};
    
            \addplot[barfive] table [x=index,y=pg0_avg_save_s] {\tableexpi};
            \addplot[barsix] table [x=index,y=pg1_avg_save_s] {\tableexpi};
    
            \addplot[barmain] table [x=index,y=pga_avg_save_s] {\tableexpi};


            \draw[thick, dashed, color=cCmain] 
            (axis cs:\xpnvaicode, 0.01) -- (axis cs:\xpnvaicode, \yfailtopsave);
            \node[cCmain, align=center, rotate=0, anchor=south, font=\footnotesize] at (axis cs:\xpnvaicode, \yfailtopsave) {Failed};

            \draw[thick, dashed, color=cCmain] 
            (axis cs:\xpnvecomsmph, 0.01) -- (axis cs:\xpnvecomsmph, \yfailtopsave);
            \node[cCmain, align=right, rotate=0, anchor=south east, inner xsep=0pt, font=\footnotesize] at (axis cs:\xpnvecomsmph, \yfailtopsave) {Failed};
            \draw[thick, dashed, color=cBmain] 
            (axis cs:\xprandecomsmph, 0.01) -- (axis cs:\xprandecomsmph, \yfailtopsave);
            \node[cBmain, align=left, rotate=0, anchor=south west, inner xsep=0pt, font=\footnotesize] at (axis cs:\xprandecomsmph, \yfailtopsave) {Failed};
        \end{axis}
        \end{tikzpicture}
        \vspace{-4mm}
        \caption{Average saving time}
        \label{fig:exp_iii_save}
    \end{subfigure}
    \vspace{-3mm}
    \caption{\pglga is the most effective podding optimizer in discovering compact podding compared to naive methods (\bundle, \pnv, \prand), manually derived heuristic (\pfl), and \pglga with inaccurate volatility models (\pgz, \pgi). Thick horizontal lines indicate loose theoretical lower bounds of the optimal storage costs.
    The \exhaust baseline is studied in \cref{fig:exp_vi}.
    }
    \label{fig:exp_iii}
    \vspace{-4mm}
    \vspace{\undercaptionspace}
\end{figure*}
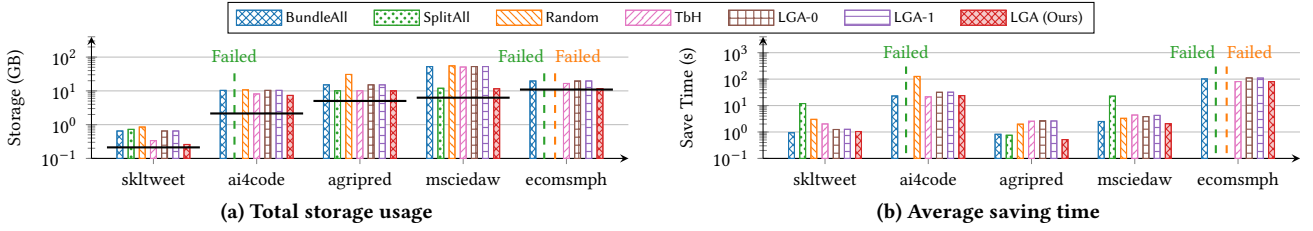

\def\subfigwidth{0.49\linewidth}
\def\subfigheight{30mm}

\begin{figure*}[t]
    \centering
    \begin{subfigure}[b]{0.44\linewidth}
        \begin{tikzpicture}
        \begin{axis}[
            height=\subfigheight,
            width=1.0\linewidth,
            ybar,
            bar width=0.12,
            axis lines=left,
            ylabel=Storage (GB),
            ymajorgrids,
            ymode=log,
            ymin=0.1,
            ymax=1500,
            log origin=infty,
            enlarge x limits=0.15,
            ytick={0.01, 0.1, 1, 10, 100, 1000, 10000},
            xtick={data},
            xticklabels from table={\tableexpii}{nb},
            nodes near coords,
            point meta=explicit symbolic,
            nodes near coords style={font=\scriptsize, align=left, rotate=90, anchor=west},
            legend columns=4,
            legend style={
                at={(0.03,1.0)},anchor=north west,
                legend style={nodes={scale=0.65, transform shape}},
            },
            label style={font=\normallabelsize},
            every tick label/.append style={font=\normallabelsize},
        ]
            \addplot[barzero] table [x=index,y=pgcache0noavf_storage_gb] {\tableexpii};
            \addplot[barone] table [x=index,y=pgnoavf_storage_gb,meta=pgnoavf_rel_storage_gb] {\tableexpii};
            \addplot[barthree] table [x=index,y=pgcache0_storage_gb,meta=pgcache0_rel_storage_gb] {\tableexpii};
            \addplot[barmain] table [x=index,y=pga_storage_gb,meta=pga_rel_storage_gb] {\tableexpii};
        \end{axis}
        \end{tikzpicture}
        \vspace{-4mm}
        \caption{Total storage size}
        \label{fig:exp_ii_avf_storage}
    \end{subfigure}
    \begin{subfigure}[b]{0.55\linewidth}
        \begin{tikzpicture}
        \begin{axis}[
            height=\subfigheight,
            width=0.8\linewidth,
            ybar,
            bar width=0.12,
            axis lines=left,
            ylabel=Save Time (s),
            ymajorgrids,
            ymode=log,
            ymin=0.1,
            ymax=90000,
            log origin=infty,
            enlarge x limits=0.15,
            ytick={0.1, 1, 10, 100, 1000, 10000, 100000},
            xtick={data},
            xticklabels from table={\tableexpii}{nb},
            nodes near coords,
            point meta=explicit symbolic,
            nodes near coords style={font=\scriptsize, align=left, rotate=90, anchor=west},
            area legend,
            legend columns=1,
            legend style={
                at={(1.05,0.5)},anchor=west,
                legend style={nodes={scale=0.65, transform shape}},
            },
            legend cell align={left},
            label style={font=\normallabelsize},
            every tick label/.append style={font=\normallabelsize},
        ]
            \addplot[barzero] table [x=index,y=pgcache0noavf_avg_save_s] {\tableexpii};
            \addlegendentry{\pgcacheznoavf}
            \addplot[barone] table [x=index,y=pgnoavf_avg_save_s,meta=pgnoavf_rel_avg_save_s] {\tableexpii};
            \addlegendentry{\pgnoavf}
            \addplot[barthree] table [x=index,y=pgcache0_avg_save_s,meta=pgcache0_rel_avg_save_s] {\tableexpii};
            \addlegendentry{\pgcachez}
            \addplot[barmain] table [x=index,y=pga_avg_save_s,meta=pga_rel_avg_save_s] {\tableexpii};
            \addlegendentry{\pga}
        \end{axis}
        \end{tikzpicture}
        \vspace{-4mm}
        \caption{Average saving time}
        \label{fig:exp_ii_avf_save}
    \end{subfigure}
    \vspace{-4mm}
    \caption{\fix[R1O2,R1D3]{Change detector (CD) and active variable filter (AVF) both contribute to storage savings and speedups.}}
    \label{fig:exp_ii_avf}
    \vspace{-3mm}
    \vspace{\undercaptionspace}
\end{figure*}
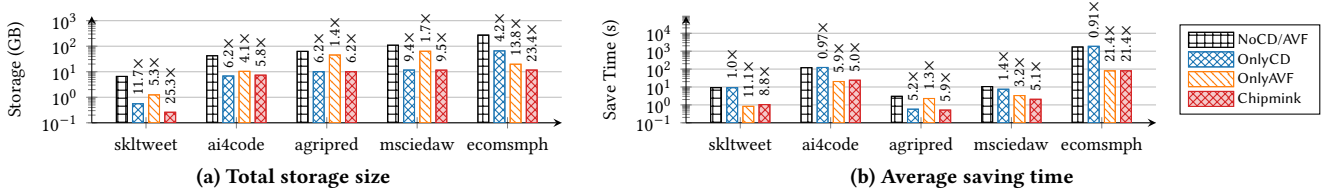

\subsection{\system Scales with Data Sizes}
\label{sec:exp-scalability}

This experiment studies how \system's storage and saving times grow as the number of objects increases.
It creates (1) a \emph{small-scale notebook} with 2 cells that randomly mutate \{1, 2, 3\} lists containing \{1, 2, 3\} byte strings, each of size 100 bytes, and (2) a \emph{normal-scale notebook} with 10 cells that randomly mutate 1\% of 100 lists containing \{1, 10, 100, \dots, 1M\} byte strings, each of size 100 bytes.
\cut{
    Similarly to \cref{sec:exp-storage} and \cref{sec:exp-save}, we measure the total storage usage (in bytes) and the average saving time (in seconds) as the number of objects varies.
}

\paragraph{Scaling with mutation rate} \system scales both performance dimensions linearly with the number of objects (\Cref{fig:exp_vi}).
At a small scale, we measure \system's metadata of around 600 bytes (\Cref{fig:exp_vi_storage}).
As object count increases, \system's storage usage converges to the ideal storage usage as we observe in \cref{sec:exp-mutation}.
In terms of saving time, \system has a 1.3~ms overhead, making it slower than \dill and \zodb at small scale (\Cref{fig:exp_vi_save}).
As object count increases, \system saves quicker than all baselines by leveraging inactive variables.
In this synthetic benchmark (1\% mutation), \system's throughput is around 1.16 million objects per second.

\paragraph{Optimality at small scale} Compared to the exhaustive search (i.e., trying all $2^{\text{decision}}$ combinations of podding decisions), \system finds $>99.19\%$ optimally compact solution made by the exhaustive search (\Cref{fig:exp_vi_storage}). Unfortunately, as podding is an NP problem, we can only afford to test the exhaustive search at a small scale; at 18 decisions, \exhaust takes 1.3 hours to execute (\Cref{fig:exp_vi_save}).

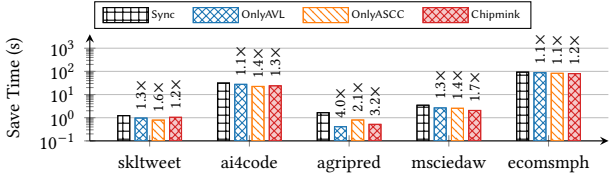
\begin{figure}[t]
    \centering
    \begin{tikzpicture}
    \begin{axis}[
        height=30mm,
        width=1.0\linewidth,
        ybar,
        bar width=0.12,
        axis lines=left,
        ylabel=Save Time (s),
        ymajorgrids,
        minor y tick num=1,
        ymode=log,
        ymin=0.1,
        ymax=4000,
        log origin=infty,
        enlarge x limits=0.15,
        ytick={0.1, 1, 10, 100, 1000},
        xtick={data},
        xticklabels from table={\tableexpiv}{nb},
        nodes near coords,
        point meta=explicit symbolic,
        nodes near coords style={font=\scriptsize, align=left, rotate=90, anchor=west},
        area legend,
        legend columns=4,
        legend style={
            at={(0.43,1.05)},anchor=south,
            /tikz/every even column/.append style={column sep=2mm},
            legend style={nodes={scale=0.5, transform shape}},
        },
        label style={font=\normallabelsize},
        every tick label/.append style={font=\normallabelsize},
    ]
        \addplot[barzero] table [x=index,y=pglnoavlstatic_avg_save_s] {\tableexpiv};
        \addlegendentry{\pglnoavlstatic}
        \addplot[barone] table [x=index,y=pgnostatic_avg_save_s,meta=pgnostatic_rel_avg_save_s] {\tableexpiv};
        \addlegendentry{\pgnostatic}
        \addplot[barthree] table [x=index,y=pgnoavl_avg_save_s,meta=pgnoavl_rel_avg_save_s] {\tableexpiv};
        \addlegendentry{\pgnoavl}
        \addplot[barmain] table [x=index,y=pga_avg_save_s,meta=pga_rel_avg_save_s] {\tableexpiv};
        \addlegendentry{\pga}
    \end{axis}
    \end{tikzpicture}
    \vspace{-6mm}
    \caption{\fix[R1O2,R1D3, R2W2,R2D2]{Asynchronous saving reduces saving time further with active variable locking (AVL) and allowlist-based static code checker (ASCC). \system applies both AVL and ASCC.}}
    \label{fig:exp_iiii_async_save}
    \vspace{\undercaptionspace}
\end{figure}

\subsection{Quality of Internal Podding Optimizer}
\label{sec:exp-podoptimize}

This experiment tests whether \system's podding optimizer \pglga can group objects during podding more compactly and quickly than 6 following alternative podding optimizers. (1) \bundle bundles all objects. (2) \pnv splits objects into their pod. (3) \prand randomly samples actions uniformly. (4) \pfa is a type-based heuristic tuned manually through extensive benchmarks
. (5) \pgz and (6) \pgi refer to the learned greedy algorithm (\pglga) with inaccurate volatility models $\lambda_{\text{\pgz}}(u) = 0$ and $\lambda_{\text{\pgi}}(u) = 1$.
\cut{
    Similarly to \cref{sec:exp-storage} and \cref{sec:exp-save}, we measure the total storage usage (in bytes) and the average saving time (in seconds).
}
Although the podding problem is an NP problem, we gauge the optimal storage requirement with a loose theoretical lower bound: the maximum namespace size of the notebook (in bytes).

\paragraph{Still more compact storage} \pglga consistently produces smaller storage solutions across all notebooks (\Cref{fig:exp_iii_storage}).
On some notebooks, \pnv and \pfa are competitive methods in terms of storage usage.
In addition, the gaps between \system and \pgz/\pgi suggest the importance of the volatility model.
Compared to the loose lower bounds of storage costs, \pglga finds near-optimal podding solutions on \skltweet and \ecomsmph while occupying $3.4\times$, $1.8\times$, and $2.0\times$ larger storage for \aicode, \agripred, and \msciedaw.

\paragraph{Still faster saving} Furthermore, \pglga makes decisions quickly without significant time overhead (\Cref{fig:exp_iii_save}).
\bundle is competitively fast in this regard because it reverts to picking altogether, but it remains slower due to ineffective podding.
Conversely, \pnv and \prand are much slower and sometimes time out since they split objects into too many pods for \system to process.




\begin{table}[t]
    \centering
    \caption{\fix[R2W2,R2D2]{Allowlist-based static code checker (ASCC) varies in its recall and accuracy, but never produces a false negative.}}
    \label{tab:exp_cii}
    \vspace{-2mm}
    \footnotesize
    \begin{tabular}{lrrrr}
        \toprule
        \textbf{Notebooks} & \textbf{Precision} & \textbf{Recall} & \textbf{Accuracy} & \textbf{\#Stable Ser.} \\
        \midrule
        \skltweet & 100\% & 57.1\% & 86.4\% & 38 \\
        \aicode & 100\% & 100\% & 100\% & 7 \\
        \agripred & 100\% & 100\% & 100\% & 8 \\
        \msciedaw & 100\% & 100\% & 100\% & 3 \\
        \ecomsmph & 100\% & 50\% & 91.7\% & 11 \\ 
        \bottomrule
    \end{tabular}
\end{table}

\fix[R1O2,R1D3]{
    \subsection{Impact of CD and AVF}
    \label{exp2:podding}

    This experiment validates design decisions in \cref{sec:podding} where baselines \pgcacheznoavf, \pgnoavf, and \pgcachez accordingly enable and disable change detector (CD) and active variable filter (AVF). 
    By comparing to \pgcacheznoavf, \pgnoavf and \pgcachez saves 4.2--11.7$\times$ and 1.4--13.8$\times$ storage (\Cref{fig:exp_ii_avf_storage}) with no clear winner (e.g., \pgnoavf in \msciedaw, \pgcachez in \ecomsmph), suggesting that both play crucial roles in \system's effectiveness. Change detector helps skipping writing many synonymous pod bytes to storage, while active variable filter removes metadata overheads to manage synonymous pods.
    On the other hand, active variable filter (i.e., \pgcachez and \pga) effectively bypasses object serializations, and thus, is crucial for the reduction in overhead time 2.2--11.1$\times$ faster compared to \pgcacheznoavf. In this aspect, the change detector can only cut down the I/O time, smaller compared the time to serialize the massive number of objects in most notebooks, except \agripred. 
}





\fix[R1O2,R1D3, R2W2,R2D2]{
    \subsection{Impact of Asynchronous Saving}
    \label{exp4:async}

    To study the effectiveness of asynchronous saving techniques, we implement \system without asynchronous saving \pglnoavlstatic, and two ablation baselines \pgnostatic and \pgnoavl as \system with and without active variable locking (AVL) and allowlist-based static code checker (ASCC) described in \cref{sec:async}.
    Despite the quick succession of cell executions in our benchmark, both asynchronous saving via active variable locking (AVL) and allowlist-based static code checker (ASCC) are required to improve \system's distribution of saving times over synchronous saving \pglnoavlstatic (\Cref{fig:exp_iiii_async_save}).
}

\fix[R1O2,R1D3, R2W2,R2D2]{
    \subsection{ASCC Accuracy}
    \label{exp4:ascc-accuracy}

    We evaluate the allowlist-based static code checker (ASCC) across our benchmark notebooks by confirming the equality in serialized bytes (\Cref{tab:exp_cii}), only in cases when pickled bytes are deterministic (noted as numbers of stable serialization). 
    ASCC achieves 100\% precision in all cases, meaning it never produces false positives; this correctness guarantee ensures that any code cells it flags as static are indeed safe to execute with an ongoing saving. 
    However, recall varies across workloads, ranging from 50--100\%. For example, in \skltweet ASCC only identifies 57.1\%, leaving 42.9\% unrecognized for potential optimization.
    Overall, ASCC is conservative by design: it guarantees correctness by avoiding false positives, while false negatives highlight directions for extending the checker to better capture stability patterns in complex notebooks.
}

\fix[R2W3,R2D3]{
    \subsection{Comparison with Graph Storage}
    \label{sec:exp-graph-storage}

    This experiment highlights the importance of podding and our LGA podding algorithm for storage efficiency in graph object stores. Without delta identification, \pnjpickle consumes substantially more space, even though it can occasionally be smaller than \dill, because every save redundantly writes large portions of the object graph. As in our earlier results (\cref{sec:exp-podoptimize}), \system’s delta identification is necessary to significantly reduce storage over the naïve delta identification strategy (\pnjdelta) by grouping co-mutating objects. The choice of backend store makes only a small difference in space: \pganj uses at most 17\% more storage than \system (\pglga with a specialized key–value store). Saving latency is a bigger differentiator: Neo4j is currently slower by up to $104\times$, $59\times$, and $58\times$ for \pnjpickle, \pnjdelta, and \pganj respectively, compared to \dill, due to the cost of updating millions of objects or pods as individual nodes at each save. We expect this gap can be narrowed by improving Neo4j’s batch write interface and by implementing write-optimized data structures.
}

\def\subfigwidth{1.0\linewidth}
\def\subfigheight{30mm}
\def\xpnjpickleecomsmph{3.9}
\def\xpnjdeltaecomsmph{4}
\def\yfailtop{400}
\def\yfailtopsave{3000}

\tikzset{
    idealline/.style={%
        thick,
        color=cZmain,
    },
    >=LaTeX
}

\begin{figure}
    \centering
    \begin{subfigure}[b]{\linewidth} \centering
        \begin{tikzpicture}
        \begin{axis}[
            ybar,
            ticks=none,
            height=20mm,
            width=\linewidth,
            hide axis,
            xmin=10,  
            xmax=50,
            ymin=0,
            ymax=0.4,
            area legend,
            legend columns=-1,
            legend style={at={(0.0,0.0)},anchor=south,align=center,/tikz/every even column/.append style={column sep=2mm},nodes={scale=0.5, transform shape}},
        ]
            \node[align=center, opacity=1] {
                \addlegendimage{barone}
                \addlegendentry{\snp}
                \addlegendimage{bartwo}
                \addlegendentry{\pnjpickle}
                \addlegendimage{barthree}
                \addlegendentry{\pnjdelta}
                \addlegendimage{barfour}
                \addlegendentry{\pganj}
                \addlegendimage{barmain}
                \addlegendentry{\system}
            };
        \end{axis}
        \end{tikzpicture}
    \end{subfigure}
    \begin{subfigure}[b]{\subfigwidth}
        \begin{tikzpicture}
        \begin{axis}[
            height=\subfigheight,
            width=1.0\linewidth,
            ybar,
            bar width=0.057,
            axis lines=left,
            ylabel=Storage (GB),
            ymajorgrids,
            ymode=log,
            ymin=0.1,
            ymax=4000,
            log origin=infty,
            enlarge x limits=0.15,
            ytick={0.01, 0.1, 1, 10, 100, 1000},
            xtick={data},
            xticklabels from table={\tableexpci}{nb},
            nodes near coords,
            point meta=explicit symbolic,
            nodes near coords style={align=left, rotate=90, anchor=west},
            legend columns=7,
            legend style={
                at={(0.02,1.0)},anchor=north west,
                legend style={nodes={scale=0.65, transform shape}},
            },
            legend cell align={left},
            label style={font=\normallabelsize},
            every tick label/.append style={font=\normallabelsize},
        ]
            \addplot[barone] table [x=index,y=snpfsync_storage_gb] {\tableexpci};
    
            \addplot[bartwo] table [x=index,y=pnjpickle_storage_gb] {\tableexpci};
            \addplot[barthree] table [x=index,y=pnjdelta_storage_gb] {\tableexpci};
    
            \addplot[barfour] table [x=index,y=pganj_storage_gb] {\tableexpci};
    
            \addplot[barmain] table [x=index,y=pgafsync_storage_gb] {\tableexpci};


            \draw[thick, dashed, color=cCmain] 
            (axis cs:\xpnjpickleecomsmph, 0.01) -- (axis cs:\xpnjpickleecomsmph, \yfailtop);
            \node[cCmain, align=right, rotate=0, anchor=south east, inner xsep=0pt, font=\footnotesize] at (axis cs:\xpnjpickleecomsmph, \yfailtop) {Failed};
            \draw[thick, dashed, color=cBmain] 
            (axis cs:\xpnjdeltaecomsmph, 0.01) -- (axis cs:\xpnjdeltaecomsmph, \yfailtop);
            \node[cBmain, align=left, rotate=0, anchor=south west, inner xsep=0pt, font=\footnotesize] at (axis cs:\xpnjdeltaecomsmph, \yfailtop) {Failed};
        \end{axis}
        \end{tikzpicture}
        \vspace{-4mm}
    \end{subfigure}
    \vspace{-6mm}
    \caption{\fix[R1O2,R1D3, R2W2,R2D2]{Storage usage of graph object stores under different configurations: without delta identification (\pnjpickle), with naïve delta identification that splits all objects into separate pods (\pnjdelta), and with our \pglga (\pganj).}}
    \label{fig:exp_ci}
    \vspace{\undercaptionspace}
\end{figure}  

\cut{
    \subsection{Access Control Overhead}
    \label{sec:exp-overhead}
    
    Finally, we verify that \system's instrumentation (i.e., active variable tracker and locks) does not incur significant overhead, even if no saving is done.
    Given a notebook, we execute all cells and measure the execution times with and without \system. The execution time differences signify how much execution delay the user would perceive if they install \system but decide not to save any state.
    According to this benchmark, the instrumentation overhead is so small that there was no statistically significant runtime difference in four of five notebooks, i.e., $p \geq 0.05$ on Welch’s t-test. Only for \skltweet, we found a significant 4.6\% overhead with $p = 1.7 \times 10^{-10}$ (\Cref{fig:exp_c_overhead}).
}

\section{Related Works}
\label{sec:related}

This section lists and discusses related works to \system in the context of object storage and persistence for computational notebooks, temporal exploration, and podding optimization.


\paragraph{Python object storages}
Only a few Python object stores correctly persist interdependent objects but none implement a partial object store. \zosp~\cite{zodb,lerner2002forgezope} is a mature object storage that correctly handles object references~\cite{zodbpersistref} with many database features like historical connection and concurrent transactions.
Meanwhile, other object storages only support simplified object types. For example, \shev~\cite{shelve}, redis\_shelve\cite{redisshelve}, and Chest~\cite{chest} (no longer maintained) do not support shared references across entries. pySOS~\cite{pysos} and pickleDB~\cite{pickledb} (despite its name) only support JSON-able objects. Nonetheless, to the best of our knowledge, all of the existing works store snapshots of the objects in their entirety; \system is the first Python object store that partially saves and loads objects, making temporal exploration feasible.

\cut{
    \paragraph{Middleware persistence} Middleware persistence typically follows orthogonal persistence principles, allowing persistence with minimal language modifications and independence from data stores like databases or file systems.
    Examples include CORBA Persistent State Service~\cite{cobraspecification}, EJB's Container Managed Persistence (CMP)~\cite{ejbcmp}, Java Data Objects (JDO)~\cite{apachejdo}, and Pjama~\cite{atkinson2000persistencepjama}. 
    Unlike these systems that focus on improving persistence programming, \system focuses on reducing overheads for explicit persistence to make temporal exploration feasible while achieving high compatibility.
    Spark lets users explicitly declare persistence levels for Resilient Distributed Datasets~\cite{sparkrddpersist} but similarly relies on Pickle for Python objects.
    Many middleware solutions avoid time-consuming serialization by using alternatives like memory paging by bitmap and memory blasting~\cite{soukup1994taming}, file mapping~\cite{mmap}, persistent pointer~\cite{biliris1993makingpersistpointer,aritsugi1995severalpersistpointer}, and quasi-single page~\cite{soukup2014fundamentals}. 
    Most notably, CRIU~\cite{criu,criuincremental} can capture full process states, but lacks the object awareness of \system, leading to higher storage overheads. In persistent memory, middleware solutions support failure-atomicity through methods like universal construction~\cite{correia2020persistentcxpuc} or persistent software transactional memory (PSTM)~\cite{krauter2021persistenthaskell}, with concurrency controls like Atlas's locks~\cite{chakrabarti2014atlas} and Pisces' snapshot isolation~\cite{gu2019pisces}. While these approaches do not focus on temporal exploration, they propose relevant techniques that may improve \system’s asynchronous saving.
}

\cut{
    \def\subfigheight{30mm}

\begin{figure}[t]
    \centering
    \begin{subfigure}[b]{\linewidth} \centering
        \begin{tikzpicture}
        \begin{axis}[
            ybar,
            ticks=none,
            height=20mm,
            width=\linewidth,
            hide axis,
            xmin=10,  
            xmax=50,
            ymin=0,
            ymax=0.4,
            area legend,
            legend columns=-1,
            legend style={at={(0.0,0.0)},anchor=south,align=center,/tikz/every even column/.append style={column sep=3mm},nodes={scale=0.75, transform shape}},
        ]
            \node[align=center, opacity=1] {
                \addlegendimage{barzero}
                \addlegendentry{No instrumentation (Baseline)}
                \addlegendimage{barpga}
                \addlegendentry{w/ \pga}
            };
        \end{axis}
        \end{tikzpicture}
    \end{subfigure}
    \begin{tikzpicture}
    \begin{axis}[
        height=\subfigheight,
        width=1.0\linewidth,
        ybar,
        bar width=0.16,
        axis lines=left,
        ymajorgrids,
        ylabel=Exec. Time (s),
        ymin=0.0,
        ymax=1200,
        log origin=infty,
        enlarge x limits=0.15,
        xtick={data},
        xticklabels from table={\tableexpc}{nb},
        minor y tick num=1,
        nodes near coords,
        point meta=explicit symbolic,
        nodes near coords style={font=\scriptsize, align=left, rotate=0, anchor=south west},
        legend columns=2,
        legend style={at={(0.03,1.0)},anchor=north west},
        label style={font=\normallabelsize},
        every tick label/.append style={font=\normallabelsize},
    ]
        \addplot[barzero, error bars/.cd, y dir=both, y explicit] table [x=index,y=noop,y error=noopstdtwo] {\tableexpc};
        \addplot[barpga, error bars/.cd, y dir=both, y explicit] table [x=index,y=pgnoop,meta=overhead,y error=pgnoopstdtwo] {\tableexpc};
    \end{axis}
    \end{tikzpicture}
    \vspace{-6mm}
    \caption{\pga's access control overhead. Error bars show the two standard deviations of time measurements. Percentages refer to relative overheads over the baseline.}
    \label{fig:exp_c_overhead}
    \vspace{\undercaptionspace}
\end{figure}
}

\paragraph{Object databases}
Object-oriented database management systems (OODBMS) (e.g., \cite{intersysiris,actiannosql,db4o,DBLP:journals/cacm/LambLOW91objectstore,DBLP:conf/sigmod/OrensteinHMS92objectstoreqp,DBLP:journals/cacm/ButterworthOS91Gemstone,DBLP:conf/sigmod/CopelandM84Gemstonesmalltalk}) may seem like natural solutions for an object store.
However, these OODBMS'es mainly focus on relaxing the first normal form (1NF)~\cite{codd1970relational} to manage user-defined complex structures, rather than persisting general execution states with rapid mutations addressed by \system.
As a result, relying on an OODBMS may restrict users to a few object types, require manual work, and incur excessive storage overheads.


\fix[R2W3,R2D3]{
    \paragraph{Graph databases} Graph databases such as Neo4j~\cite{neo4j} and Dgraph~\cite{dgraph} primarily support storage of the current graph state, offering only coarse-grained versioning through transaction logs or application-level timestamping~\cite{angles2016foundations}. Object–graph layers such as Neomodel and Renesca for Neo4j~\cite{neomodel,renesca} and GQLAlchemy for Memgraph~\cite{gqlalchemy} provide higher-level APIs but do not change the underlying storage semantics, and thus inherit the same limitations in fine-grained delta tracking. Similarly, newer systems such as Kùzu Graph Store~\cite{kuzu} are designed for efficient graph queries but still operate at vertex/edge granularity.
    Temporal graph systems extend these designs by recording historical states using either model-based encodings (timestamped nodes/edges) or snapshot-based checkpoints~\cite{campos2016towards,cai2022temporal}, where their coarse-grained storage would incur significant overheads at the scale of fine-grained evolving object graphs. For example, AeonG periodically materializes anchor versions and stores deltas between anchors (a.k.a., snapshots)~\cite{zhao2021aeong}, while ChronoGraph adopts a key–time–value–based storage~\cite{byun2019chronograph}. These techniques require replay of full object changes. By contrast, \system proposes a \emph{delta identification} over evolving in-memory (i.e., out-of-storage) graphs: its podding algorithm pinpoints fine-grained changes across millions of objects, achieving a level of efficiency and precision that coarse-grained existing graph database techniques cannot provide.
}




\section{Conclusion}
\label{sec:conclude}

Through our novel idea of subgraph deltas, \system is an object store with an efficient delta identification for massive and evolving object graphs.
Unlike existing tools, \system accurately, comprehensively, and efficiently identifies deltas,  achieved through several new techniques
    such as podding, virtual memo space, change detector, synonym resolver, active variable filter, learned greedy algorithm for podding optimization, asynchronous saving, active variable locking, and allowlist-based static code checker.
Our empirical experiments show \system is faster and requires less storage space than existing solutions and alternative designs.

\section{Acknowledgments}

This work was supported in part by the National Science Foundation under Awards \#2440498 and \#2312561, and by the National Center for Supercomputing Applications.


\bibliographystyle{ACM-Reference-Format}
\bibliography{_refs}


\appendix
\section{Podding Optimization Analysis}

\subsection{Type-based Heuristic (\pfa)}
\label{sec:podding-optimization-type}

Type-based heuristic (\pfa) decides each podding action solely based on the type of the object in question. It categorizes object types into either one of the three actions to create \emph{type-action catalog}. To make the decision, TbH looks up the object's type in the type-action catalog and acts accordingly. This type-action catalog can be easily configurable by our users to leverage the domain knowledge on object types and mutation patterns.

In our study, we manually analyze computational notebooks and constructs the type-action catalog using the following observations. First, many application types such as NumPy's array, Pandas's data frame, and Matplotlib's figure form a coherent group of objects that would likely mutate together. Secondly, variable-sized immutable types such as string and byte array are stable but their parents are often not. Lastly, other compositional types such as list, dictionary, module, and callable often mutate independently from their parents. For these reasons, we assign application types and variable-sized immutable types into the split-final catagory, and assign other compositional types into the split-continue category. Otherwise, other types fall into bundle category.

\pfa is simple, fast, and effective over the two podding extremes in our experiments. However, it relies on the expertise in specific notebooks and manual efforts to derive the proper type-action catalog. Moreover, it does not work well with unforeseen or newly declared object types. Our next approach not only amends these limitations but also discovers better podding solutions.

\subsection{Proof of Prior \podgraph Connectedness}

\begin{lemma} \label{lemma:priorconnection-2}
    By cell execution locality (\cref{sec:background-datamodel}), if $u, u' \in \calU$ are connected in $\calG$, $u$ is connected to a $v \in \calV_C$ in the prior $\calG^{0}$.
\end{lemma}

\begin{proof}
    This is trivial when $u$ and $u' = v$ were already connected in $\calG^{0}$. Otherwise, if they were not connected in $\calG^{0}$ but are connected in $\calG$, the cell execution (so far) has added some edges $\{(u^{i}, u^{i+1})\}_i$ between the two nodes. By cell execution locality (\cref{sec:background-datamodel}), focusing on the first added edge, it implies that both $u^{1}$ and $u^{2}$ was connected to an accessed variable $v$ in the prior $\calG^{0}$. Since the first added edge $(u^{1}, u^{2})$ is in a path between $u$ and $u'$, $u$ was connected to $u^{1}$; and by transitivity of connectedness, $u$ is also connected to $v$ via $u^{1}$.
\end{proof}

\begin{lemma} \label{lemma:connectednessinpod-2}
    By pod dependency graph properties (\cref{sec:podding-protocol}), if $u, v \in \calU$ are connected in $\calG$, $u_p \ni u$ and $v_p \ni v$ are also connected in $\calG_p$
\end{lemma}

\begin{proof}
    Each consecutive pair of nodes $(u^{i}, u^{i+1})$ in the path $(u, u^1, u^2, \dots, v)$ connecting $u$ and $v$ in $\calG$ either belong to the same pod $u^{i}, u^{i+1} \in u_p$ or different pods with a connected edge by definition of $\calG_p$.
\end{proof}

\subsection{Proof of Hardness}

\pglga's podding problem is NP-hard, via equivalences to a \emph{tree partitioning} and subsequently a \emph{supermodular minimization}, a dual to the NP-hard \textit{submodular maximization} problem~\cite{doi:10.1137/090779346maxsubmodnphard}.

The formulated problem is a graph partitioning problem over the graph structure of the \datamodel $(\calU, \calE)$ (without naming) where a partition is a pod $u_p$ and the cost for each partition is $c_{\text{pod}} + s(u_p) \lambda(u_p)$. More specifically, it is a tree partitioning problem; a podding decides which pod $u_p$ an object $u$ belongs to when it first sees the object $u$. For each vertex $u \in \calU$, there exists exactly one incoming edge $e_f(u) = (u', u) \in \calE, \exists u' \in \calU$ from which the podding visits $u$ first. Respectively, at $u$ from $e_f(u)$, bundle and split actions correspond to keeping the edge and cutting the edge. Let $\calE_f = \{ e_f(u), \forall u \in \calU \}$ be the set of these first-time edges, \pglga's podding problem is a tree partitioning over the graph $G = (\calU, \calE_f)$

Furthermore, the problem is a supermodular minimization. Let $A \subseteq \calE_f$ be split edges. That is, podding splits the objects $u$ and $v$ if $(u, v) = e \in A$; otherwise, it bundles the objects. Denote $\calU_p(A)$ as the resulting pods from splitting edges in $A$. Therefore, we can define a cost function $f: 2^{\calE_f} \rightarrow \bbR^{+}$ such that $f(A) = \calL(\calU_p(A); \calG)$.

\begin{theorem}
    $f$ is supermodular. That is, $\forall A \subseteq B \subseteq \calE_f, e \in \calE_f$,
    \begin{equation} \label{eq:submodular}
        f(B\cup\{e\}) - f(B) \geq f(A\cup\{e\}) - f(A)
    \end{equation}
\end{theorem}



We prove \cref{eq:submodular} as follows. Let an initial configuration of three arbitrary pods be $\calU_p = \{u_{p_1}, u_{p_2}, u_{p_3}\}$ defined by the subset of edges $A$. Then, let $B$ = $A \cup\{f\}$, where $f$ is an edge joining $u_{p_1}$ and $u_{p_2}$, and let $e$ be an edge joining $u_{p_2}$ and $u_{p_3}$. therefore:
\begin{itemize}
    \item $f(A) = \calL(\{u_{p_1}, u_{p_2}, u_{p_3}\}; \calG)$
    \item $f(A \cup\{e\}) = \calL(\{u_{p_1}, \underbrace{u_{p_2}\cup u_{p_3}}_{u_{p_{23}}}\}; \calG)$
    \item $f(B) = f(A \cup\{f\}) = \calL(\{\underbrace{u_{p_1}\cup u_{p_2}}_{u_{p_{12}}}, u_{p_3}\}; \calG)$
    \item $f(B\cup\{e\}) = f(A \cup\{e, f\}) = \calL(\{\underbrace{u_{p_1}\cup u_{p_2}\cup u_{p_3}}_{u_{p_{123}}}\}; \calG)$
\end{itemize}

We can now prove \cref{eq:submodular} by expanding each term on the left-hand side. First $f(B\cup\{e\})$:
\begin{equation}
\begin{split}
    f(B\cup\{e\}) &=\calL(\{u_{p_{123}}\}; \calG) \\
    &= \left(c_{\text{pod}} + s(u_{p_{123}})\lambda(u_{p_{123}})\right) \\
    &= \underbrace{c_{\text{pod}} + s(u_{p_{23}})\lambda(u_{p_{23}}) + s(u_{p_{1}})\lambda(u_{p_{1}})}_{\calL(\{u_{p_1}, u_{p_{23}}\}; \calG) - c_{\text{pod}}} \\
    &\qquad + \underbrace{s(u_{p_{1}})\lambda(u_{p_{23}}) + s(u_{p_{23}})\lambda(u_{p_{1}})}_{\alpha} \\
    f(B\cup\{e\}) &= f(A\cup\{e\}) + \alpha - c_{\text{pod}} \\
\end{split}
\end{equation}

Next is $f(B)$:
\begin{equation}
\begin{split}
    f(B) &=\calL(\{u_{p_{12}}, u_{p_3}\}; \calG) \\
    &= \left(2c_{\text{pod}} + s(u_{p_{12}})\lambda(u_{p_{12}}) + s(u_{p_{3}})\lambda(u_{p_{3}})\right) \\
    &= (\underbrace{c_{\text{pod}} + s(u_{p_{1}})\lambda(u_{p_{1}}) + s(u_{p_{2}})\lambda(u_{p_{2}}) + s(u_{p_{3}})\lambda(u_{p_{3}})}_{\calL(\calU_{p}; \calG) - c_{\text{pod}}} \\ 
    &\qquad + \underbrace{s(u_{p_{1}})\lambda(u_{p_{2}}) + s(u_{p_{2}})\lambda(u_{p_{1}})}_{\beta} \\
    f(B) &= f(A) + \beta - c_{\text{pod}}
\end{split}
\end{equation}

Putting everything together, the left-hand side is equal to:
\begin{equation}
\begin{split}
    f(B\cup\{e\}) - f(B) &= \left(f(A\cup\{e\}) + \alpha - c_{\text{pod}}\right) - \left(f(A) + \beta - c_{\text{pod}}\right) \\
    &= f(A\cup\{e\}) - f(A) + \alpha - \beta \\
    &\geq  f(A\cup\{e\}) - f(A)
\end{split}
\end{equation}
Where assuming non-negative object sizes and change probabilities, the inequality holds due to $\alpha \geq \beta$.

\subsection{Proof of Approximate Optimality}

The proof comprises three parts. We first start from simple-but-coarse bounds (\cref{sec:optimality-crude-bound}) and improve it iteratively (\cref{sec:optimality-iteration}) to arrive at a tighter bound (\cref{sec:optimality-combine}).

\subsubsection{Upper-lower-bound Optimality}
\label{sec:optimality-crude-bound}

Since all partitioning always has to pay at least $\sum_u s(u) \lambda(u)$ and each action increments the additional cost (not counting $s(u) \lambda(u)$) at most $c_{\text{pod}}$ (greedy decision), for $n$ objects, we can say that $\sum_u s(u) \lambda(u) \leq c_{\text{pod}} + \sum_u s(u) \lambda(u) \leq \calL^{*} \leq \calL^{\pglga} \leq \sum_u s(u) \lambda(u) + n c_{\text{pod}}$. So at worst,
\begin{equation}
    \calL^{\pglga} \leq \left( 1 + \frac{(n - 1) c_{\text{pod}}}{c_{\text{pod}} + \sum_u s(u) \lambda(u)} \right) \calL^{*}
\end{equation}

 Considering bundling all, We can also say that $\sum_u s(u) \lambda(u) \leq \calL^{*} \leq \calL^{\pglga} \leq c_{\text{pod}} + \sum_u s(u) \sum_{v} \lambda(v)$. So at worst,
\begin{equation}
    \calL^{\pglga} \leq \left( 1 + \frac{\sum_u s(u) \sum_{v \neq u} \lambda(v)}{c_{\text{pod}} + \sum_u s(u) \lambda(u)} \right) \calL^{*}
\end{equation}

Also, considering fraction podding with arbitrary objects in a pod, optimizing over equal split into $m$ pods. This is better (tighter bound) if object sizes and volatilities are more uniform. If they vary a lot, bound above is tighter.
\begin{equation} \label{eq:fracpodbound}
    2 \sqrt{\frac{\sum_u s(u) \sum_u \lambda(u)}{c_{\text{pod}}}} \leq \calL^{*}
\end{equation}

\subsubsection{Iterative Improvement}
\label{sec:optimality-iteration}

Let $f(X)$ be the podding cost of splitting objects in $X$, so that $f(X_n) = f(OPT_n) = \calL^{\pglga}$ and $f(OPT) = f(OPT_0) = \calL^{*}$. We want to minimize this supermodular $f$.

\begin{lemma} \label{lemma:lgaapprox-step}
    For all $i \in [n]$, $f(OPT_{i-1}) - f(OPT_i) \geq f(X_i) - f(X_{i-1})$.
\end{lemma}
\begin{proof}
    Let $b_i = \calL_{\text{bundle}} - \Delta \calL_{\text{split}} = (s(u_p) \lambda(u) + s(u) \lambda(u_p)) - c_{\text{pod}}$.

    If $b_i < 0$ (bundle), there are two cases to be considered. (A) If $u_i \notin OPT$, then $f(X_i) - f(X_{i-1}) = b_i < 0$ and $OPT_i = OPT_{i-1}$, so $f(OPT_{i-1}) - f(OPT_i) = 0 > f(X_i) - f(X_{i-1})$. (B) If $u_i \in OPT$, then notice that $u^{OPT_i}_p$ (parent pod in $OPT_i$) strictly encompasses $u_p$ (parent pod in $X_i$), implying that the size $s$ and volatility $\lambda$ are larger, and also $u$ may already be in a pod in $OPT_i$, so

    \begin{align}
        f(OPT_{i-1}) - f(OPT_i)
            &\geq (s(u^{OPT_i}_p) \lambda(u) + s(u) \lambda(u^{OPT_i}_p)) - c_{\text{pod}} \\
            &\geq (s(u_p) \lambda(u) + s(u) \lambda(u_p)) - c_{\text{pod}} \\
            &= f(X_i) - f(X_{i-1}) \\
    \end{align}
    
    If $b_i \geq 0$ (split), $X_i = X_{i-1}$ and $OPT_i = OPT_{i-1}$, so $f(OPT_{i-1}) - f(OPT_i) = 0 \geq 0 = f(X_i) - f(X_{i-1})$.
\end{proof}

\subsubsection{Telescoping and Parameterization}
\label{sec:optimality-combine}

Now we can apply the iterative improvement on the initial approximation bound with a tangible parameterization to prove the main theorem.

\begin{theorem}
    Let $\calL^{\pglga}$ be \pglga's cost and $L^{*}$ be the optimal cost, $\calL^{\pglga} \leq \left( 1 + \alpha \right) \calL^{*}$ where $\alpha = \min \left\{ c_{\text{pod}} / (2 \gamma), \sqrt{c_{\text{pod}} / (16 \mu_s \mu_\lambda)} \right\}$ is defined by average sized volatility $\gamma = \frac{1}{|\calU|} \sum_u s(u) \lambda(u)$, average size $\mu_s = \frac{1}{n} \sum_u s(u)$, and average volatility $\mu_\lambda = \frac{1}{|\calU|} \sum_u \lambda(u)$.
\end{theorem}
\begin{proof}
    As a result of \Cref{lemma:lgaapprox-step}, we can bound \pglga's cost using the telescoping series of the steps.
    \begin{align}
        \sum_{i=1}^n f(OPT_{i-1}) - f(OPT_i)
            &\geq \sum_{i=1}^n f(X_i) - f(X_{i-1}) \\
        f(OPT_0) - f(OPT_n)
            &\geq f(X_n) - f(X_0) \\
        f(X_n)
            &\leq \frac{1}{2} \left[ f(OPT_0) + f(X_0) \right] \\
        \calL^{\pglga}
            &\leq \frac{1}{2} \left[ \calL^{*} + \calL^{\text{split}} \right]
    \end{align}

    So we can bound the upper-lower-bound better:
    \begin{equation}
        \calL^{\pglga} \leq \left( 1 + \frac{\frac{n - 1}{2} c_{\text{pod}}}{c_{\text{pod}} + \sum_u s(u) \lambda(u)} \right) \calL^{*}
    \end{equation}
    
    Define $\gamma = \frac{1}{n c_{\text{pod}}} \sum_u s(u) \lambda(u)$.
    \begin{equation}
        \calL^{\pglga} \leq \left( 1 + \frac{n - 1}{2 + 2 \gamma n} \right) \calL^{*}
    \end{equation}
    
    When $n$ is large,
    \begin{equation}
        \lim_{n \rightarrow \infty} \calL^{\pglga} \leq \left( 1 + \frac{1}{2 \gamma} \right) \calL^{*}
    \end{equation}
    
    Which is tight when $\gamma$ is large (that is, when the average sized volatility is large). This implies that podding is harder when objects are tiny and rarely change.
    
    Using \Cref{eq:fracpodbound}, we can bound differently
    \begin{equation}
        \calL^{\pglga} \leq \left( 1 + \frac{n-1}{4} \sqrt{\frac{c_{\text{pod}}}{\sum_u s(u) \sum_u \lambda(u)}} \right) \calL^{*}
    \end{equation}
    
    Define $\overline{s} = \frac{1}{n} \sum_u s(u)$ and $\overline{\lambda} = \frac{1}{n} \sum_u \lambda(u)$.
    \begin{equation}
        \calL^{\pglga} \leq \left( 1 + \frac{n-1}{4n} \sqrt{\frac{c_{\text{pod}}}{\overline{s} \overline{\lambda}}} \right) \calL^{*}
    \end{equation}
    
    When $n$ is large,
    \begin{equation}
        \calL^{\pglga} \leq \left( 1 + \sqrt{\frac{c_{\text{pod}}}{16\overline{s} \overline{\lambda}}} \right) \calL^{*}
    \end{equation}
    
    Which is tight when average size and volatility are large, compared to $c_{\text{pod}}$.
\end{proof}
\section{Additional Experiments}

Here we collect additional experiments that help us understand \system even further.

\subsection{Thesaurus Trade-off}
\label{exp2:thesaurus}

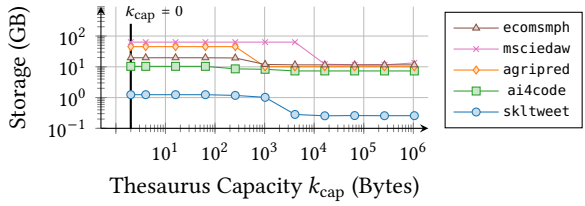
\begin{figure}[t]
    \centering
    \begin{tikzpicture}
    \begin{axis}[
        height=32mm,
        width=0.7\linewidth,
        axis lines=left,
        xlabel=Thesaurus Capacity $k_{\text{cap}}$ (Bytes),
        ylabel=Storage (GB),
        ymajorgrids,
        ymode=log,
        ymin=0.1,
        ymax=900,
        log origin=infty,
        ytick={0.01, 0.1, 1, 10, 100, 1000},
        xmode=log,
        xmin=0.0005,
        xmax=2000,
        xmajorgrids,
        xtick={0.001, 0.01, 0.1, 1, 10, 100, 1000},
        xticklabels={, $10^{1}$, $10^{2}$, $10^{3}$, $10^{4}$, $10^{5}$, $10^{6}$},
        every tick label/.append style={font=\small},
        nodes near coords,
        point meta=explicit symbolic,
        nodes near coords style={align=left, rotate=90, anchor=west},
        legend columns=1,
        legend style={
            at={(1.05,0.0)},anchor=south west,align=left,
            nodes={scale=0.75, transform shape}
        },
        reverse legend,
    ]
        \draw[thick, color=black] 
        (axis cs:0.002, 0.01) -- (axis cs:0.002, 250);
        \node[align=center, anchor=south] (zerolabel) at (0.006, 150) {\scriptsize $k_{\text{cap}} = 0$};


        \addplot[draw=cAmain, mark=*, mark options={draw=cAmain, fill=cAlight, scale=0.75}] table [x=cache_size,y=skltweet_storage_gb] {\tableexpiicache};
        \addlegendentry{\skltweet}
        \addplot[draw=cCmain, mark=square*, mark options={draw=cCmain, fill=cClight, scale=0.75}] table [x=cache_size,y=ai4code_storage_gb] {\tableexpiicache};
        \addlegendentry{\aicode}
        \addplot[draw=cBmain, mark=diamond*, mark options={draw=cBmain, fill=cBlight, scale=0.75}] table [x=cache_size,y=agripred_storage_gb] {\tableexpiicache};
        \addlegendentry{\agripred}
        \addplot[draw=cGmain, mark=x, mark options={draw=cGmain, fill=cGlight, scale=0.75}] table [x=cache_size,y=msciedaw_storage_gb] {\tableexpiicache};
        \addlegendentry{\msciedaw}
        \addplot[draw=cFmain, mark=triangle*, mark options={draw=cFmain, fill=cFlight, scale=0.75}] table [x=cache_size,y=ecomsmph_storage_gb] {\tableexpiicache};
        \addlegendentry{\ecomsmph}
    \end{axis}
    \end{tikzpicture}
    \vspace{-3mm}
    \caption{Higher thesaurus capacity reduces storage usage by detecting more synonymous pods. In these notebooks, storage usage converges with $k_{\text{cap}}$ around 100~KB.}
    \label{fig:exp_ii_cache_storage}
    \vspace{\undercaptionspace}
\end{figure}

\paragraph{Thesaurus capacity vs. recall trade-off} A pod thesaurus with a larger capacity naturally recalls more synonymous pods and so reduces the storage size (\Cref{fig:exp_ii_cache_storage}). Most notebooks in our benchmark only require 4~KB while \msciedaw requires 20~KB for the recall to converge. Note that without any capacity ($k_{\text{cap}} = 0$), \system reverts to \pgcachez. As the capacity increases, the chance of hash collision increases as well; however, even at 1~GB capacity, \system would still be able to handle over 1 billion pods with a collision probability less than $1.8 \times 10^{-22}$.

\subsection{Unblocking Improves Productivity}
\label{exp4:async}

\pgfplotsset{
    ecdfplot/.style={
        height=32mm,
        width=\linewidth,
        axis lines=left,
        xlabel=Saving Time (s),
        ylabel=Percentile,
        ymajorgrids,
        minor y tick num=1,
        ymin=0,
        ymax=110,
        ytick={0, 25, 50, 75, 100},
        xmajorgrids,
    },
}
\tikzset{
    ecdfdill/.style={%
        draw=cAmain,
        mark=x,
        mark options={
            draw=cAmain,
            fill=cAlight,
            fill opacity=0.0,
            scale=1.0
        }
    },
    >=LaTeX
}
\tikzset{
    ecdfpglnoavlstatic/.style={%
        draw=cZmain,
        mark=x,
        mark options={
            draw=cZmain,
            fill=cZlight,
            fill opacity=0.0,
            scale=1.0
        }
    },
    >=LaTeX
}
\tikzset{
    ecdfpgnoavlstatic/.style={%
        draw=cEmain,
        mark=square*,
        mark options={
            draw=cEmain,
            fill=cElight,
            fill opacity=0.0,
            scale=1.0
        }
    },
    >=LaTeX
}
\tikzset{
    ecdfpgnostatic/.style={%
        draw=cAmain,
        mark=diamond*,
        mark options={
            draw=cAmain,
            fill=cAlight,
            fill opacity=0.0,
            scale=1.0
        }
    },
    >=LaTeX
}
\tikzset{
    ecdfpgnoavl/.style={%
        draw=cBmain,
        mark=triangle*,
        mark options={
            draw=cBmain,
            fill=cBlight,
            fill opacity=0.0,
            scale=1.0
        }
    },
    >=LaTeX
}
\tikzset{
    ecdfpga/.style={%
        thick,
        draw=pgacolor,
        mark=*,
        mark options={
            draw=pgacolor,
            fill=pgacolorlight,
            fill opacity=0.0,
            scale=1.0
        }
    },
    >=LaTeX
}

\def\subfigwidthwide{0.49\linewidth}
\def\subfigwidth{0.45\linewidth}

\begin{figure}[t]
    \centering
    \begin{subfigure}[b]{\linewidth} \centering
        \begin{tikzpicture}
        \begin{axis}[
            ticks=none,
            height=20mm,
            width=\linewidth,
            hide axis,
            xmin=10,  
            xmax=50,
            ymin=0,
            ymax=0.4,
            legend columns=-1,
            legend style={at={(0.0,0.0)},anchor=south,align=center,/tikz/every even column/.append style={column sep=3mm},nodes={scale=0.75, transform shape}},
        ]
            \node[align=center, opacity=1] {
                \addlegendimage{ecdfpglnoavlstatic}
                \addlegendentry{\pglnoavlstatic}
                \addlegendimage{ecdfpgnostatic}
                \addlegendentry{\pgnostatic}
                \addlegendimage{ecdfpgnoavl}
                \addlegendentry{\pgnoavl}
                \addlegendimage{ecdfpga}
                \addlegendentry{\pga}
            };
        \end{axis}
        \end{tikzpicture}
    \end{subfigure}
    \begin{subfigure}[b]{\subfigwidth}
        \begin{tikzpicture}
        \begin{axis}[
            ecdfplot, 
            xmin=0,
            xmajorgrids,
            minor x tick num=1,
            ymajorgrids,
            minor y tick num=1,
            ytick={0, 20, 40, 60, 80, 100},
            legend columns=2,
            legend style={
                at={(0.95,0.05)},anchor=south east,align=center,
                legend style={nodes={scale=0.75, transform shape}},
            },
            label style={font=\normallabelsize},
            every tick label/.append style={font=\normallabelsize},
        ]
            \addplot[ecdfpglnoavlstatic] table [x=save_s,y=ecdf] {\tableexpivecdfaicodepglnoavlstatic};
            \addplot[ecdfpgnostatic] table [x=save_s,y=ecdf] {\tableexpivecdfaicodepgnostatic};
            \addplot[ecdfpgnoavl] table [x=save_s,y=ecdf] {\tableexpivecdfaicodepgnoavl};
            \addplot[ecdfpga] table [x=save_s,y=ecdf] {\tableexpivecdfaicodepga};

        \end{axis}
        \end{tikzpicture}
        \vspace{-2mm}
        \caption{\aicode}
        \label{fig:exp_iv_ecdf_compare_aicode}
    \end{subfigure}
    \begin{subfigure}[b]{\subfigwidth}
        \begin{tikzpicture}
        \begin{axis}[
            ecdfplot, 
            xmin=0,
            xmax=10,
            xmajorgrids,
            minor x tick num=1,
            xtick={0, 2, 4, 6, 8, 10, 12},
            ymajorgrids,
            minor y tick num=1,
            ytick={0, 20, 40, 60, 80, 100},
            legend columns=2,
            legend style={
                at={(0.95,0.05)},anchor=south east,align=center,
                legend style={nodes={scale=0.75, transform shape}},
            },
            label style={font=\normallabelsize},
            every tick label/.append style={font=\normallabelsize},
        ]
            \addplot[ecdfpglnoavlstatic] table [x=save_s,y=ecdf] {\tableexpivecdfmsciedawpglnoavlstatic};
            \addplot[ecdfpgnostatic] table [x=save_s,y=ecdf] {\tableexpivecdfmsciedawpgnostatic};
            \addplot[ecdfpgnoavl] table [x=save_s,y=ecdf] {\tableexpivecdfmsciedawpgnoavl};
            \addplot[ecdfpga] table [x=save_s,y=ecdf] {\tableexpivecdfmsciedawpga};


        \end{axis}
        \end{tikzpicture}
        \vspace{-2mm}
        \caption{\msciedaw}
        \label{fig:exp_iv_ecdf_compare_msciedaw}
    \end{subfigure}
    \vspace{-3mm}
    \caption{With parallel saving, active variable locking (AVL) and allowlist-based static code checker (ASCC) unblock user's cell executions, improving over synchronous saving.}
    \label{fig:exp_iv_ecdf_compare}
    \vspace{\undercaptionspace}
\end{figure}
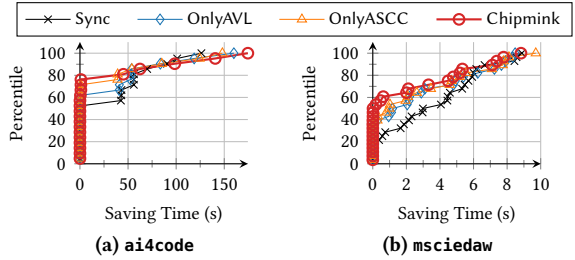

This experiment presents the perceived delay in detail by analyzing the \system's distribution of saving times against baselines'.
To study the effectiveness of asynchronous saving techniques, we implement baselines \pgnostatic and \pgnoavl as \system with and without active variable locking (AVL) and allowlist-based static code checker (ASCC) described in \cref{sec:async}. 
In addition for comparison, we include \pglnoavlstatic, \system without asynchronous saving as well.




\paragraph{Asynchronous saving locks and static executions} Despite the quick succession of cell executions in our benchmark, asynchronous saving via active variable locking (AVL) and allowlist-based static code checker (ASCC) improves \system's distribution of saving times over synchronous saving \pglnoavlstatic (\Cref{fig:exp_iv_ecdf_compare}).

\subsection{Access Control Overhead}
\label{sec:exp-overhead}

\def\subfigheight{30mm}

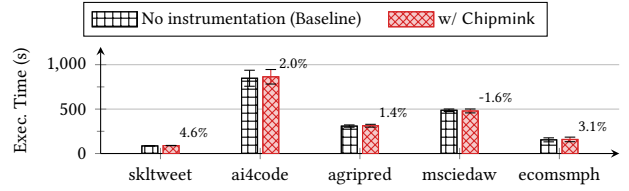
\begin{figure}[t]
    \centering
    \begin{subfigure}[b]{\linewidth} \centering
        \begin{tikzpicture}
        \begin{axis}[
            ybar,
            ticks=none,
            height=20mm,
            width=\linewidth,
            hide axis,
            xmin=10,  
            xmax=50,
            ymin=0,
            ymax=0.4,
            area legend,
            legend columns=-1,
            legend style={at={(0.0,0.0)},anchor=south,align=center,/tikz/every even column/.append style={column sep=3mm},nodes={scale=0.75, transform shape}},
        ]
            \node[align=center, opacity=1] {
                \addlegendimage{barzero}
                \addlegendentry{No instrumentation (Baseline)}
                \addlegendimage{barpga}
                \addlegendentry{w/ \pga}
            };
        \end{axis}
        \end{tikzpicture}
    \end{subfigure}
    \begin{tikzpicture}
    \begin{axis}[
        height=\subfigheight,
        width=1.0\linewidth,
        ybar,
        bar width=0.16,
        axis lines=left,
        ymajorgrids,
        ylabel=Exec. Time (s),
        ymin=0.0,
        ymax=1200,
        log origin=infty,
        enlarge x limits=0.15,
        xtick={data},
        xticklabels from table={\tableexpc}{nb},
        minor y tick num=1,
        nodes near coords,
        point meta=explicit symbolic,
        nodes near coords style={font=\scriptsize, align=left, rotate=0, anchor=south west},
        legend columns=2,
        legend style={at={(0.03,1.0)},anchor=north west},
        label style={font=\normallabelsize},
        every tick label/.append style={font=\normallabelsize},
    ]
        \addplot[barzero, error bars/.cd, y dir=both, y explicit] table [x=index,y=noop,y error=noopstdtwo] {\tableexpc};
        \addplot[barpga, error bars/.cd, y dir=both, y explicit] table [x=index,y=pgnoop,meta=overhead,y error=pgnoopstdtwo] {\tableexpc};
    \end{axis}
    \end{tikzpicture}
    \vspace{-6mm}
    \caption{\pga's access control overhead. Error bars show the two standard deviations of time measurements. Percentages refer to relative overheads over the baseline.}
    \label{fig:exp_c_overhead}
    \vspace{\undercaptionspace}
\end{figure}

Finally, we verify that \system's instrumentation (i.e., active variable tracker and locks) does not incur significant overhead, even if no saving is done.
Given a notebook, we execute all cells and measure the execution times with and without \system. The execution time differences signify how much execution delay the user would perceive if they install \system but decide not to save any state.
According to this benchmark, the instrumentation overhead is so small that there was no statistically significant runtime difference in four of five notebooks, i.e., $p \geq 0.05$ on Welch’s t-test. Only for \skltweet, we found a significant 4.6\% overhead with $p = 1.7 \times 10^{-10}$ (\Cref{fig:exp_c_overhead}).

\subsection{Validity Threats}
\label{sec:exp-validity}

We identify potential threats to the validity of the experiments below along with mitigation strategies and implications.

\paragraph{Size and time limitations} Because our testbed is limited to 800~GB of free disk space and we limit the total experiment runtime to the order of a few weeks, our experiment cannot observe the behaviors beyond these limits.
Therefore, we restrict the storage usage to 768~GB (reported as is, e.g., \ecomsmph on \shev) and the runtime to 1 day (reported as ``Failed'', e.g., \aicode on \pnv).

\paragraph{Workload simulation} While our experiments test all systems on real notebooks, we simulate the workload by executing cells back-to-back (i.e., akin to ``Run All'') and loading states back-to-back as well, which may not represent all notebook workflows. In reality, users may perform other tasks in between cell executions such as inspecting the outputs, editing the notebook, and/or thinking. A more realistic workload may improve system performance across all SUTs.
Additionally, variable loading may occur sparsely, reducing cache efficiency. To minimize caching effects, we remove the correlation between loads by shuffling the target IDs and clearing in-memory data with garbage collection before each load.

\paragraph{Hardware and software dependency} Absolute measurements depend on this specific hardware and software configuration, including but not limited to, CPU, memory, disk, operating system, and Python library versions. To help reliably reproduce the key observations on other machines, we control these environment variables by containerizing the experiment within Docker.

\end{document}